%% file: header.tex
\documentclass[conference]{IEEEtran}
\usepackage{microtype}
\usepackage{times,graphicx,color,eso-pic}
\usepackage{upgreek}
\usepackage{amsmath, amsthm, amssymb, stmaryrd, utf8math}

\usepackage{ textcomp }
\usepackage{ifthen}
\usepackage{suffix}
\usepackage{xspace}
\usepackage{algorithmicx}
\usepackage{paralist}
\usepackage{listings}% http://ctan.org/pkg/listings
\usepackage{xcolor}
\usepackage{cite}
\definecolor{codegreen}{rgb}{0,0.6,0}
\definecolor{codegray}{rgb}{0.5,0.5,0.5}
\definecolor{codepurple}{rgb}{0.58,0,0.82}
\definecolor{backcolour}{rgb}{1.0,1.0,1.0}
\lstdefinestyle{mystyle}{
    backgroundcolor=\color{backcolour},   
    commentstyle=\color{codegreen},
    keywordstyle=\color{magenta},
    numberstyle=\scriptsize\color{codegreen},
    stringstyle=\color{codepurple},
    basicstyle=\ttfamily\small,
    breakatwhitespace=false,         
    breaklines=true,                 
    captionpos=b,                    
    keepspaces=true,                 
    numbers=left,                    
    numbersep=5pt,                  
    showspaces=false,                
    showstringspaces=false,
    showtabs=false,                  
    tabsize=2
}

\usepackage{enumitem, afterpage}
\usepackage{mathrsfs}
\newtheorem{theorem}{Theorem}

\newtheorem{lemma}{Lemma}
\newtheorem{definition}{Definition}

\newtheorem{example}{Example}
\usepackage{mathtools}
\usepackage[noend]{algpseudocode}
\usepackage{hyperref}
\usepackage{float}
\usepackage{utf8math}
\usepackage{ttquot,jedmacros,macros}
\usepackage{txtt,cases}
\usepackage{stmaryrd}
\usepackage[font=small]{caption}
\usepackage{subcaption}
\usepackage{mathpartir}
\usepackage{tikz}

\usepackage{thmtools}
\usepackage{thm-restate}
\hypersetup{
	colorlinks,
	urlcolor=black,
	citecolor=black,
	filecolor=black,
	linkcolor=black
}
\usepackage{array}
\usepackage{arydshln}
\usepackage{tabu}
\usepackage{multirow}
\usepackage{graphicx}

\usepackage{listings}
\renewcommand{\floatpagefraction}{0.75}

\begin{document}
\title{Applying consensus and replication securely with FLAQR}
\author{
\IEEEauthorblockN{Priyanka Mondal}
\IEEEauthorblockA{\textit{University of California, Santa Cruz} \\
pmondal@ucsc.edu}
\and
\IEEEauthorblockN{Maximilian Algehed}
\IEEEauthorblockA{\textit{Chalmers University of Technology} \\
 algehed@chalmers.se}
\and
\IEEEauthorblockN{Owen Arden}
\IEEEauthorblockA{\textit{University of California, Santa Cruz} \\
oarden@ucsc.edu }
}
\maketitle
\bstctlcite{MyBSTcontrol}

\begin{abstract}
Availability is crucial to the security of distributed systems, but
guaranteeing availability is hard, especially when participants in the
system may act maliciously.  Quorum replication protocols provide
both integrity and availability: data and
computation is replicated at multiple independent hosts, and a quorum
of these hosts must agree on the output of all operations applied to
the data.
Unfortunately, these protocols have high overhead and can be difficult to
calibrate for a specific application's needs.
Ideally, developers could use high-level abstractions for
consensus and replication to write fault-tolerant code by
that is secure by construction.

This paper presents Flow-Limited Authorization for Quorum Replication
(FLAQR), a core calculus for building distributed applications with
heterogeneous quorum replication protocols while enforcing end-to-end
information security. Our type system ensures that well-typed FLAQR
programs cannot _fail_ (experience an unrecoverable error) in ways
that violate their type-level specifications.  We present
noninterference theorems that characterize FLAQR's confidentiality,
integrity, and availability in the presence of consensus, replication,
and failures, as well as a liveness theorem for the class of majority
quorum protocols under a bounded number of faults.
\end{abstract}

\input{body}

\end{document}

%% file: body.tex
\section{Introduction} \label{sec:Intro}

Failure is inevitable in distributed systems, but its consequences
may vary.  The consequences of failure are particularly
severe in centralized system designs, where single points-of-failure
can render the entire system inoperable.  Even distributed systems are
sometimes built using a single, centralized authority to execute
security-critical tasks.  If this trusted entity is compromised, the
security of the entire system may be compromised as well.

Building reliable _decentralized systems_, which have no single
point-of-failure, is a complex task. Quorum replication protocols such
as Paxos~\cite{paxos} and PBFT~\cite{pbft}, and blockchains such as
Bitcoin~\cite{bitcoin} replicate state
and computation at independent nodes and use consensus protocols to
ensure the integrity and availability of operations on system state.
In these protocols, there is neither centralization of function nor
centralization of trust: all honest nodes work to replicate the same
computation on the same data, and this redundancy helps the system
tolerate a bounded number of node failures and corruptions.

Within a single trust domain such as a corporate data center, replicas
likely have uniform trust relationships and may be treated
interchangeably.  However, many
large-scale systems depend on services hosted by multiple external
services. Even when a system's internal components are replicated, 
developers must take into account the failure properties of external
dependencies when considering their own robustness. 

Information flow control (IFC) has been used to enforce decentralized
security in distributed systems for confidentiality and integrity
(e.g., Fabric~\cite{jfabric} and DStar~\cite{dstar}).  Less
attention has been paid to enforcing decentralized availability
policies with IFC. In particular, no language (or protocol) we are
aware of addresses systems that compose multiple quorums or consider
quorum participants with arbitrary trust relationships.

To build a formal foundation for such languages, we
present FLAQR, a core calculus for Flow-Limited
Authorization~\cite{flam} for Quorum Replication.
FLAQR uses
high-level abstractions for replication and consensus that help
manage tradeoffs between the availability and integrity of
computation and data.

\begin{figure*}
  \centering
  \begin{subfigure}[b]{0.23\textwidth}
  \centering
  \includegraphics[scale=0.28]{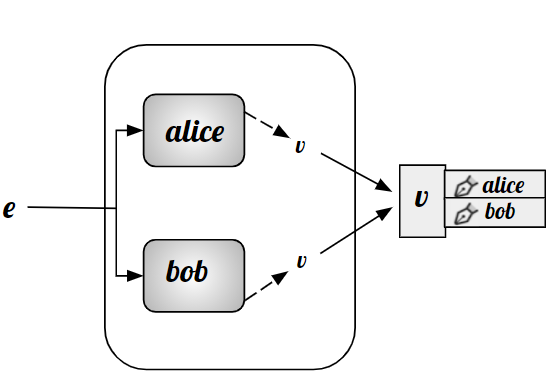}
  \caption{More Integrity}
  \label{fig:compare}
  \end{subfigure}
  \begin{subfigure}[b]{0.23\textwidth}
  \centering
  \includegraphics[scale=0.28]{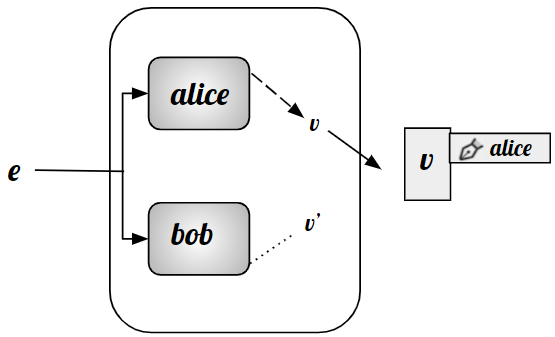}
  \caption{More Availability}
  \label{fig:select}
  \end{subfigure}
  \begin{subfigure}[b]{0.23\textwidth}
  \centering
  \includegraphics[scale=0.48]{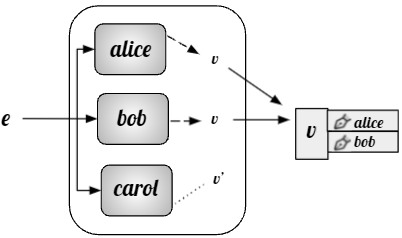}
  \caption{More Integrity and Availability}
  \label{fig:majority}
  \end{subfigure}
  \begin{subfigure}[b]{0.23\textwidth}
  \centering
  \includegraphics[scale=.28]{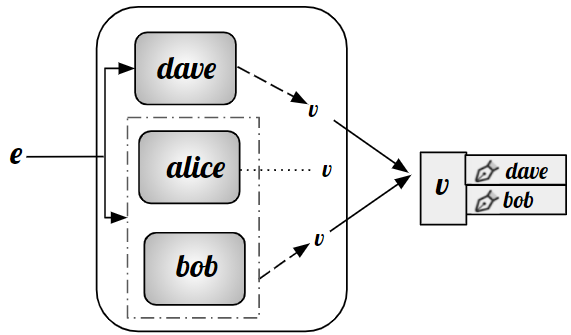}
  \caption{Heterogeneous trust}
  \label{fig:hetero}
  \end{subfigure}
  \caption{Integrity-Availability Trade-off}
  \label{fig:tradeoff}
\end{figure*}

Consider the scenarios in
Figure~\ref{fig:tradeoff}.  Shaded boxes represent hosts in a distributed
system.  Dashed lines denote outputs that contribute to the final
result, a value $v$.  Dotted lines denote ignored outputs and solid
lines indicate the flow of data from an initial expression $e$
distributed to hosts to the collected result.  Results are accompanied
by labels that indicate which hosts influenced the final result.

In Figure~\ref{fig:compare}, $e$ is distributed to hosts \texttt{alice}
and \texttt{bob}. The hosts' results are compared and, if they match,
the result is produced.  Since a value is output only
if the values match, we can treat the output of this protocol as
having \emph{more integrity} than just \texttt{alice} or 
\texttt{bob}.
While both \texttt{alice} and
\texttt{bob} technically influence the output, 
neither host can unilaterally control its value.
However, either host can cause the protocol to fail.   

By contrast, the protocol in Figure~\ref{fig:select} prioritizes
availability over integrity: if either \texttt{alice} or \texttt{bob}
produce a value, the protocol outputs a value—in this case
\texttt{alice}'s.  Here, neither host can unilaterally cause a
failure; the protocol only fails if both \texttt{alice} and
\texttt{bob} fail.  Either \texttt{alice} or \texttt{bob} (but not
both) has complete control over the result in the event of the other's
failure, so we should treat the output as having \emph{less integrity}
than just \texttt{alice} or \texttt{bob}.

With an adequate number of hosts, we can combine these two techniques
to form the essential components of a quorum system.  In Figure~\ref{fig:majority}, $e$ is
replicated to \texttt{alice}, \texttt{bob}, and \texttt{carol}.  This protocol outputs a
value if any two hosts have matching outputs.  Since \texttt{alice} and \texttt{bob}
both output $v$, the protocol outputs $v$ and attaches \texttt{alice} and
\texttt{bob}'s signatures.  The non-matching value $v'$ from \texttt{carol} is
ignored.  Hence, this protocol prevents any single host from
unilaterally controlling the failure of the protocol or its output.

Figure~\ref{fig:majority} is similar in spirit to consensus protocols
such as Paxos or PBFT where quorums of independent replicas are used to
tolerate a bounded number of failures.  FLAQR also permits us to write
protocols where principals have differing trust relationships.
Figure~\ref{fig:hetero} illustrates a protocol that tolerates failure
(or corruption) of either \texttt{alice} or \texttt{bob}, but requires \texttt{dave}'s
output to be part of any quorum.  This protocol will fail if both
\texttt{alice} and \texttt{bob} fail to produce matching outputs, but can also fail
if \texttt{dave} fails to produce a matching output. This example illustrates
the distributed systems where the hosts do not have homogeneous trust.

The main contributions of this paper are as follows:
\begin{itemize}
\item An extension of the static fragment of the Flow Limited Authorization 
  Model (FLAM) \cite{flam} with availability policies 
  and algebraic operators representing the effective authority of
  consensus and replication protocols (§\ref{sec:FLAMalgebra}-§\ref{sec:types}).

\item A formalization of the FLAQR language (§\ref{sec:FLAQRprimitives}) and accompanying results:
\begin{itemize}
    \item A liveness theorem for majority-quorum FLAQR protocols (§\ref{blameproofs}) which
      experience a bounded number of faults using a novel proof technique:
      a _blame semantics_ that associates failing executions of a FLAQR
      program with a set of principals who may have caused the failure.
      \item Noninterference theorems for confidentiality, integrity, and availability (§\ref{sec:secPropNI}). 
\end{itemize}
\end{itemize}

\paragraph*{Non-goals} The design of FLAQR is motivated by application-agnostic
consensus protocols such as Paxos~\cite{paxos} and PBFT~\cite{pbft}, but our present goal is not to develop a framework for
_verifying_ implementations of such protocols (although it would be interesting future work).
Rather, the goal is to develop
security abstractions that make it easier to create 
components with application-specific integrity and availability guarantees,
and compose them in a secure and principled way.

In particular, the
FLAQR system model lacks some features that a protocol verification model
would require, most notably a concurrent semantics.
Although this simplifies some aspects of consensus protocols, our
model retains many of the core challenges present in fault tolerance
models. For example, perfect fault detection is impossible and faulty
nodes can manipulate data to cause failures to manifest at other
hosts.

\section{Motivating Examples}\label{sec:section2.0}

In this section we present two motivating examples.  The first example
highlights the trade-off between integrity and availability.
The second example highlights the need for availability policies in distributed systems.

\subsection{Tolerating failure and corruption}
\label{sec:section2.1}
If a bank's deposit records are stored in a single node,
then customers will be unable to access their accounts if that
node is unavailable or is compromised.  To eliminate this single
point-of-failure, banks can replicate their records on multiple hosts
as illustrated in Figure~\ref{fig:majority}.  If a majority of nodes
agree on an account balance, then the system can tolerate 
the remaining minority
of nodes failing or returning corrupted results.

Consider a quorum system with three nodes: \texttt{alice},
\texttt{bob}, and \texttt{carol}.  To tolerate the failure of a single
node, balance queries attempt to contact all three nodes and compare
the responses.  As long as the client receives two responses with the
same balance, the client can be confident the balance is correct even
if one node is compromised or has failed.

\begin{figure}
\begin{lstlisting} 
getBalance($acct$):
    $bal_a$ = fetch $bal(acct)$ $@$ $alice$;$\label{fetchb}$
    $bal_b$ = fetch $bal(acct)$ $@$ $bob$;
    $bal_c$ = fetch $bal(acct)$ $@$ $carol$;$\label{fetche}$ 
    
    $bal_{ab}$ = ($bal_a$==$bal_b$) ? $bal_a$ : $\fail{}$;$\label{compb}$ 
    $bal_{bc}$ = ($bal_b$==$bal_c$) ? $bal_b$ : $\fail{}$;$\label{compbb}$
    $bal_{ca}$ = ($bal_c$==$bal_a$) ? $bal_c$ : $\fail{}$; $\label{compe}$  
    
    if $bal_{ab}$ != $\fail{}$ then $\label{retb}$  
       return $bal_{ab}$;
    else if $bal_{bc}$ != $\fail{}$ then
       return $bal_{bc}$;
    else if $bal_{ca}$ != $\fail{}$ then
       return $bal_{ca}$;
    else return $\fail{}$;$\label{rete}$   
\end{lstlisting}
\caption{Majority quorum}
\label{fig:majex}
\end{figure}

Figure~\ref{fig:majex} illustrates pseudocode for 
the \texttt{getBalance} function for this system.
The code fetches balances from the three nodes (lines
\ref{fetchb}-\ref{fetche}) and compares them (lines
\ref{compb}-\ref{compe}). If two balances match, then the balance
is returned, otherwise the function returns $\fail{}$ (lines
\ref{retb}-\ref{rete}).

The downside of this approach is that it is quite verbose and 
has repetitive checks. 
Small mistakes in any of these lines could have significant
consequences. For example, suppose a programmer wrote 
$bal_b$ instead of $bal_c$ on line \ref{compbb}. This small change
gives $\texttt{bob}$ (or an attacker in control of \texttt{bob}'s node)
the ability to unilaterally choose the return
value of the function, even when $\texttt{alice}$ and $\texttt{carol}$
agree on a different value. 

\subsection{Using best available services}\label{sec:section2.2}
Real world applications often consist of 
communication between entities with mutual distrust. 
The pseudocode in Figure \ref{fig:example2} 
communicates with two banks, represented by 
$b$ and $b'$,  
during a distributed computation.
A user has two accounts
${{acc}_1}$, and ${{acc}_2}$
with $b$ and $b'$ 
respectively.
The user links both 
accounts to an online service and 
\begin{enumerate}
\item wants to use the account with the 
        highest balance to pay the bill 
\item does not want to miss a payment
        if the account with maximum 
        balance is unavailable
\end{enumerate}

The pseudocode in Figure \ref{fig:example2} attempts to fetch the
balances of both accounts. Lines \ref{retbal1}-\ref{availe} selects
the fallback balance if one balance is unavailable. If both balances
are available, then the comparison on line $\ref{compbal12}$ selects
the largest balance.

\begin{figure}
\begin{lstlisting}
$bal_1$ =  fetch $bal(acc_1)$ $@$ $b$;
$bal_2$ =  fetch $bal(acc_2)$ $@$ $b'$;

if $bal_1$ == $\fail{\relax}$ && $bal_2$ == $\fail{\relax}$ then $\label{avails}$
   return $\fail{}$;
else if $bal_1$ == $\fail{\relax}$ then $\label{retbal1}$
   return $bal_2$;
else if $bal_2$ == $\fail{\relax}$ then
   return $bal_1$; $\label{availe}$

if $bal_1$ > $bal_2$ then $\label{compbal12}$
    return $bal_1$;
else
    return $bal_2$; $\label{retbal2}$
\end{lstlisting}
\caption{Available largest balance}
\label{fig:example2}
\end{figure}

This example demonstrates that the availability of data can affect a
program's output, making it an important security property in
distributed computations.  Unfortunately, ensuring the availability of
a distributed application can be a tedious and error-prone process.

\section{Specifying availability policies} \label{sec:FLAMalgebra}
FLAQR policies are specified using an extension of the
FLAM~\cite{flam,jflac} principal algebra that includes availability
policies.\footnote{Specifically, we extend the
  static fragment of FLAM's principal algebra defined by FLAC~\cite{jflac}.}
FLAM principals represent both the
_authority_ of entities in a system as well as bounds on the
_information flow policies_ that authority entails. For example,
Alice's authority is represented by the principal \texttt{alice}.
_Authority projections_ allow us to refer to specific categories of
Alice's authority.  The principal $\texttt{alice}^{\confid}$
 refers to
Alice's confidentiality authority: what Alice can read. Principal
$\texttt{alice}^{\integ}$ refers to Alice's integrity authority: what
Alice can write or influence.\footnote{Prior FLAM-based formalizations have used $\rightarrow$ and $\leftarrow$ for confidentiality and integrity, respectively.} Principal $\texttt{alice}^{\avail}$
refers to her availability authority: what Alice can cause to _fail_.
We refer to the set of all _primitive principals_ such as \texttt{alice} and \texttt{bob}
as $\N$. 

We can write conjunction of two principals
as $\texttt{alice} \wedge \texttt{bob}$, 
the combined authority of Alice and Bob. Put another way, $\texttt{alice} \wedge \texttt{bob}$
is a principal both Alice and Bob \emph{trust}.
The disjunction of two
principals' authority is written $\texttt{alice} ∨ \texttt{bob}$.  
This is a principal whose authority is less than both Alice and Bob; either Alice or Bob
can act on behalf of the principal $\texttt{alice} ∨ \texttt{bob}$. 

The confidentiality, integrity, and availability authorities
make up the totality of a principal's authority, so writing
$\texttt{alice}^{\confid} ∧ \texttt{alice}^{\integ}  ∧ \texttt{alice}^{\avail}$
is equivalent to writing $\texttt{alice}$. For brevity, we sometimes write
$\texttt{alice}^{\confid \integ}$ as a shorthand for $\texttt{alice}^{\confid} ∧
\texttt{alice}^{\integ}$ when we wish to include all but one kind of authority.

In addition to conjunctions and disjunctions of authority, FLAQR also
introduce two new operators: _partial conjunction_ ($\comand{}{}$),
and _partial disjunction_ ($\selor{}{}$).  These operations are
necessary to represent the tradeoffs between integrity and
availability mediated by consensus and replication.  Consider the
``more integrity'' protocol from Figure~\ref{fig:compare}.  It is
reasonable to think of the consensus value $v$ as having more
integrity than (or at least, ``not less integrity than'') Alice or Bob
alone, but it is important that we distinguish between this authority
and the combined integrity authority of Alice and Bob -
$(\texttt{alice} ∧ \texttt{bob})^{\integ}$.  A principal with
integrity authority $(\texttt{alice} ∧ \texttt{bob})^{\integ}$ may act
arbitrarily on behalf of both Alice and Bob since it is trusted by them.
In contrast, the integrity authority of the protocol in Figure~\ref{fig:compare}
is _not_ fully trusted by Alice and Bob. Instead, Alice and Bob only trust
the protocol when Alice and Bob agree on the value $v$.  If they do not agree,
then the protocol has failed and no value should be produced.  For this reason,
we describe the integrity of consensus values such as $v$ as the _partial conjunction_
of Alice and Bob, written ${(\comand{\texttt{alice}}{\texttt{bob}})}^{\integ}$.

Similarly, for replication protocols like that in
Figure~\ref{fig:select}, we want to distinguish the integrity of
values that may have been received from either Alice or Bob due to
failure (or arbitrarily choosing one over the other), from the
integrity of values that may have been influenced by both Alice and
Bob: ${(\texttt{alice} ∨ \texttt{bob})}^{\integ}$.  The integrity of value
replicated by Alice and Bob, written as the _partial disjunction_
${(\selor{\texttt{alice}}{\texttt{bob}})}^{\integ}$, does not have more integrity
than Alice or Bob alone since we cannot guarantee which host's value
will be used in the event of a failure.  The replicated value does,
however, have more integrity than ${(\texttt{alice} ∨ \texttt{bob})}^{\integ}$,
since this policy permits both principals to influence associated
values.

We can compare the authority of principals using the _acts-for_ relation
$≽$, which partially orders principals by increasing authority.
We form the set of all principals $\P$ as the closure.
We
say Alice _acts for_ Bob (or equivalently, Bob trusts Alice) and write
$\texttt{alice} ≽ \texttt{bob}$ when Alice has at least as much
authority as Bob.  The $≽$ relation forms a lattice with join
$\wedge$, meet $\vee$, greatest element $\top$, and least element
$\bot$.

\begin{figure*}
  \begin{mathpar}
\Rule{PAndL}
{\rafjudge{\Pi}{p_i}{p} 
% \rafjudge{\Pi}{p_2}{p}
}
{\rafjudge{\Pi}{\comor{p_1}{p_2}}{p}}

\Rule{PAndR}
{\rafjudge{\Pi}{p}{p_1}\\\\
 \rafjudge{\Pi}{p}{p_2}}
{\rafjudge{\Pi}{p}{\comor{p_1}{p_2}}}

\Rule{AndPAnd}
{}
{\rafjudge{\Pi}{p \wedge q}{\comor{p}{q}}}

%\Rule{PartandActsDisj}
%{}
%{\rafjudge{\Pi}{\comor{p}{q}}{p \vee q}}

\Rule{PAndPOr}
{}
{\rafjudge{\Pi}{\comor{p}{q}}{\selor{p}{q}}}

\Rule{ProjPAndL}
{}
{\rafjudge{\Pi}{\comor{p^{\pi}}{q^{\pi}}}{(\comor{p}{q})^{\pi}}}

\Rule{ProjPAndR}
{}
{\rafjudge{\Pi}{(\comor{p}{q})^{\pi}}{\comor{p^{\pi}}{q^{\pi}}}}

\Rule{ProjPOrL}
{}
{\rafjudge{\Pi}{\selor{p^{\pi}}{q^{\pi}}}{(\selor{p}{q})^{\pi}}}

\Rule{ProjPOrR}
{}
{\rafjudge{\Pi}{(\selor{p}{q})^{\pi}}{\selor{p^{\pi}}{q^{\pi}}}}
\hfill

\Rule{POrOr}
{}
{\rafjudge{\Pi}{\selor{p}{q}}{p \vee q}}

\end{mathpar}
\caption{Selected acts-for rules for partial conjunction and disjunction.}
\label{fig:partialAF}
\end{figure*}

\begin{figure}
\centering\includegraphics[height=0.4\textwidth,width=0.4\textwidth]{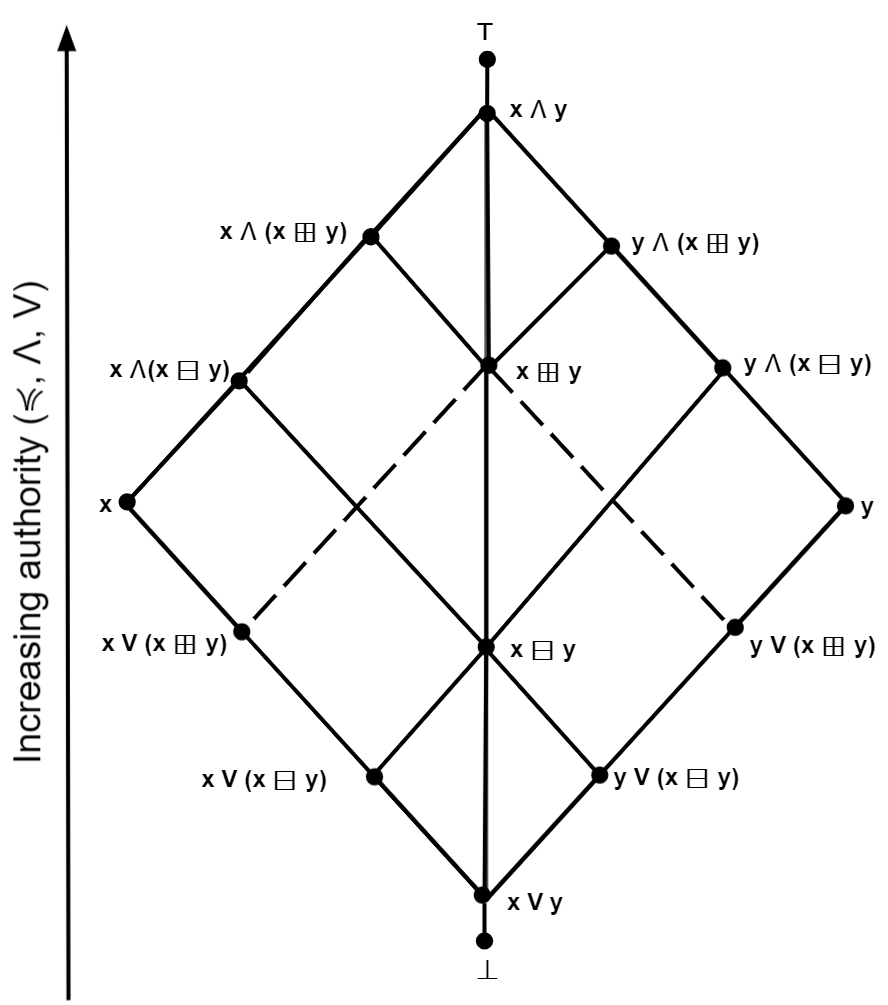}
\caption{The FLAQR authority lattice for the principal set $\{\bot, x, y, \top \}$.}
\label{fig:flaqrlat}
\end{figure}

We extend the acts-for relation defined by Arden et al.~\cite{jflac}
with new rules for availability authority and partial conjunction and
disjunction.  Figure~\ref{fig:partialAF} presents a selection of these
rules—we have omitted the distributivity rules for brevity. The
complete rule set is presented in Figure~\ref{fig:partialactsforfull}
in the appendices.

As a consequence of these new acts-for rules we
have additional points on the authority lattice.
Figure \ref{fig:flaqrlat} illustrates 
the authority sublattice over 
elements $\{\bot, x, y, \top\}$, with $\wedge$ 
as join and $\vee$ as meet. 
Figure \ref{fig:flaqrlat} shows
the trust ordering of all 
possible combination of elements
that can be formed on the set 
$\{\bot, x, y, \top\}$ with
operations $\wedge$, $\vee$,
$\comor{}{}$ and $\selor{}{}$ over them.
The relationship between principals 
$\bot, x, y, x \wedge y , x \vee y,$ and $\top$
is the same as in FLAM, but Figure~\ref{fig:flaqrlat}
also includes principals constructed using partial conjunctions
and disjunctions.
For example, $x \wedge (\comor{x}{y})$ is the least upper bound of
$x \wedge (\selor{x}{y})$ and $\comor{x}{y}$.
This is the case because of the newly introduced 
rule \ruleref{PAndPOr} in Figure~\ref{fig:partialAF},
which lets us simplify $x \wedge (\selor{x}{y}) ∧ \comand{x}{y}$
to $x \wedge (\comor{x}{y})$.

To compare the restrictiveness of information flow policies, we use
the _flows-to_ relation $⊑$, which partially orders principals by
increasing policy restrictiveness, rather than by authority. For
example, we say Alice's integrity flows to Bob's integrity and write
$\texttt{alice}^{\integ} ⊑ \texttt{bob}^{\integ}$ if Bob trusts
information influenced by Alice at least as much as information he
influenced himself.  Likewise, we write $\texttt{alice}^{\confid} ⊑
\texttt{bob}^{\confid}$ if Alice trusts Bob to protect the
confidentiality of her information, and 
$\texttt{alice}^{\avail} ⊑ \texttt{bob}^{\avail}$ if
Bob is trusted to keep Alice's data available. 
The flows-to relation behaves
similarly to a sub-typing relation.  Treating information 
labeled $\texttt{alice}^{\confid\integ\avail}$ 
(i.e. $\texttt{alice}$) as though it was labeled
$\texttt{bob}^{\confid\integ\avail}$ (i.e. $\texttt{bob}$)  
is only safe (doesn't violate anyone's
policies) if $\texttt{alice}^{\confid\integ\avail} ⊑
\texttt{bob}^{\confid\integ\avail}$ 
(i.e. $\texttt{alice} ⊑ \texttt{bob}$). 

An advantage of using FLAM principals is that we can define the 
flows-to relation in terms
of the acts-for relation, simplifying our formalism.
$$ p ~ \flowsto ~q\ \text{ if and only if }
q^{\confid} ≽ p^{\confid}
\text{ and }\ p^{\integ} ≽ q^{\integ}
\text{ and } p^{\avail} ≽ q^{\avail}$$
Based on this, a flows-to relation exists between 
any two points in Figure \ref{fig:flaqrlat} and a flow from $p$ to $q$ is secure only when $q^{\confid}$ is at least as
confidential as $p^{\confid}$, $q^{\integ}$ trusts information influenced by $p^{\integ}$, and $q^{\avail}$
cannot cause failures that $p^{\avail}$ cannot. 
The flows-to relation also forms a (distinct) lattice with joins $⊔$ and meet $⊓$ defined
in terms of their authority lattice counterparts.
\begin{align*}
  p \join q &\triangleq (p^{\confid} \wedge q^{\confid}) \wedge (p^{\integ} \vee q^{\integ}) \wedge (p^{\avail} \vee q^{\avail})\\
  p \meet q &\triangleq (p^{\confid} \vee q^{\confid}) \wedge (p^{\integ} \wedge q^{\integ}) \wedge (p^{\avail} \wedge q^{\avail})
\end{align*}

%The formalization of our principal algebra is based on the static
%fragment of FLAM used by FLAC~\cite{flac,jflac}. The complete ruleset is
%included in Appendix~\ref{fig:static-af}\owen{TODO!}.  
%We have also partially adapted
%the FLAM Coq formalization~\cite{flamtr} to include availability policies.
%\owen{should we make an anonymized link to the flaqr Coq version?}

\section{Syntax and security abstractions}\label{sec:FLAQRprimitives}
\begin{figure}
  {\small
  \[
    \begin{array}{rcl}
      \multicolumn{3}{l}{ π ∈ {\{"c","i","a"\}}  \text{  (projections)}} \\
      \multicolumn{3}{l}{ n ∈ \N \text{  (primitive principals)}} \\
      \multicolumn{3}{l}{ x ∈ \mathcal{V} \text{  (variable names)}} \\
      \\
     \multicolumn{3}{l}{
      p,ℓ,\pc \quad ::=\quad   n \sep \top \sep \bot \sep p^{π} \sep p ∧ p\sep p ∨ p}\\[0.4em]
          &\sep&  p ⊔ p \sep p ⊓ p \sep \selor{p}{p} \sep \comor{p}{p}\\[0.4em] 
      τ &::=& \voidtype \sep X \sep \sumtype{τ}{τ} \sep \prodtype{τ}{τ} \\[0.4em]
         & \sep & \func{τ}{\pc}{τ} \sep \tfunc{X}{\pc}{τ} \sep \says{ℓ}{τ}  \\[0.4em]
      v &::=& \void \sep \inji{v} \sep \pair{v}{v} \sep 
              \underline{\returnv{\ell}{v}} \\[0.4em]
        & \sep & \lamc{x}{τ}{\pc}{e} \sep \tlam{X}{\pc}{e}  \\[0.8em]
      f &::=& v \sep  \underline{\faila{\tau}} \\[0.8em]
      e &::=& f \sep x \sep e~e \sep e~τ \sep \return{\ell}{e} \sep 
             \pair{e}{e} \sep \proji{e} \sep \inji{e} \\[0.4em]
        & \sep & \bind{x}{e}{e} \sep \casexp{e}{x}{e}{e} \\[0.4em]
         & \sep &  \runa{\tau}{e}{p} \sep \underline{\ret{e}{p}} \sep \underline{\expecta{\tau}} \\[0.4em]  
        & \sep & \select{e}{e} \sep \compare{e}{e} 
    \end{array}
  \]
  }
  \caption{FLAQR Syntax. Underlined terms are generated during evaluation and are not available at the source level.}
  \label{fig:syntax}
\end{figure}

\begin{figure}
  \begin{subfigure}{0.5\textwidth}
  {\small
  \begin{flushleft}
  %  \boxed{{e} \stepsone {e'}} \\
  \begin{mathpar}
    \erule{E-Sealed}{}{\return{ℓ}{v}}{\returnv{ℓ}{v}}

    \erule{E-BindM}{}{ {\bind{x}{\returnv{ℓ}{v}}{e}}}{ {\subst{e}{x}{v}}}

    \erule{E-Compare}{}{\compare{\returnv{ℓ₁}{v}}{\returnv{ℓ₂}{v}}}{\returnv{\txcmp{ℓ₁}{ℓ₂}}{v}} 

    \erule{E-CompareFail}{v_1 \neq v_2 \\ τ=\says{(\txcmp{\ell_1}{\ell_2})}{\tau'}} {\compare{\returnv{ℓ₁}{v_1}}{\returnv{ℓ₂}{v_2}}}
       {\faila{τ}} 

% \erule{E-CompareFail}{f_1 \neq f_2}{\compare{\says{\txcmp{ℓ₁}{ℓ₂}}{τ}}{f_1}{f_2}}{\fail{\says{\txcmp{ℓ₁}{ℓ₂}}{τ}}}

    \erule{E-CompareFailL*}{\tau₁=\says{\ell_1}{\tau} \\ τ'=\says{(\txcmp{\ell_1}{\ell_2})}{\tau} }{\compare{(\faila{τ₁})} {\returnv{\ell_2}{v}}}
      {\faila{τ'}}
  
%\erule{E-CompareFailR}{}{\compare{\returnv{\ell_1}{v}}{(\faila{\says{\ell_2}{\tau}})}}
%{\faila{\says{\txcmp{\ell_1}{\ell_2}}{\tau}}}
 
%\erule{E-Select}{ f_{i} = \returnv{ℓᵢ}{wᵢ}  \\  f_{j} = \fail }
%     {\select{f₁}{f₂}{\says{\txsel{ℓ₁}{ℓ₂}}{τ}}}{\returnv{\txsel{ℓ₁}{ℓ₂}}{wᵢ}} 

    \erule{E-Select}{}{\select{\returnv{\ell_1}{v_1}}
    {\returnv{\ell_2}{v_2}}}
    {\returnv{\txsel{\ell_1}{\ell_2}}{v_1}} 
    
    \erule{E-SelectL}{}{\select{\returnv{\ell_1}{v}}
    {(\faila{\says{\ell_1}{\tau}})}}
    {\returnv{\txsel{\ell_1}{\ell_2}}{v}} 
    %
    %\erule{E-SelectR}{}{{\select{(\faila{\says{\ell_1}{\tau}})}
    %{\returnv{\ell_2}{v}}}}
    %{\returnv{\txsel{\ell_1}{\ell_2}}{v}} 
    
    \erule{E-SelectFail}{τᵢ=\says{\ell_i}{\tau} \\ τ'=\says{(\txsel{\ell_1}{\ell_2})}{\tau}} {\select{(\faila{τ₁})}
    {(\faila{τ₂})}}
    {\faila{τ'}} 
    
    \erule{E-RetStep}{e \stepsone e'}{\ret{e}{c}}{\ret{e'}{c}}
    \erule{E-Step}{ {e} \stepsone  {e'}}{{E[e]}}{{E[e']}}
\hfill
\\
\end{mathpar}
\end{flushleft}
}
\end{subfigure} 

\begin{subfigure}{0.5\textwidth}
  \[
    \begin{array}{rcl}
      E & ::= & [\cdot] \sep E~e \sep v~E  \sep \return{ℓ}{E}  \sep \bind{x}{E}{e} \\[0.4em]
        & \sep & \ret{E}{p} \sep \select{E}{e} \sep \select{f}{E} \\[0.4em]
        &\sep  & \compare{E}{e} \sep \compare{f}{E} \\[0.4em]
    \end{array}
  \]
%\caption{operational semantics.}
\end{subfigure} 
\caption{FLAQR local semantics and evaluation context}
\label{fig:semantics}
\end{figure}

\begin{figure}
  \small
%\begin{flalign*}
%  & \boxed{{e} \stepsone {\faila{\tau}}} &
%\end{flalign*}
  \begin{mathpar}
\erule{E-AppFail}{}{\lamc{x}{\tau}{pc}{e}~{\faila{\tau}}}{\subst{e}{x}{\faila{\tau}}}

\erule{E-SealedFail}{}{{\returnp{ℓ }{\faila{\tau}}}}{{\faila{\says{\ell}{\tau}}}}

\erule{E-InjFail}{}{\injia{\sumtype{\tau_1}{\tau_2}}{\faila{\tau_i}}}{\faila{\sumtype{\tau_1}{\tau_2}}}

\erule{E-ProjFail}{}{\proji{\faila{\prodtype{\tau_1}{\tau_2}}}}{\faila{\tau_i}} 
\hfill
\end{mathpar}
\caption{Propagation of fail terms}
\label{fig:failprop}
\end{figure}

\begin{figure*}
%\begin{flalign*}
%  &\boxed{\distcon{e}{c}{s} \Longrightarrow  \distcon{e'}{c'}{s'}} &
%\end{flalign*}

\begin{mathpar}
\derule{E-DStep} {e \stepsone e'}{\distcon{E[e]}{c}{t}}{\distcon{E[e']}{c}{t}}

   \derule{E-Run}{}{\distcon{E[\runa{\tau}{e}{c'}]}{c}{t}}{\distcon{\ret{e}{c}}{c'}
{\stackapp{E[\expecta{\tau}]}{c}{t}}}

   \derule{E-RetV}{}
    {\distcon{\ret{v}{c}}{c'}{\stackapp{E[\expecta{\says{pc^{{\integ}{\avail}}}{\tau'}}]}{c}{t}}}
    {\distcon{E[\returnv{pc^{{\integ}{\avail}}}{v}]}{c}{t}}

   \derule{E-RetFail}{}
{\distcon{\ret{(\faila{\tau'})}{c}}{c'}
{\stackapp{E[\expecta{\says{pc^{{\integ}{\avail}}}{\tau'}}]}{c}{t}}}
{\distcon{E[\faila{\says{pc^{{\integ}{\avail}}}{\tau'}}]}{c}{t}}
\hfill  
\end{mathpar}
\caption{Global semantics}
\label{fig:globsem}
\end{figure*}

\begin{figure}
  \[
    \begin{array}{rcl}
    s & ::= & \distcon{e}{c}{t} \\[0.4em]
    t & ::=  & \emptystack \sep 
                  \stackapp{E[\expecta{\tau}]}{c}{t} \\[0.4em]
    \end{array}
  \]
\caption{Global configuration stack}
\label{fig:sstack}
\end{figure}
Figures~\ref{fig:syntax} and~\ref{fig:semantics} present the FLAQR syntax
and selected evaluation rules. For exposition purposes, we omit some term
annotations and standard lambda calculus rules in order to focus on
FLAQR's contributions, but the complete, annotated FLAQR syntax and semantics are
found in Figures~\ref{fig:Annotatedsyntax}
and~\ref{fig:annotatedlocsem} in Appendix~\ref{sec:fulllang}.

FLAQR is based on FLAC~\cite{flac,jflac}, a monadic calculus in the
style of Abadi's Polymorphic DCC~\cite{abadi06}.  In addition to
standard additions to System F~\cite{girard71,girard72,reynolds74}
such as pairs and tagged unions, an Abadi-style calculus supports
monadic operations on values in a monad indexed by a lattice of
security labels.  Such a value has a type of the form
$\says{\ell}{\tau}$, meaning that it is a value of type $τ$,
_protected_ at level $ℓ$, where $ℓ$ is an element of the security
lattice.

FLAQR builds on FLAC's expressive principal algebra and type system to
model distributed security policies for applications that use
replication and consensus to enforce integrity and availability
policies.  FLAC supports policy downgrades through dynamic
delegations of authority between principals, but for simplicity we
omit these features in FLAQR.  

The monadic unit or return term $\return{ℓ}{e}$ protects the value
that $e$ evaluates to at level $ℓ$
(\ruleref{E-Sealed}).\footnote{Polymorphic DCC does not define a term
similar to $\returnv{ℓ}{v}$ and thus does not have an rule equivalent to
\ruleref{E-Sealed}.  The presence of side-effects in FLAC and FLAQR
require protected runtime values $\returnv{ℓ}{v}$ to be typed
differently from the source terms that created them $(\return{ℓ}{v})$. \ref{sec:types}}
Protected values, $\returnv{ℓ}{v}$ cannot be operated on
directly. Instead, a bind expression must be used to bind the
protected value to a variable whose scope is limited to the body of
the bind term (\ruleref{E-BindM}). The body performs the desired
computation and ``returns'' the result to the monad, ensuring the
result is protected.

The primary novelty in the FLAQR calculus is the introduction of
"compare" and "select" terms for expressing consensus and replication
workflows.  We represent the consensus problem as a comparison of two
values with the same underlying type but distinct outer security
labels.  In other words, we want to 
check the equality of values produced by two different principals.
If the
values match, we can treat them as having the (partially) combined integrity of
the principals.  If not, then the principals failed to reach
consensus.

Rule \ruleref{E-Compare} defines the former case: two syntactically equal
values protected at different labels evaluate to a value that combines
 labels using the _compare action_ on labels $\txcmp{}{}$. Intuitively,
$\txcmp{ℓ₁}{ℓ₂}$ determines the increase in integrity and the corresponding
decrease in availability inherent in requiring a consensus. We define $\txcmp{}{}$ 
formally in Definition~\ref{def:compare}.
\begin{definition}[Compare action on principals]
  \label{def:compare}
\[
  \txcmp{ℓ₁}{ℓ₂} \triangleq (ℓ₁^{\confid} ∧ ℓ₂^{\confid}) ∧ (\comor{ℓ₁^{\integ}}{ℓ₂^{\integ}}) ∧ (\ell_1^{\avail} \vee \ell_2^{\avail})
\]
\end{definition}
We also lift this notation to \says{}{} types by defining
$$\txcmp{\says{ℓ₁}{τ}}{\says{ℓ₂}{τ}} \triangleq \says{(\txcmp{ℓ₁}{ℓ₂})}{τ}$$

As discussed in Section~\ref{sec:FLAMalgebra}, the integrity authority
of a "compare" protocol is not as trusted as the conjunction of $ℓ₁$
and $ℓ₂$'s integrity.  Instead, we represent the limited ``increase''
in integrity authority\footnote{Strictly speaking, $\comand{x}{y}$ is
not an increase in integrity over $x$ (or $y$); $\comand{x}{y}$ and
$x$ are incomparable.} using a partial conjunction in
Definition~\ref{def:compare}.  In contrast, the decrease in
availability is represented by a (full) conjunction $\ell_1^{\avail}
\vee \ell_2^{\avail}$ since either $ℓ₁$ or $ℓ₂$ could unilaterally
cause the compare expression to fail.

The decreased availability resulting from applying "compare" is more
apparent in rules~\ruleref{E-CompareFail}
and~\ruleref{E-CompareFailL*}.  In~\ruleref{E-CompareFail}, two
unequal values are compared, resulting in a failure.  Failure is
represented syntactically using a $\faila{τ}$ term.  We use a type
annotation $τ$ on many terms in our formal definitions so that our
semantics is well defined with respect to failure terms, but we omit
most of these annotations in Figure~\ref{fig:semantics}.  In practice, we
expect a FLAQR compiler would infer these annotations automatically.

A "compare" term may also result in failure if either subexpression
fails. Rule~\ruleref{E-CompareFailL*} (simplified for exposition),
defines how failure of an input propagates to the output.  In fact,
most FLAQR terms result in failure when a subexpression fails.
Figure~\ref{fig:failprop} presents selected failure propagation
rules (complete failure propagation rules are presented in Figure
\ref{fig:Fullfailprop}).  
Note that "fail" terms are treated similarly to values, but
are distinct from them.  For example, in \ruleref{E-AppFail}, applying a lambda
term to a fail term substitutes the failure as it would a value, but in
\ruleref{E-SealedFail} the failure is propagated beyond the monadic
unit term.  This latter behavior captures the idea that failures cannot be
hidden or isolated in the same way as secrets or untrusted data.

Failures are tolerated using replication. A "select" term will
evaluate to a value as long as at least one of its subexpressions
does not fail. For example, rule~\ruleref{E-SelectL} returns its
left subexpression when the right subexpression fails.
In contrast to "compare", applying "select" increases availability
since either subexpression can be used, but reduces integrity
since influencing only one of the subexpressions is 
potentially sufficient
to influence the result of evaluating "select".
The effect of a "select" statement on the labels of its 
sub-expressions is captured with the
_select action_ $\txsel{}{}$.
\begin{definition}[Select action on principals]
  \label{def:select}
\[
  \txsel{ℓ₁}{ℓ₂}  \triangleq (ℓ₁^{\confid} ∧ ℓ₂^{\confid}) ∧ (\selor{ℓ₁^{\integ}}{ℓ₂^{\integ}}) ∧ (ℓ₁^{\avail} ∧ ℓ₂^{\avail})
\]
\end{definition}
We define the select action on types similarly to compare:
$$\txsel{\says{ℓ₁}{τ}}{\says{ℓ₂}{τ}} =\says{(\txsel{ℓ₁}{ℓ₂})}{τ} $$

The end result of a _select_ statement, 
$\select{\returnv{\ell_1}{v}}{\returnv{\ell_2}{v}}$, 
will have integrity of either $\ell_1^{\integ}$ or $\ell_2^{\integ}$
since only one of the two possible values will be used.  We
use a partial disjunction to represent this integrity since the result
does not have the same integrity as $ℓ₁$ or $ℓ₂$, but does have
more integrity than $ℓ₁ ∨ ℓ₂$ since it is never the case that _both_
principals influence the output.
%In this sense, we use $\selor{}{}$ 
%in a similar way that Hunt and Sands use the disjoint union operator
%on information flow quantales\footnote{A quantale is a lattice with an
%  additional tensor operator}~\cite{quantale}.

%here glob
\section{Global semantics}
We capture the distributed nature of quorum replication by embedding
the local semantic rules within a global distributed semantics in Figure~\ref{fig:globsem}.  This
semantics uses a _configuration stack_ $s=\distcon{e}{c}{t}$ (Figure~\ref{fig:sstack})
to keep track of the currently executing expression $e$, the host on which it is
executing $c$, and the remainder of the stack $t$. We also make explicit use of the
evaluation contexts from Figure~\ref{fig:semantics} to identify the reducible subterms across stack
elements. 

The core operation for distributed computation is 
$\runa{τ}{e}{p}$ which runs the computation $e$
of type $τ$ on node $p$.
Local evaluation steps are captured in the global semantics
via rule \ruleref{E-DStep}.  This rule says that if $e$ steps to $e'$ 
locally, then $E[e]$ steps to $E[e']$ globally.

Rule \ruleref{E-Run} takes an expression $e$ at host $c$, pushes a new
configuration on the stack containing $e$ at host $c'$ and places an
"expect" term at $c$ as a place holder for the return value.  

Once the remote expression is fully evaluated, rule \ruleref{E-RetV} 
pops the top configuration off the stack and replaces the "expect" term
at $c$ with the protected value $\returnv{\pc^{{\integ}{\avail}}}{v}$.
Rule \ruleref{E-RetFail} serves the same purpose for "fail" terms, but is necessary since
"fail" terms are not considered values (see Figure~\ref{fig:syntax}).
The label $\pc^{{\integ}{\avail}}$ reflects both the integrity and availability context
of the caller ($c$) as well as the integrity and availability of the remote host ($c'$).
We discuss this aspect of remote execution in more detail in Section~\ref{sec:types}.

\section{FLAQR typing rules}\label{sec:types}
\begin{figure*}
%\begin{subfigure}
\begin{flushleft}
  %\rulefiguresize
  \boxed{\TValGpcw{e}{τ}} \\
  \begin{mathpar}
    %\Rule{Var}{\Gamma(x)=\tau  \\ \rafjudge{\Pi}{c}{pc}}{\TValGpcw{x}{τ}}
    \Rule{Unit}{\rafjudge{\Pi}{c}{pc}}{\TValGpcw{\void}{\voidtype}}
    
    \Rule{Fail}{ \rafjudge{\Pi}{c}{pc}}
     {\TValGpcw{\faila{\tau}}{\tau}}

     \Rule{Expect}{ \rafjudge{\Pi}{c}{pc}}
      {\TValGpcw{\expecta{\tau}}{\tau}}
    
    \Rule{Lam}{%
     \TVal{\Pi;Γ,x\ty τ_1;\pc';u}{e}{τ_2} \\ \rafjudge{\Pi}{c}{pc} \\\\
     %\rafjudge{\Pi}{c}{\UB{\func{τ_1}{\pc'}{τ_2}}} \\  
     u = \UB{\func{τ_1}{\pc'}{τ_2}} \\
     \rafjudge{Π}{c}{u}
    }{\TValGpcw{\lamc{x}{τ_1}{\pc'}{e}}{\func{τ_1}{\pc'}{τ_2}}}

    \Rule{App}{%
      \TValGpcw{e_1}{\func{τ'}{\pc'}{τ}}\\\\
      \TValGpcw{e_2}{τ'} \\ 
      \drflowjudge{\Pi}{\pc}{\pc'}\\\\
      \rafjudge{\Pi}{c}{pc}
    }{\TValGpcw{e_1~e_2}{τ}}

    \Rule{UnitM}{
      \TValGpcw{e}{τ} \\ 
      \drflowjudge{\Pi}{\pc}{ℓ} \\\\ %pc influences the signature ℓ 
      \rafjudge{\Pi}{c}{pc}  
%\\ \stflowjudge*{\jmath(\tau)}{ℓ}
    }{\TValGpcw{\return{ℓ}{e}}{\says{ℓ}{τ}}}   

     \Rule{Sealed}{
      \TValGpcw{v}{τ} \\
      \rafjudge{\Pi}{c}{pc}   %p{W}{ℓ}\\\\
      %\protjudge*{w^I}{ℓ^I} \\ %implicit due to p{c}{pc} and \stflowjudge*{\pc}{ℓ}
      %\protjudge*{ℓ^C}{w^C} \\\\
    }{\TValGpcw{\returnv{ℓ}{v}}{\says{ℓ}{τ}}} 

    \Rule{BindM}{%
      \TValGpcw{e'}{\says{ℓ}{τ'}} \\ 
      \drflowjudge{\Pi}{\ell \sqcup pc}{\tau} \\\\
      %\stflowjudge*{\jmath(\tau)} \\ 
      \TVal{\Pi;Γ,x\ty τ';ℓ\sqcup pc;\worker}{e}{τ} \\ 
      \rafjudge{Π}{c}{pc} 
    }{\TValGpcw{\bind{x}{e'}{e}}{τ}}

\Rule{Run}{
\TVal{\Pi;\Gamma;\pc';c'}{e}{τ'} \\
\drflowjudge{\Pi}{pc}{pc'} \\\\
\rafjudge{\Pi}{c}{pc} \\
\rafjudge{\Pi}{c}{\UB{\tau'}}\\\\
\tau = \says{pc'^{\integ \avail}}{\tau'} \\
%\rafjudge{\Pi}{\pc'^a}{\tau'^a}
}
{\TValGpcw{\runa{\tau}{e}{c'}}{τ}}

\Rule{Ret}{
\TValGpcw{e}{τ} \\
%\runnable{\Pi}{\tau}{c'} \\\\
\rafjudge{\Pi}{c'}{\UB{\tau}} \\\\
\rafjudge{\Pi}{c}{pc} \\
}
{\TValGpcw{\ret{e}{c'}}{\says{pc{^{\integ \avail}}}{\tau}}}

\Rule{Compare}{      
\forall i \in \{1,2\}.\TValGpcw{e_i}{\says{\ell_i}{τ}} \\\\
\readjudge{\Pi}{c}{\says{\ell_i}{\tau}} \\
% j= (j_1 \sqcup j_2 )\\
\rafjudge{\Pi}{c}{pc}
}
{\TValGpcw{\compare{e_1}{e_2}}{\says{(\txcmp{\ell_1}{\ell_2})}{\tau}}}
%{\says{(\ell_1 \wedge \ell_2 )}
%{\tau}{j_1 \sqcup j_2}}}

\Rule{Select}{
\forall i \in \{1,2\}.
\TValGpcw{e_i}{\says{\ell_i}{τ}} \\\\
%\rafjudge{\Pi}{c}{(\ell_2 \join \ell_2)^i} \\ 
% the above does not make sense as in select we dont sign we downgrade integrity.
\rafjudge{\Pi}{c}{pc} \\ 
}
{\TValGpcw{\select{e_1}{e_2}}{\says{(\txsel{\ell_1}{\ell_2})}{\tau}}}

%\end{mathpar}
%\end{flushleft}
%\end{figure*}
%
%\begin{figure*}
%\begin{flushleft}
%  \rulefiguresize
%  \boxed{\TValGpcw{e}{τ}} \\
%  \begin{mathpar}
%
\hfill
\end{mathpar}
\end{flushleft}
\caption{Typing rules for expressions}
\label{fig:types}
\end{figure*}

There are two forms of typing judgements in FLAQR: local typing
judgments for expressions and global typing judgments for the stack.
Local typing judgments have the form
$\TVal{\Pi;\Gamma;pc;c}{e}{\tau}$.  $Π$ is the program's delegation
context and is used to derive acts for relationships with the rules in
Figures~\ref{fig:partialAF} and~\ref{fig:partialactsforfull}. $Γ$ is
the typing context containing in-scope variable names and their types.
The $pc$ label tracks the information flow policy on the program
counter (due to control flow) and on unsealed protected values such as
in the body of a "bind".  

Figure~\ref{fig:types} presents a selection of local typing rules.
Each typing rule includes an acts-for premise of the form $\rafjudge{\Pi}{c}{pc}$. 
This captures the invariant that each host principal $c$ has control of the
program it executes locally. Thus the $\pc$ an expression is typed at should never
exceed the authority of the principal executing the expression.
Rules~\ruleref{Fail} and~\ruleref{Expect} type "fail" and "expect" terms according to
their type annotation $τ$.
Rule~\ruleref{Lam} types lambda abstractions.  Since functions are first-class values,
we have to ensure that the $\pc$ annotation on the lambda term preserves the
invariant $\rafjudge{\Pi}{c}{pc}$.  The _clearance_ of a type $τ$, written $\UB{\tau}$,
is an upper bound on the $\pc$ annotations of the function types in $τ$. By checking
that $\rafjudge{Π}{c}{\UB{\func{\tau_1}{pc'}{\tau_2}}}$ holds (along with similar checks in
\ruleref{Run} and \ruleref{Ret}), we ensure the contents of the lambda term is protected when
sending or receiving lambda expressions, and that hosts never receive a function they
cannot securely execute.
Due to space constraints, the definition of $\UB{·}$ is presented in
Appendix~\ref{sec:ubrules}, Figure~\ref{fig:UBFunction}.

The \ruleref{App} rule requires the \textsl{pc} label at any function
application to flow to the function's \textsl{pc} label annotation.
Hence the premise $\drflowjudge{\Pi}{pc}{pc'}$.

Protected terms $\return{\ell}{e}$ are typed by rule \ruleref{UnitM}
as $\says{\ell}{τ}$ where $τ$ is the type of $e$. Additionally, it
requires that $\drflowjudge{\Pi}{pc}{\ell}$.  This ensures that any
unsealed values in the context are adequately protected by policy $ℓ$
if they are used by $e$.
The \ruleref{Sealed} rule types protected values $\returnv{\ell}{v}$.
These values are well-typed at any host that $v$ is, and does not require
$\drflowjudge{\Pi}{pc}{\ell}$ since no unsealed values in the context could be
captured by the (closed) value $v$.
 
Computation on protected values occurs in bind terms
$\bind{x}{e'}{e}$.  The policy protecting $e$ must be at least as
restrictive as the policy on $e'$ so that the occurrences of $x$ in
$e$ are adequately protected.  Thus, rule \ruleref{BindM} requires
$\drflowjudge{\Pi}{\ell \join pc}{\tau}$, and furthermore $e$ is typed
at a more restrictive program counter label $\ell \join \pc$ to reflect
the dependency of $e$ on the value bound to $x$.

Rule \ruleref{Run} requires that the $\pc$ at the local host flow to
the $\pc'$ of the remote host, and that $e$ be well-typed at $c'$,
which implies that $c'$ acts for $\pc'$. Additionally, $c$ must act
for the clearance of the remote return type $τ'$ to ensure $c$ is
authorized to receive the return value. The type of the run expression
is $\says{pc'^{{\integ}{\avail}}}{τ'}$, which reflects the fact that
$c'$ controls the availability of the return value and also has some
influence on which value of type $τ'$ is returned. Although $c'$ may
not be able to \emph{create} a value of type $τ'$ unless
$pc'^{{\integ}{\avail}}$ flows to $τ'$, if $c'$ has \emph{access} to
more than one value of type $τ'$, it could choose which one to return.
Rule \ruleref{Ret} requires that expression $e$ is welltyped at $c$ and
that $c'$ is authorized to receive the return value based on the clearance of $τ$.

The \ruleref{Compare} rule gives type $\says{(\txcmp{\ell_1}{\ell_2})}{\tau}$ to
the expression $\compare{e_1}{e_2}$
where $e_1$ and $e_2$ have types $\says{\ell_1}{\tau}$ and 
$\says{\ell_2}{\tau}$ respectively. 
Additionally, it requires that $c$, the host executing the "compare",
is authorized to fully examine the results of evaluating $e₁$ and $e₂$
so that they may be checked for equality.\footnote{Assuming a more
sophisticated mechanism for checking equality that reveals less
information to the host such as zero-knowledge proofs or a trusted
execution environment could justify relaxing this constraint.}
This requirement is captured by the premise
$\readjudge{\Pi}{c}{\says{\ell_i}{\tau}}$, pronounced ``$c$ _reads_ $\says{\ell_i}{\tau}$''.
The inference rules for the reads judgment are found in Figure~\ref{fig:readjudgment} in
Appendix~\ref{sec:fullrules}.

Finally, the \ruleref{Select} rule gives type 
$\says{(\txsel{\ell_1}{\ell_2})}{\tau}$ to the expression $\select{e_1}{e_2}$ 
where $e_1$ and $e_2$ have types $\says{\ell_1}{\tau}$ and 
$\says{\ell_2}{\tau}$ respectively. 

The typing judgment for the global 
configuration is presented in Figure~\ref{fig:stacktypes}
and consists of three rules.
Rule \ruleref{Head} shows that 
the global configuration $\distcon{e}{c}{t}$,
is well-typed if the 
expression $e$ is well-typed at host 
$c$ with program counter $pc'$
where $\drflowjudge{\Pi}{pc}{pc'}$ and the  
tail $t$ is well-typed.  
$[\tau']\tau$ means that 
the tail of the stack is of type $\tau$ 
while the expression in the 
head of the configuration is
of type $\tau'$. We introduced rules
\ruleref{Tail}(when $t \neq \emptystack$) and 
\ruleref{Emp}(when $t = \emptystack$) to
typecheck the tail $t$.

$\stackapp{E[\expecta{{\tau'}}]}{c}{t}$
is well-typed with type $[\tau']\tau$, 
if expression $E[\expecta{{\tau'}}]$
is well-typed with type $\hat{\tau}$ at host c.
And, the rest of the stack $t$ needs to be well-typed 
with type $[\hat{\tau}]\tau$.
%but now the type of the head of the 
%configuration is updated to $\hat{\tau}$.
Rule \ruleref{Emp} says the tail is empty
and the type of the expression in the head of the  
configuration is $\tau$, 
in which case the type of the whole stack is $[\tau]\tau$.

\begin{figure}
%\begin{subfigure}
\begin{flushleft}
  %\rulefiguresize
  \boxed{\TValP{\Gamma;\pc}{\distcon{e}{c}{t}}{τ}} \\
\begin{mathpar}
\Rule{Head}{
\TVal{\Pi;\Gamma;pc';c}{e}{\tau'} \\ \TValGpcS{t}{[\tau']\tau} \\\\
\drflowjudge{\Pi}{\pc}{\pc'} \\ 
\rafjudge{\Pi}{c}{pc} 
}
{\TValGpcS{\distcon{e}{c}{s}}{\tau}}
\end{mathpar}
 \boxed{\TValP{\Gamma;\pc}{\distS{e}{c}{t}}{[{\tau'}]τ}} \\
\begin{mathpar}
\Rule{Tail}{
\TVal{\Pi;\Gamma;pc';c}{E[\expecta{{\tau'}}]}{\hat{\tau}} \\ 
\TValGpcS{t}{[\hat{\tau}]\tau} \\\\
 \drflowjudge{\Pi}{\pc}{\pc'} \\
\rafjudge{\Pi}{c}{pc} \\
}
{\TValGpcS{\distS{E[\expecta{{\tau'}}]}{c}{t}}{[{\tau'}]\tau}} 

\Rule{Emp}{}
{\TValGpcS{\emptystack}{[\tau]\tau}}
\end{mathpar}
\end{flushleft}
\caption{Typing rules for configuration stack}
\label{fig:stacktypes}
\end{figure}

\section{Availability Attackers} \label{sec:availAttacks}
Availability attackers are fundamentally different from traditional integrity
and confidentiality attackers.  While an integrity attacker's goal is to
manipulate data and a confidentiality 
attacker's goal is to learn secrets, an
availability attacker's goal is to cause failures. 
An availability attacker can substitute a value 
only with a "fail" term.
Integrity attackers may also cause failures in consensus
based protocols when consensus is not reached because 
of their data manipulation.
In FLAQR this scenario is relevant during executing 
a "compare" statement: if one of the values 
in the "compare" statement is substituted with a wrong value,
by an integrity attacker,
then a "fail" term is returned.
Thus we need to consider an 
availability attacker's
integrity authority as well while reasoning about its power to "fail" a 
program. 
Which means the maximum authority 
of principal $\ell$ as an availability
attacker is $\reachn{\ell}$.

We consider a
static but active attacker model 
similar to those used in Byzantine consensus protocols. 
By static we mean which principal or
collection of principals can 
act maliciously is fixed prior 
to program execution. By active we mean that
the attackers may manipulate 
inputs (including higher-order functions) 
during run time.  We formally define the
power of an availability attacker with respect to quorum systems.

Availability attackers in FLAQR 
are somewhat different than 
integrity and confidentiality attackers 
because we want to represent
multiple possible attackers but limit which
attackers are active for a particular execution.
This goal supports the bounded fault assumptions
found in consensus protocols where system configurations
assume an upper bound on the number of faults possible.

A quorum system $\quo$ is represented as set of sets of hosts (or
principals) e.g. $\quo = \{q_1, q_2, \ldots, q_n\}$. Here each $q_i$
represents a set of principals whose consensus is adequate for the
system to make progress. 
We define availability attackers in terms
of the _toleration set_ $\trans{\quo}$ of a quorum system $\quo$.
The toleration set is a set of principals where
each principal represents an upper bound
on the authority of an attacker the quorum can tolerate without failing.

\begin{example}\hfill
\begin{enumerate}
\item The toleration set for quorum 
$\quo_1=\{q_1:=\{a,b\};q_2:=\{b,c\};q_3:=\{a,c\}\}$ is
${\transi{\quo_1}{}} = \{\reachn{a}, \reachn{b}, \reachn{c}\}$,
\item For heterogeneous quorum system 
$\quo_2 = \{q_1 := \{p,q\}; q_2 := \{r\}\}$
the toleration set is
${{\transi{\quo_2}{}}}= 
\{ \reachn{p} \wedge \reachn{q}, \reachn{r}\}$ 
%and attacker set is $\mathcal{A_{\transi{\quo_2}{\avail}}}=$.
%\end{example}

%\begin{example}
\item For $\quo_3 = \{q := \{alice\}\}$
the toleration set is
${{\transi{\quo_3}{}}}=\{\}$, i.e. $\quo_3$ can
not tolerate any fault.
%and attacker set is $\mathcal{A_{\transi{\quo_2}{\avail}}}=$.
\end{enumerate}
\end{example}

An availability attacker's authority
is at most equivalent to a (single) principal's authority in the toleration set.
We define the set of all such attackers for a quorum  
$\mathcal{Q}$ as
$$\mathcal{A_{\transi{\mathcal{Q}}{}}} 
= \{\ell \mid \exists \ell' \in {\transi{\quo}{}}.
                             \rafjudge{\Pi}{\ell'}{\ell}\}.$$
which includes weaker attackers who a principal in the toleration set may act on behalf of.

\begin{figure}
{\footnotesize
\begin{flalign*}
& \boxed{\recrafjudge{\Pi}{ℓ}{τ}} &
\end{flalign*}
\begin{mathpar}

	\Rule{A-Pair}
	{\recrafjudge{\delegcontext}{\ell}{\tau_i} \quad i \in \{1,2\}
	}
	{\recrafjudge{\delegcontext}{\ell}{\prodtype{\tau_1}{\tau_2}}} 
	\Rule{A-Sum}
	{\recrafjudge{\delegcontext}{\ell}{\tau_i} \quad i \in \{1,2\}
	}
	{\recrafjudge{\delegcontext}{\ell}{\sumtype{\tau_1}{\tau_2}}} 

        \Rule{A-Fun}{
            \recrafjudge{\delegcontext}{ℓ}{\tau_2}}
          {\recrafjudge{\Pi}{\ell}{\func{τ_1}{\pc'}{\tau_2}}}

	\Rule{A-Type}
	{
        \recrafjudge{\Pi}{\ell}{\tau}
	}
	{\recrafjudge{\delegcontext}{\ell}{\says{\ell'}{\tau}}}
%
%     \Rule{A-TFun}
%       {\recrafjudge{\delegcontext}{ℓ}{\tau}}
%       {\recrafjudge{\delegcontext}{ℓ}{\tfuncpc{X}{\pc'}{\tau}}}
%

        \Rule{A-Avail}
        {
        \rafjudge{\Pi}{\ell^{\avail}}{\ell'^{\avail}}
        }
   	{\recrafjudge{\delegcontext}{\ell}{\says{\ell'}{\tau}}}
	\Rule{A-IntegCom}
	{\rafjudge{\delegcontext}{\ell^{\integ}}{{{\ell_j}^{\integ}}} 
	, j\in \{1,2\}}
	{\recrafjudge{\delegcontext}{\ell}
        {\says{(\txcmp{\ell_1}{\ell_2})}{\tau}}}
	
\end{mathpar}
}
\caption{$fails$ judgments.}
\label{fig:avail-actsfor}
\end{figure}
The fails relation ($\gtrdot$)
determines whether a principal can cause a program of a particular type to
evaluate to "fail".  Similar to the reads judgment, the fails judgment not only considers the
outermost "says" principal, but also any nested "says" principals whose propagated
failures could cause the whole term to fail.
Figure~\ref{fig:avail-actsfor} defines the fails judgment, written $\recrafjudge{\Pi}{l}{\tau}$,
which describes when a principal $l$ can fail an expression of type $\tau$ in delegation
context $\Pi$.

Consider an expression 
 $\return{\ell}{(\return{\ell'}{e})}$ and an attacker 
principal $l_a$. If
$\rafjudge{\Pi}{l_a^{\confid}}{\ell'^{\confid}}$, and
$\notrafjudge{\Pi}{l_a^{\confid}}{\ell^{\confid}}$, 
then the attacker learns nothing by evaluating 
$\return{\ell}{(\return{\ell'}{e})}$.
Similarly, if $\rafjudge{\Pi}{l_a^{\integ}}{\ell'^{\integ}}$ and
$\notrafjudge{\Pi}{l_a^{\integ}}{\ell^{\integ}}$, then the attacker 
cannot influence the value $\return{\ell}{(\return{\ell'}{e})}$.

In contrast, if $\rafjudge{\Pi}{{l_a}^{\avail}}{{\ell'}^{\avail}}$, and
$\notrafjudge{\Pi}{l_a^{\avail}}{{\ell}^{\avail}}$, 
an availability attacker may cause
$\return{\ell'}{e}$ to evaluate to $\faila{\says{\ell'}{\tau}}$,
which steps to $\faila{\says{\ell}{(\says{\ell'}{\tau})}}$ by
\ruleref{E-SealedFail}.
The fails relation reflects this possibility.
Using \ruleref{A-Type} and 
\ruleref{A-Avail} ( or \ruleref{A-IntegCom} if $\ell'$
was of form $(\txcmp{\ell_1}{\ell_2})$ ) we get
$\recrafjudge{\Pi}{l_a}{\says{\ell}{(\says{\ell'}{\tau})}}$.

We use the fails relation and the attacker set to define which
availability policies a particular quorum system is capable of
enforcing. We say $\quo$ _guards_ $τ$ if the following rule applies:
\[
\Rule{Q-Guard}
        {
         %\notrafjudge{\delegcontext}{\ell}{\ell'} \\
         \forall \ell \in \A_{\transi{\quo}{}}.
          \notrecrafjudge{\Pi}{\ell}{\tau}
        }
        {\protsA{\quo}{\tau}}
\]

\begin{definition}[Valid quorum type]
A type $\tau$ is a _valid quorum type_
with respect to quorum system $\quo$
and delegation set $\Pi$ if the condition 
$\protsA{\quo}{\tau}$ is satisfied. 
\end{definition}

\begin{example}
If
$\quo=\{q_1:=\{a,b\};q_2:=\{b,c\};q_3:=\{a,c\}\}$ and 
$\ell_{\quo}= \txsel{(\txcmp{a}{b})}{\txsel{(\txcmp{b}{c})}{(\txcmp{a}{c})}}$
then $\says{\ell_{\quo}}{(\says{a}{\tau})}$ is not 
a valid quorum type because 
$\nprotsA{\quo}{(\says{\ell_{\quo}}{(\says{a}{\tau})})}$ as 
$\recrafjudge{\Pi}{a^{{\integ}{\avail}}}{\says{\ell_{\quo}}{(\says{a}{\tau})}}$
and $a^{{\integ}{\avail}} \in \A_{\trans{\quo}}$.      
But it is a valid quorum type for heterogeneous quorum system  
$\quo' = \{q_1:=\{a,b\}; q_2:=\{a,c\}\}$ as 
$a^{{\integ}{\avail}} \notin \A_{\trans{\quo'}}$.
\end{example}

%Our goal is to ensure low-availability principals 
%cannot interfere with highly 
%available outputs. The types of the programs
%indicate whether the output will be eventually produced 
%with respect to a certain quorum and 
%its attacker set, or not.
%Thus we can reason about availability (along with
%integrity and confidentiality) of outputs 
%statically.
%In section
%\ref{sec:secPropNI} we show that
%FLAQR programs ensure noninterference.
 
\section{Security Properties}\label{sec:secProp}

To evaluate the formal properties of FLAQR, we prove that FLAQR
preserves noninterference for confidentiality, integrity, and
availability (section \ref{sec:secPropNI}).  
These theorems state that attackers cannot learn secret
inputs, influence trusted outputs, or control the failure behavior of
well-typed FLAQR programs. In addition, we also prove additional
theorems that formalize the soundness of our type system with
respect to a program's failure behavior.

\subsection{Soundness of failure}
\label{blameproofs}

FLAQR's semantics uses the "compare" and "select" security
abstractions and the failure propagation rules to model failure and
failure-tolerance in distributed programs, and FLAQR's type
system lets us reason statically about this failure behavior.  To verify
that such reasoning is _sound_, we prove two related theorems
regarding the type of a program and the causes of potential failures.

In pursuit of this goal, this section introduces our _blame semantics_
which reasons about failure-causing (faulty) principals during program
execution. 
The goal is to record the set of 
principals which may cause run-time failures as
a constraint on the set of faulty nodes $\mathbf{\FN}$.
Figure~\ref{fig:blame} presents the syntax of _blame constraints_, which are boolean 
formulas representing a lower bound on the contents of $\FN$.
Atomic constraints $\inF{\ell}{\FN}$ denote that label $\ell$ is in faulty set $\FN$.
This initial blame constraint ($\blame_{init}$) is represented 
using the toleration set of
the implied quorum system.  

\begin{figure}
  \small
  \[
    \begin{array}{rcl}
      \blame &::=& \FN = \emptyset \sep \blam \\[0.4em]
      \blam &::=& \inF{\ell}{\FN} \sep \blam_1 \OR \blam_2 \sep 
                   \blam_1 \AND \blam_2 \\[0.4em] 
% \sep \conS \OR \conS \sep \conS \AND \conS
    \end{array}
  \]
  \caption{Blame constraint syntax}
  \label{fig:blame}
\end{figure}

\begin{figure}
{\footnotesize
\begin{flalign*}
& \boxed{\entails{\mathcal{C}_1}{\inF{\ell}{\FN}}} &
\end{flalign*}
\begin{mathpar}
        \Rule{C-In}{\rafjudge{\Pi}{\ell'}{\ell}}
        {\entails{\inF{\ell'}{\FN}}{\inF{\ell}{\FN}}}

	\Rule{C-Or}{\entails{\blame_1}{\inF{\ell}{\FN}} \\ \entails{\blame_2}{\inF{\ell}{\FN}}}
	{\entails{\blame_1 \OR \blame_2}{\inF{\ell}{\FN}}}

        \Rule{C-AndL}{\exists i \in \{1,2\}. ~\entails{\blame_i}{\inF{\ell}{\FN}}}
        {\entails{\blame_1 \AND \blame_2}{\inF{\ell}{\FN}}} 
\end{mathpar}
}
\caption{Blame membership}
\label{fig:bmmain}
\end{figure}

\begin{definition}[Initial blame constraint]
  \label{bist}
  For toleration set $\trans{\quo}$ of the form
  $\{\reachn{(p¹₁ ∧ ... ∧ p¹_{m₁})}, ..., \reachn{(p^{k}₁ ∧ ... ∧ p^{k}_{m_{k}})} \}$ 
  the initial blame constraint $\blame_{init}$ is defined as a (logical) disjunction
  of conjunctions:
  \begin{align*}
    \blame_{init} ≜ ~& (\inF{p¹₁}{\FN} \AND ... \AND  \inF{p¹_{m₁}}{\FN}) \OR ... \\
                            & \qquad \OR (\inF{p^{k}₁}{\FN} \AND ... \AND \inF{p^{k}_{m_{k}}}{\FN}) 
  \end{align*}
\end{definition}
Each disjunction represents a minimal subset of a possible satisfying assignment
for the faulty set $\FN$.  For brevity, we will refer to these subsets
as the _possible faulty sets_ implied by a particular blame constraint.
Observe that for quorum system $\quo$, there is a one-to-one
correspondence between every  $tᵢ \in {\transi{\quo}{}}$
and every possible faulty set $\FN₁, ..., \FN_{k}$ in $\blame_{init}$
where $\FN_{i}$ is the set implied by the $i^{th}$ disjunction in $\blame_{init}$
such that $tᵢ = \reachn{bᵢ}$, where $bᵢ = \bigwedge_{p \in \FNᵢ} p$.

Evaluation rule \ruleref{C-CompareFail}, 
in Figure \ref{fig:ccomparefail},
shows how function $\last$ (discussed below) 
updates the blame constraint from $\blame$ to $\blame'$.  
We omit the blame-enabled versions of other evaluation
rules since they simply propagate the blame constraint without
modification.

\begin{figure*}
  {\small
\begin{mathpar}
\derule{C-CompareFail}{v_1 \neq v_2 \quad \blame':= \last({v_1},{v_2}, \blame,
\ell_1,\ell_2)}
{\concon{\compare
{\returnv{ℓ₁}{v_1}}{\returnv{ℓ₂}{v_2}}}{c}{s}{\blame}}
{\concon{\faila{\says{(\txcmp{ℓ₁}{ℓ₂})}{τ}}}{c}{s}{\blame'}}
\end{mathpar}
}
\caption{\ruleref{E-CompareFail} with Blame Semantics.}
\label{fig:ccomparefail}
\end{figure*}

\begin{example}\hfill
\label{ex:fn1}

\begin{enumerate}
\item Quorum system
$\quo_1 = \{ q_1=\{a,b\}; q_2=\{b,c\}; q_3=\{a,c\} \}$ 
has toleration set
${\transi{\quo_1}{}}= \{ \reachn{a}, \reachn{b}, \reachn{c} \}$
and three possible faulty sets in $\blame_{init}$:
$\FN = \{a\}$ or $\FN = \{b\}$ or
$\FN = \{c\}$
\item Quorum system
$\quo_2 = \{q_1:=\{p,q\};q_2:=\{r\}\}$
has toleration set  
${{\transi{\quo_2}{}}}= 
\{ \reachn{p} \wedge \reachn{q}, \reachn{r}\}$ 
and two possible faulty sets in $\blame_{init}$:
$\FN = \{p,q\}$ or $\FN = \{r\}$.
\end{enumerate}
\end{example}

While, $\blame_{init}$ is defined statically according to
the type of the program, rule \ruleref{C-CompareFail}
updates these constraints according to actual failures
that occur during the program's execution. This approach
identifies ``unexpected'' failures not implied by
$\blame_{init}$.

For example, 
$\quo_2 =\{q_1:=\{p,q\};q_2:=\{r\}\}$ has two possible faulty sets
$\FN=\{p,q\}$ or $\FN=\{r\}$. 
The initial blame constraint is
$\blame_{init}  ::=  (\inF{p}{\FN} \AND \inF{q}{\FN})
            \OR (\inF{r}{\FN})  $

Placing blame for a specific failure in a distributed system is
challenging, (and often impossible!). For example, when a comparison
of values signed by $ℓ₁$ and $ℓ₂$ fails, it is unclear who to blame
since either principal (or a principal acting on their behalf) could
have influenced the values that led to the failure. We do know, however,
that at least one of them is faulty; recording this information
helps constrain the contents of possible faulty sets.

%\paragraph{Blame membership}
We can reason about principals that _must_ be in
$\FN$ by considering all possible faulty sets
implied by the blame constraints. We write
$\entails{\blame}{\inF{\ell}{\FN}}$ 
(read as $\blame$ \emph{entails} $\inF{\ell}{\FN}$), 
when every possible faulty set in $\blame$, 
has the $\inF{\ell}{\FN}$ clause. 
Figure~\ref{fig:bmmain} presents inference rules
for the $\entails{}{}$ relation.

For example, since
$\ell_1$ is included in all satisfying choices of $\FN$ below, 
we can say $\entails{\blame}{\inF{\ell_1}{\FN}}$.
\begin{align*}
\blame  = & (\inF{\ell_1}{\FN} \AND \inF{\ell_2}{\FN}) 
            \OR (\inF{\ell_1}{\FN} \AND \inF{\ell_3}{\FN})  \\
          & \OR (\inF{\ell_1}{\FN} \AND \inF{\ell_4}{\FN}) 
            \OR (\inF{\ell_1}{\FN} \AND \inF{\ell_5}{\FN})
\end{align*}

The $\last$ function 
(full definition in Appendix~\ref{sec:blameproof},
Figure~\ref{fig:Blameconst})
is used by rule \ruleref{C-CompareFail} 
to update $\blame$.
For an expression:
$$\comparea{}{\returnv{\ell_1}{v_1}}{\returnv{\ell_2}{v_2}}$$
with $v₁ ≠ v₂$, $\last({v_1},{v_2}, \blame, \ell_1,\ell_2)$  
updates the formulas in $\blame$ to reflect that either
$ℓ₁$ or $ℓ₂$ is faulty.
If $ℓ₁$ or $ℓ₂$ already _must_ be faulty, specifically if
$\entails{\blame}{\inF{\ell_1}{\FN}}$ 
or $\entails{\blame}{\inF{\ell_2}{\FN}}$, 
then the function does not update any formulas.  This
approach avoids blaming honest principals when the
other principal is already known to be faulty.

If neither $ℓ₁$ nor $ℓ₂$ are known to be faulty.  then function
$\last$ is called recursively on inner layers (i.e., nested
$\returnv{}{}$ expressions) of $v_1$ and $v_2$ until a subexpression
protected by a known-faulty principal is found. If no such layer is
present, then the principal protecting the innermost layer is added to
$\blame$ (or the outer principals if there are no inner layers).  Only
this principal has seen the unprotected value and thus could have
knowingly protected the wrong value.  Observe that for well-typed
"compare" expressions, only the outer layer of "compare"d terms may
differ in protection level, so there is less ambiguity when blaming an
inner principal.  See Figure~\ref{fig:Blameconst} for more details on
how blame is assigned based on the structure of compared terms.

Updated constraints are kept in disjunctive normal form.
Specifically, for compared terms $\returnv{ℓ₁}{v₁}$ and
$\returnv{ℓ₂}{v₂}$, with $v₁ ≠ v₂$, if initial blame constraint is:
$\blame_{init} ::= (\inF{p}{\FN} \AND \inF{q}{\FN}) \OR
(\inF{r}{\FN})$ Then $\last({v_1},{v_2}, \blame_{init},
\ell_1,\ell_2)$ returns
\begin{align*}
\blame'  = & (\inF{p}{\FN} \AND \inF{q}{\FN} \AND \inF{\ell_1}{\FN}) \\
           & \OR (\inF{p}{\FN} \AND \inF{q}{\FN} \AND \inF{\ell_2}{\FN})  \\
          & \OR (\inF{r}{\FN} \AND \inF{\ell_1}{\FN}) 
            \OR (\inF{r}{\FN} \AND \inF{\ell_2}{\FN})
\end{align*}

We can now state the soundness theorem for our blame semantics,
and apply it to prove a liveness result.
Theorem~\ref{th:failresult0} states that for any well-typed FLAQR program
with a failing execution, and
the faulty sets $\FNᵢ$ implied by $\blame'$ (the final constraint computed by
the blame semantics), it must be the case that the program's type $τ$
reflects the ability of the (possibly colluding) principals in $\FNᵢ$ to fail the
program.

\begin{restatable}[Sound blame]{theorem}{failresult}
%                                                 ^ env type  ^ macro name
\label{th:failresult0}
Given,
\begin{enumerate}
\item $\TValGpcw{\concon{e}{c}{\emptystack}{\blame_{init}}}{\tau}$
\item $\concon{e}{c}{\emptystack}{\blame_{init}} \stepsto 
\concon{\faila{\tau}}{c}{\emptystack}{\blame'}$ 
\end{enumerate} where $e$ is a source-level expression,\footnote{In other words,
 $e$ does not contain any $\fail$ terms.}

then for each possible faulty set $\FN_{i}$ implied by $\blame'$,
there is a principal $b_i = \bigwedge_{p \in \FNᵢ} p$ such that
$\recrafjudge{\Pi}{\reachn{b_i}}{\tau}$.
\end{restatable}
\begin{proof}
See Appendix for full proof.
\end{proof}

While Theorem~\ref{th:failresult0} characterizes the relationship
between a program's type and the the possible faulty sets for a
failing execution, it does not explicitly tell us anything about the
fault-tolerance of a particular program.  Since the type of a FLAQR
program specifies its availability policy (in addition to its
confidentiality and integrity), different FLAQR types will be tolerant
of different failures.  Below, we prove a liveness result for a common
case, majority quorum protocols.

\begin{theorem}[Majority Liveness] \label{th:majorityLive}
If $e$ is a source-level expression and: 
\begin{enumerate} 
\item $\TValGpcw{\concon{e}{c}{\emptystack}{\blame_{init}}}
{\tau}$ 
\item \label{c:guard} $\protsA{\quo}{\tau}$ 
\item $\quo$ is a $m/n$ majority quorum system
\item $\concon{e}{c}{\emptystack}{\blame_{init}} 
\stepsone^{*} \concon{\faila{\tau}}{c}{\emptystack}{\blame'}$ 
\end{enumerate}
then for every possible faulty set $\FN'$ implied by $\blame'$, 
$\lvert\FN'\rvert > (n-m)$.
\end{theorem}
\begin{proof}
From \eqref{c:guard}, we know
$\tau$ is a valid quorum type
for $\quo$ so
$\forall \ell \in \A_{\transi{\quo}{}}. \notrecrafjudge{\Pi}{\ell}{\tau}$.
Since $\mathcal{A}_{\transi{\quo}{}}$
is a superset of $\transi{\quo}{}$,
we also have $\forall t \in \transi{\quo}{}. \notrecrafjudge{\Pi}{t}{\tau}$.
Furthermore, from Definition~\ref{bist}, for each possible
faulty set $\FN_{i}$ implied by $\blame_{init}$, we know there
is a principal $tᵢ ∈ \trans{\quo}$ 
such that $tᵢ = \reachn{bᵢ}$, where $bᵢ = \bigwedge_{p \in \FNᵢ} p$.
Therefore, for each such $bᵢ$, we know $\notrecrafjudge{\Pi}{\reachn{bᵢ}}{\tau}$.

%From condition $3.$ we know, 
Since $\quo$ is an $m/n$ majority quorum system, 
every quorum is of size $m$ and every faulty set
in $\blame_{init}$ is of size $(n-m)$. 
For contradiction, assume there exists a
faulty set $\FN'$ satisfying
$\blame'$ that has size $(n-m)$.
Then by the definition of $\last$, all possible
faulty sets implied by $\blame'$ also have
size $(n-m)$ since $\last$ monotonically increases
the size of all possible faulty sets or none of them.
Furthermore, each possible faulty set implied by $\blame_{init}$
is a subset (or equal to) a possible faulty set implied by $\blame'$,
so $\lvert\FN'\rvert = (n-m)$ implies $\blame_{init}=\blame'$.

From Theorem~\ref{th:failresult0}
we know for every possible faulty set $\FN'ᵢ$
implied by $\blame'$, it must be the case that
$\recrafjudge{\Pi}{\reachn{b'ᵢ}}{\tau}$, where
$\bigwedge_{p\in \FN'ᵢ}p$. However, since $\blame_{init} = \blame'$,
we have a contradiction since \eqref{c:guard} implies
$\notrecrafjudge{\Pi}{\reachn{b'ᵢ}}{\tau}$.
Thus there cannot exist a possible faulty set of size (at least) $(n-m)$
implied by $\blame'$, and all possible faulty sets must have size greater
than $(n-m)$.
\end{proof}

\subsection{Noninterference}\label{sec:secPropNI}

We prove noninterference by extending
the FLAQR syntax with bracketed expressions in the
style of Pottier and Simonet~\cite{ps02}. 
Figure \ref{fig:bracketProj} shows bracket projections.
Figure \ref{fig:brackets} shows 
bracketed evaluation rules and
Figure \ref{fig:bracketTypes} and 
\ref{fig:distbracketTypes}
show the typing rules for bracketed terms.
The soundness and completeness of the
bracketed semantics are proved in Appendix~\ref{sec:subjRedRel}
(Lemmata \ref{lemma:sound} - \ref{lemma:completedist}).

Noninterference often is expressed
with a distinct attacker label. 
We use $H$ to denote the attacker. 
This means the attacker can 
read data with label $\ell$ if
$\drflowjudge{\Pi}{\ell^{\confid}}{H^{\confid}}$
and can forge or influence it if
$\drflowjudge{\Pi}{H^{\integ}}{\ell^{\integ}}$
and can make it unavailable if 
$\drflowjudge{\Pi}{H^{\avail}}{\ell^{\avail}}$

%FLAQR is a language for writing fault-tolerant programs, 
%so even programs that successfully terminate may have subexpressions that 
%fail.
An issue in typing brackets is how to deal with $\fail$ terms.
Our confidentiality and integrity results are _failure-insensitive_
in the sense that they only apply to terminating executions.  This
is similar to how termination-insensitive noninterference is typically
characterized for potentially non-terminating programs.

Traditionally, bracketed typing rules require that bracketed terms
have a restrictive type, ensuring that only values derived from secret
(or untrusted) inputs are bracketed. In FLAQR, there are several
scenarios where a bracketed value may not have a restrictive type.
For example, when a "run" expression is evaluated within a bracket, it
pushes an element onto the configuration stack, but only in one of the
executions.  Another example is when a bracketed value occurs in a
"compare" expression, but the result is no longer influenceable by the
attacker $H$.  For these scenarios, several of the typing rules in
Figure~\ref{fig:bracketTypes} permit bracketed values to have less
restrictive types. Because of these rules, subject reduction does not
directly imply noninterference as it does in most bracketed approaches,
but the additional proof obligations are relatively easy to discharge.

\begin{center}
\begin{tabular}{ |c|c|c| } 
 \hline
  & \multicolumn{2}{l|}{May have less restrictive type?} \\
 \hline
 Term & $\pi = "i"$ & $\pi = "a"$ \\
 \hline
 $\bracket{v}{v'}$ & No & Yes \\ 
 \hline
 $\bracket{v}{\faila{\tau}}$ & Yes & No \\ 
 \hline
 $\bracket{v}{v}$ & Yes & Yes \\ 
 \hline
 $\bracket{\faila{\tau}}{\faila{\tau}}$ & Yes & Yes \\ 
 \hline
\end{tabular}
\label{table:availIntNI}
\end{center}

The table above summarizes how bracketed terms are typed depending on
whether we are concerned with integrity or availability. For
integrity, unequal bracketed values must have a restrictive type
(i.e., one that protects $H$), but equal bracketed values may have a
less restrictive type.  For availability, only bracketed terms where
one side contains a value and the other a failure must have a
restrictive type.

\subsubsection{Confidentiality and Integrity Noninterference} \label{sec:congIntNI}
To prove confidentiality (integrity) noninterference 
we need to show that given 
two different secret (untrusted) inputs to 
an expression $e$ the evaluated 
public (trusted) outputs are equivalent.
Equivalence is defined in terms of an observation function $\observe$ adapted from
FLAC~\cite{jflac} in Appendix~\ref{sec:subjRedRel}, Figure \ref{fig:observe}. 

\begin{restatable}["c"-"i" Noninterference]{theorem}{ciNonInterference} 
\label{th:ciNI}
If $\TValP{\Gamma,x:\says{\ell'}{\tau'}}
{\distcon{e}{c}{\emptystack}}{\says{\ell}{\tau}}$ where
\begin{enumerate}
\item $\TValGpcw{v_i}{\says{\ell'}{\tau'}}$, $i \in \{1,2\}$
\item $\distcon{\subst{e}{x}{\bracket{v_1}{v_2}}}{c}{\emptystack}
\stepsto \distcon{v}{c}{\emptystack}$ 
\item $\drflowjudge{\Pi}{H^{\pi}}{\ell'}$ and 
$\ndrflowjudge{\Pi}{H^{\pi}}{\ell}$, $\pi \in \{"c","i"\}$.
\end{enumerate}
then,
$\observet{\outproj{v}{1}}{\Pi}{\ell}{\pi} = 
\observet{\outproj{v}{2}}{\Pi}{\ell}{\pi}$
%\end{theorem} 
\end{restatable}
\begin{proof}
See Appendix for full proof.
\end{proof}

\subsubsection{Availability Noninterference}\label{sec:availNI}
Similar to \cite{qimp} our end-to-end 
availability guarantee is also expressed 
as noninterference property.
Specifically, if one run of a well-typed FLAQR program 
running on a quorum system
terminates successfully (does not fail),
then all other runs of the program also terminate successfully.

This approach treats ``buggy'' programs where every execution returns
$\fail$ regardless of the choice of inputs as noninterfering.  This
behavior is desirable because here we are concerned with proving the
absence of failures that attackers can _control_.  For structured
quorum systems with a liveness result such as
Theorem~\ref{th:majorityLive} for $m/n$ majority quorums, we can
further constrain when failures may occur.  For example,
Theorem~\ref{th:majorityLive} proves failures can only occur when more
than $(n-m)$ principals are faulty.  In contrast,
Theorem~\ref{th:availNI} applies to arbitrary quorum systems provided
they guard the program's type, but cannot distinguish programs where
all executions fail.

\begin{restatable}[Availability Noninterference]{theorem}{availNonInterference} 
\label{th:availNI}
If \\ $\TValP{\Gamma,x:\says{\ell}{\tau'}} {\distcon{e}{c}{\emptystack}}{\says{\ell_{\quo}}{\tau}}$
where
\begin{enumerate}
\item $\TValGpcw{f_i}{\says{\ell}{\tau'}}, i \in \{1,2\}$
\item $\distcon{\subst{e}{x}{\bracket{f_1}{f_2}}}{c}{\emptystack}
\stepsto \distcon{f}{c}{\emptystack}$
\item $\recrafjudge{\Pi}{H}{\says{\ell}{\tau'}}$ and
$H^{{\integ}{\avail}} \in \A_{\trans{\quo}}~$ and \\
${\protsA{\quo}{(\says{\ell_{\quo}}{\tau}})}$
\end{enumerate}
then $\outproj{f}{1} \neq \faila{\says{\ell_{\quo}}{\tau}} 
\Longleftrightarrow  \outproj{f}{2} \neq \faila{\says{\ell_{\quo}}{\tau}}$
\end{restatable}
\begin{proof}
See Appendix for full proof.
\end{proof}

\section{Examples revisited} \label{sec:examRev}
%Now that the readers are familiar with
%FLAQR semantics, type system and security
%properties 
We are now ready to implement  
the examples from section \ref{sec:section2.0}
with FLAQR semantics.
To make these implementations intuitive 
we assume that
our language supports integer ($int$)
type, mathematical
operator ">" (greater than), and
ternary operator ":?". 
Beacuse $int$ is a base type
$\UB{int}$ returns $\bot$.
The examples
also read from the local state of the 
participating principals. 
Which is fine because there are
standard ways to encode memory (reads/writes)
into lambda-calculus. 
%The example codes in \ref{sec:section2.0} did not consider
%confidentiality, 
%but the FLAQR implementations here
%takes confidentiality into account.
%We will see that
%With FLAQR semantics programmers 
%can take care of unavailability
%failures with blocks of "select" constructs.
%need not check for failures due
%to unavailability, the "select" blocks take care
%of that. 

\subsection{Tolerating failure and corruption}
In this FLAQR implementation (Figure \ref{flaqrimpl2.1}) of 
$2/3$ majority quorum example of section 
\ref{sec:section2.1}, 
we refer principals representing
\texttt{alice}, 
\texttt{bob} and \texttt{carol} as 
$a$, $b$ and $c$ respectively.
%to save space.
The program is executed at host $c'$ 
with program counter $pc$. 
Which means condition 
$\rafjudge{\Pi}{c'}{pc}$ holds.
The program body consists of a function
of type $\tau_f =$($\func{\tau_a}{pc}{\func{\tau_b}{pc}{\func{\tau_c}{pc}{}}}$
${(\says{(\txsel{(\txcmp{a^{{\integ}{\avail}}}{b^{{\integ}{\avail}}})}
{\txsel{(\txcmp{b^{{\integ}{\avail}}}
{c^{{\integ}{\avail}}})}
{(\txcmp{a^{{\integ}{\avail}}}{c^{{\integ}{\avail}}})}})}{\tau})
}$)
and the three arguments 
to the function are
"run" statements. 
Here $\tau$ is $\says{(\textit{a} \wedge \textit{b}
\wedge \textit{c})^{\confid}}{int}$.
Which means $\UB{\tau_f} = pc$. The
function body can be evaluated at $c'$, as 
condition $\rafjudge{\Pi}{c'}{pc}$ is true.% already.

\begin{figure}
\begin{lstlisting}
$(\lamc{x}{\tau_a}{pc}{}\lamc{y}{\tau_b}{pc}{}\lamc{z}{\tau_c}{pc}{}$
$\quad(\selex{}$
$\quad\quad({\compare{x}{y}}){}{}$ $\label{com1}$
$\quad\orex$
$\quad\quad(\selex{}$
$\quad\quad\quad{(\compare{y}{z)}}$ $\label{com2}$
$\quad\quad\orex$
$\quad\quad\quad{(\compare{x}{z})})))$ $\label{com3}$
$~(\runa{\tau_a}{e_a}{a})~(\runa{\tau_b}{e_b}{b})~(\runa{\tau_c}{e_c}{c})$
\end{lstlisting}
\caption{FLAQR implementation of majority quorum example}
\label{flaqrimpl2.1}
\end{figure}

Here $e_a$, $e_b$ and $e_c$ are the expressions
that read the balances for 
account $acct$ from the local
states of \textit{a}, \textit{b} and 
\textit{c} respectively.
The program counter at \textit{a}, 
\textit{b}, and \textit{c} are
\textit{a}, \textit{b} and \textit{c}
respectively.
The data returned from \textit{a}
has type $\tau_a$, which is basically 
$\says{{\textit{a}}^{{\integ}{\avail}}}{\tau}$.
Similarly $\tau_b$ is
$\says{{\textit{b}}^{{\integ}{\avail}}}{\tau}$ and
 $\tau_c$ is
$\says{{\textit{c}}^{{\integ}{\avail}}}{\tau}$.
Because each "run" returns a balance, the base type 
of $\tau$ 
is an $int$ type, and it is protected
with confidentiality label
$(\textit{a} \wedge \textit{b} \wedge \textit{c})^{\confid}$,
meaning anyone who can read all the three labels
(\textit{a}, \textit{b} and \textit{c}), 
can read the returned balances.

In order to typecheck the "run" statements the conditions 
$\drflowjudge{\Pi}{pc}{a}$,
$\drflowjudge{\Pi}{pc}{b}$, and 
$\drflowjudge{\Pi}{pc}{c}$ need to hold.
The condition $\rafjudge{\Pi}{c'}{\UB{\tau_a}}$
is trivially true as $\UB{\tau_a}= \bot$.
Similarly 
$\UB{\tau_b}= \bot$ and $\UB{\tau_c}= \bot$ as well.

The host executing the code need to be 
able to read the return values from the 
three hosts. This means conditions
$\readjudge{\Pi}{c'}{\says{{\textit{a}}^{{\integ}{\avail}}}{\tau}}$
$\readjudge{\Pi}{c'}{\says{{\textit{b}}^{{\integ}{\avail}}}{\tau}}$ and
$\readjudge{\Pi}{c'}{\says{{\textit{c}}^{{\integ}{\avail}}}{\tau}}$ need 
to hold in order to typecheck the "compare" statements.
% Which basically means $c$ can read values with
% label \textit{a}, \textit{b} and \textit{c}.
The type of the whole program is 
$(\says{(\txsel{(\txcmp{a^{{\integ}{\avail}}}
{b^{{\integ}{\avail}}})}
{\txsel{(\txcmp{b^{{\integ}{\avail}}}
{c^{{\integ}{\avail}}})}
{(\txcmp{a^{{\integ}{\avail}}}{c^{{\integ}{\avail}}})}})}{\tau})$
%The label $(\txsel{(\txcmp{a}{b})}
%\txsel{(\txcmp{b}{c})}{(\txcmp{a}{c})}})$
, which is a valid quorum type for 
$\quo = \{q_1:=\{a, b\}; q_2:=\{b,c\}; q_3:=\{a, c\}\}$.

Based on the security properties defined in section \ref{sec:secProp}
this program offers the confidentiality, integrity 
and availability guaranteed 
by quorum system $\quo$.
Therefore, the result cannot be learned or influenced by
unauthorized principals, and will be available 
as long as two hosts out of \textit{a}, 
\textit{b}, and \textit{c} are non-faulty.

The toleration set here is
${\transi{\quo}{}} = \{ \reachn{a}, \reachn{b}, \reachn{c}\}$. 
So, the program is not safe against an attacker with label
$l_a=\reachn{a} \wedge \reachn{b}$ 
(or, $a^{\integ} \wedge b^{\avail}$), for example. 
This is because 
$ \nexists t \in {\transi{\quo}{}}.\rafjudge{\Pi}{t}{l_a}$.
Since $\rafjudge{{\Pi}}{l_a}{\reachn{a}}$, principal $l_a$ can 
fail two "compare" statements on lines \ref{com1} and \ref{com3}. 
And, because $\rafjudge{{\Pi}}{l_a}{\reachn{b}}$, 
${l_a}$ can also fail
another two "compare" statements 
(one overlapping "compare" statment)
on lines \ref{com1} and \ref{com2}.
Thus the whole program evaluates to $\fail{}$. 

This FLAQR code also prevents programmers 
from comparing wrong values.
For instance 
the programmer can not write $y$ instead 
of $z$ on line \ref{com3}, as then the program 
will not typecheck.

\subsection{Using best available services}
The code in Figure \ref{flaqrimpl2.2} is the FLAQR implementation 
of Figure \ref{fig:example2}. 
The program runs at a host $c$ 
with program counter $pc$.
The expressions $e$ and $e'$ 
read account balances 
from principals $b$ and $b'$, representing the banks.
The values returned from $b$ and $b'$ have types 
$\tau_b = (\says{b^{{\integ}{\avail}}}
{\says{(b^{\confid}\wedge b'^{\confid})}{"int"}})$ 
and $\tau_{b'} = (\says{b'^{{\integ}{\avail}}}
{\says{(b^{\confid}\wedge b'^{\confid})}{"int"}})$ respectively.

The type of the whole program is 
$(\says{(\txsel{d}{\txsel{b^{{\integ}{\avail}}}
{b'^{{\integ}{\avail}}})}}
{(\says{b^{\confid}\wedge b'^{\confid})}{int}})$.
Here $d = pc \join b \join b'$.
In order to typecheck the "run" statements,
the conditions $\drflowjudge{\Pi}{pc}{b}$ and
$\drflowjudge{\Pi}{pc}{b'}$ need to hold.
The program counter at $b$ is $b$ and $b'$
is $b'$.
The bind statements (lines \ref{ex:bind1}-\ref{ex:bind2}) 
typecheck because conditions
$\drflowjudge{\Pi}{pc\join b^{{\integ}{\avail}}}{d}$, 
$\drflowjudge{\Pi}{pc\join b^{{\integ}{\avail}} 
\join b'^{{\integ}{\avail}}}{d}$, 
$\drflowjudge{\Pi}{pc\join b^{{\integ}{\avail}} 
\join b'^{{\integ}{\avail}} \join b^{\confid}}{d}$, and  
$\drflowjudge{\Pi}{pc\join b^{{\integ}{\avail}}
\join b'^{{\integ}{\avail}} \join b^{\confid} \join b'^{\confid}}{d}$ 
hold, because of our choice of $d$.

\begin{figure}
\begin{lstlisting}
$(\lamc{arg_1}{\tau_b}{pc}{(\lamc{arg_2}{\tau_{b'}}{pc}{}}$
$\quad(\selonly$
$\quad\quad{(\bind{x}{arg_1}{}}({\bind{y}{arg_2}{}}$ $\label{ex:bind1}$
$\quad\quad\quad(\bind{x'}{x}{}(\bind{y'}{y}{}$ $\label{ex:bind2}$
$\quad\quad\quad\quad{\returnv{d}{\returnv{(b^{\confid}\wedge b'^{\confid})}{(x'>y'~?~x'~:~y')}}}))))$
$\quad\orex$
$\quad(\select{(arg_1)}{(arg_2))})))(\runa{\tau_{b'}}{e'}{b'}))(\runa{\tau_b}{e}{b})$
\end{lstlisting}
\caption{FLAQR implementation of available largest balance example}
\label{flaqrimpl2.2}
\end{figure}

\section{Related work}

Only a limited number of previous approaches~\cite{avail, qimp}
combine availability with more common confidentiality and integrity
policies in distributed systems.
Zheng and Myers \cite{avail} extend the Decentralized Label Model
\cite{ml-tosem} with availability policies, but 
focus primarily tracking dependencies rather than applying mechanisms such as
consensus and replication to improve availability and integrity.
Zheng and Myers later introduce the language Qimp~\cite{qimp} with a type system
explicitly parameterized on a quorum system for offloading computation while enforcing
availability policies.
Instead of treating quorums specially, FLAQR quorums emerge naturally from interactions between
principals using "compare" and "select", and enable application-specific integrity
and availability policies that are secure by construction.
%
%Qimp also considers availability and integrity to be two separate entities,
%whereas, our technique shows that there is a trade-off between integrity and
%availability. In fact, Qimp accepts values whose integrity
%is compromised, but in FLAQR
%the outcomes either have the required integrity or they are a $\fail$ term.

Hunt and Sands~\cite{quantale} present a novel generalisation of
information flow lattices that captures disjunctive flows similar
to the influence of replicas in FLAQR on a "select"
result.  Our partial-or operation was inspired by their treatment of
disjunctive dependencies.

Models of distributed system protocols are often verified with model
checking approaches such as TLA+~\cite{plusCal}. Model checking
programs is typically undecidable, making it ill-suited to integrate
directly into a programming model in the same manner as a (decidable)
type system.  To make verification tractable, TLA+ models are often
simplified versions of the implementations they represent, potentially
leading to discrepancies.  FLAQR is designed as a core calculus for a
distributed programming model, making direct verification of
implementations more feasible.

BFT protocols \cite{pbft,bessani2014state} use consensus and
replication to protect the integrity and availability of operations on
a system's state. Each instance of a BFT protocol essentially enforces
a single availability policy and a single integrity policy.  While composing multiple
instances is possible, doing so provides no end-to-end availability or integrity
guarantees for the system as a whole.  FLAQR programs, by constrast, routinely
compose consensus and replication primitives to enforce multiple policies while
also providing end-to-end system guarantees.

\if 0
\section{Future Work}
Although FLAQR already 
is a strong programming model
for quorum replication based computations,
there is a number of ways 
FLAQR can be extended with.
Dynamic authorization 
can be easily added into FLAQR
by adding "assume" and "where"
constructs, following FLAC \cite{flac}.
This is possible because FLAQR security policies
are built using FLAM principal algebra
and FLAM integrates authorization logic
into information flow model.
Cryptographic mechanisms can also 
be incorporated into FLAQR.
For instance, it would be nice 
to have homomorphic comparisons
in \ruleref{Compare}, because 
then the host executing the 
compare will not require to read 
the two values, relaxing some 
restrictions on confidentiality.
To make distibuted computations
efficient, FLAQR programming model
can also be updated with channels
and prallel processes instead of having
configuration stacks.
%\cite{deflate}.
%, as a future goal 
\fi

\section{Conclusion}
%The main contributions of this work are as follows:
In this work, we extend Flow Limited Authorization Model \cite{flam}
with availability policies. We introduce a core calculus and 
type-system, 
FLAQR, for building decentralized applications 
that are secure by construction. 
We identify a trade-off relation between 
integrity and availability, and introduce two 
binary operations _partial-and_ and _partial-or_, 
specifically to express integrities of 
quorum based replicated programs.
We define $fails$ relation and judgments 
that help us reason about a principal's 
authority over availability of a type.  
We introduce blame semantics that 
associate failures with malicious 
hosts of a quorum system 
to ensure that quorums can 
not exceed a bounded number of failures
without causing the
whole system to fail.
FLAQR ensures end-to-end information security
with noninterference for confidentiality, integrity and availability.

\if 0
\section{Acknowledgements}
Funding for this work was provided by XXX grants YYYYYYYY
and ZZZZZZZ, but any opinions, findings, conclusions, or recommendations
expressed here are those of the author(s) and do not
necessarily reflect those of the XXX.
\fi

\newcommand{\showURL}[1]{\unskip}
\bibliographystyle{IEEEtran}
\bibliography{header}
\newpage
\appendix

\section{Complete FLAQR rule set}
\label{sec:fullrules}
\subsection{Expanded syntax, semantics, and type system} \label{sec:fulllang}
See Figures \ref{fig:Annotatedsyntax}, 
\ref{fig:annotatedlocsem}, 
and \ref{fig:Annotatedtypes} for full FLAQR syntax, semantics 
and type system respectively. 
Figure \ref{fig:Fullfailprop} shows all the "fail"
propagation rules.

\begin{figure}
  \small
  \[
    \begin{array}{rcl}
      \multicolumn{3}{l}{ π ∈ {\{"c","i","a"\}}  \text{  (projections)}} \\
      \multicolumn{3}{l}{ n ∈ \N \text{  (primitive principals)}} \\
      \multicolumn{3}{l}{ x ∈ \mathcal{V} \text{  (variable names)}} \\
      \\
      p,ℓ,\pc &::=&  n \sep \top \sep \bot \sep p^{π} \sep p ∧ p\sep p ∨ p \\[0.4em]
                   &\sep & p ⊔ p \sep p ⊓ p \sep \selor{p}{p} \sep \comor{p}{p} \\[0.4em]
      τ &::=& \voidtype \sep X \sep \sumtype{τ}{τ} \sep \prodtype{τ}{τ} \\[0.4em]
         & \sep & \func{τ}{\pc}{τ} \sep \tfunc{X}{\pc}{τ} \sep \says{ℓ}{τ}  \\[0.4em]
      v &::=& \void \sep \returnv{\ell}{v} \sep \injia{\sumtype{\tau}{\tau}}{v} \sep \paira{v}{v}{\tau} \\[0.4em]
        & \sep & \lamc{x}{τ}{\pc}{e} \sep \tlam{X}{\pc}{e}  \\[0.8em]
      f &::=& v \sep  \faila{\tau} \\[0.8em]
      e &::=& f \sep x \sep e~e \sep e~τ \sep \return{\ell}{e} \sep 
             \paira{e}{e}{\tau} \\[0.4em]
        & \sep & \proji{e} \sep \injia{\sumtype{\tau}{\tau}}{e} \sep \bind{x}{e}{e}\\[0.4em]
        & \sep & \casexpan{e}{x}{e}{e}{\tau} \\[0.4em]
        & \sep & \runa{\tau}{e}{p}  \sep \ret{e}{p}\\[0.4em]
        & \sep & \selecta{e}{e}{\tau} \sep \comparea{\tau}{e}{e} \sep \expecta{\tau} \\[0.4em] 
    \end{array}
  \]
  \caption{Type annotated FLAQR Syntax (Full version).}
  \label{fig:Annotatedsyntax}
\end{figure}

%\clearpage
\begin{figure}
%  \underline{Evaluation context}
  \[
    \begin{array}{rcl}
      E & ::= & [\cdot] \sep E~e \sep v~E \sep E~τ \sep 
\paira{E}{e}{\tau} \sep \paira{f}{E}{\tau} \sep \return{ℓ}{E} \\[0.4em]
        & \sep & \proji{E} \sep \injia{\tau}{E} \sep \bind{x}{E}{e}  \\[0.4em]
        & \sep & \casexpan{E}{x}{e}{e}{\tau} \\[0.4em]
        & \sep & \if 0 E[\runa{\tau}{e}{p}] \sep \fi \ret{E}{p}\\[0.4em]
        & \sep & \selecta{E}{e}{\tau} \sep \selecta{f}{E}{\tau} \\[0.4em]
        & \sep & \comparea{\tau}{E}{e} \sep \comparea{\tau}{f}{E} \\[0.4em]
    \end{array}
  \]
\caption{Evaluation context.}
\end{figure}
\begin{figure*}
\label{sequential semantics}
  \small
  \boxed{{e} \stepsone {e'}} \\
  \begin{mathpar}
   \erule{E-App}{}{{(\lamc{x}{τ}{\pc}{e})~v}}{{\subst{e}{x}{v}}}

    \erule{E-TApp}{}{{(\tlam{X}{\pc}{e})~τ}}{ {\subst{e}{X}{τ}}}

    \erule{E-UnPair}{}{ {\proji{\paira{v_1}{v_2}{\tau}}}}{ {v_i}}
 
    \erule{E-Sealed}{}{\return{ℓ}{v}}{\returnv{ℓ}{v}}

    \erule{E-BindM}{}{ {\bind{x}{\returnv{ℓ}{v}}{e}}}{ {\subst{e}{x}{v}}}

    \erule{E-Case}{}{(\casexpan{(\injia{\tau}{v})}{x}{e_1}{e_2}{\tau})}{ {\subst{e_i}{x}{v}}}

 \erule{E-Compare}{v_1 = v_2}{\comparea{\says{(\txcmp{ℓ₁}{ℓ₂})}{τ}}{ \returnv{ℓ₁}{v_1}}{\returnv{ℓ₂}{v_2}}}{\returnv{\txcmp{ℓ₁}{ℓ₂}}{v₁}} 

 \erule{E-CompareFail}{v_1 \neq v_2}
{\comparea{\says{(\txcmp{ℓ₁}{ℓ₂})}{τ}}{ \returnv{ℓ₁}{v_1}}
{\returnv{ℓ₂}{v_2}}}{\faila{\says{(\txcmp{ℓ₁}{ℓ₂})}{τ}}}

 \erule{E-CompareFailL}{}{\comparea{\says{(\txcmp{ℓ₁}{ℓ₂})}{τ}}
{\faila{\says{\ell_1}{\tau}}}
{f_2}}{\faila{\says{(\txcmp{ℓ₁}{ℓ₂})}{τ}}}
  
 \erule{E-CompareFailR}{}{\comparea{\says{(\txcmp{ℓ₁}{ℓ₂})}{τ}}{f_1}
{\faila{\says{\ell_2}{\tau}}}}{\faila{\says{(\txcmp{ℓ₁}{ℓ₂})}{τ}}}
 
\erule{E-Select}{
     f_{i} = \returnv{ℓᵢ}{wᵢ}  
\\ f_{j} ∈ \{\returnv{ℓ_{j}}{v_{j}},\; \faila{\says{ℓ_{j}}{τ}}\}
     }
     {\selecta{f₁}{f₂}{\says{(\txsel{ℓ₁}{ℓ₂})}{τ}}} {\returnv{\txsel{ℓ₁}{ℓ₂}}{vᵢ}}

   \erule{E-SelectFail}{} {{\selecta{(\faila{\says{\ell_1}{\tau}})}
{(\faila{\says{\ell_2}{\tau}})}{\says{(\txsel{\ell_1}{\ell_2})}{\tau}}}}
{\faila{\says{(\txsel{\ell_1}{\ell_2})}{\tau}}} 

  \erule{E-RetStep}{e \stepsone e'}{\ret{e}{c}}{\ret{e'}{c}}
    \erule{E-Step}{ {e} \stepsone  {e'}}{ {E[e]}}{ {E[e']}}
   
    \hfill
  \end{mathpar}

\caption{Full FLAQR local semantics }
\label{fig:annotatedlocsem}
\end{figure*}

\begin{figure}
  \small
\begin{flalign*}
  & \boxed{{e} \stepsone {\faila{\tau}}} &
\end{flalign*}
  \begin{mathpar}
\erule{E-AppFailL}{}{{\lamc{x}{\tau}{pc}{\faila{\func{\tau}{pc}{\tau'}}}}~e}{\faila{\tau'}}\\
\erule{E-AppFail}{}{\lamc{x}{\tau}{pc}{e}~{\faila{\tau}}}{\subst{e}{x}{\faila{\tau}}}

\erule{E-TAppFail}{}{\faila{\tfunc{X}{pc}{\tau}}~τ'}{{\faila{\subst{\tau}{X}{\tau'}}}}

\erule{E-SealedFail}{}{{\returnp{ℓ }{\faila{\tau}}}}{{\faila{\says{\ell}{\tau}}}}

\erule{E-InjFail}{}{\injia{\sumtype{\tau_1}{\tau_2}}{\faila{\tau_i}}}{\faila{\sumtype{\tau_1}{\tau_2}}}

\erule{E-CaseFail}{}{\casexpan{\faila{\tau'}}{x}{e_1}{e_2}{}
{}}{\faila{\tau}}

\erule{E-PairFailL}{}{\paira{\faila{\tau_1}}{f_2}{\prodtype{τ_1}{τ_2}}}{\faila{\prodtype{τ_1}{τ_2}}}

\erule{E-PairFailR}{}{\paira{f_1}{\faila{\tau_2}}{\prodtype{τ_1}{τ_2}}}{\faila{\prodtype{τ_1}{τ_2}}}

\erule{E-ProjFail}{}{\proji{\faila{\prodtype{\tau_1}{\tau_2}}}}{\faila{\tau_i}} 
\hfill
\end{mathpar}
\caption{Propagation of fail terms. }
\label{fig:Fullfailprop}
\end{figure}

\begin{figure}
{\small
\begin{flalign*}
& \boxed{\readjudge{\Pi}{p}{\tau}} &
\end{flalign*}
\begin{mathpar}
\Rule{R-Unit}
     {}
     {\readjudge{\Pi}{p}{\void}}

\Rule{R-Sum}
     {\readjudge{\Pi}{p}{\tau_1} \\
      \readjudge{\Pi}{p}{\tau_2}
     }
     {\readjudge{\Pi}{p}{\sumtype{\tau_1}{\tau_2}}}

\Rule{R-Prod}
     {\readjudge{\Pi}{p}{\tau_1} \\
      \readjudge{\Pi}{p}{\tau_2}
     }
     {\readjudge{\Pi}{p}{\prodtype{\tau_1}{\tau_2}}}

\Rule{R-Lbl}
        {\rafjudge{\delegcontext}{p^{\confid}}{\ell^{\confid}} \\
         \readjudge{\Pi}{p}{\tau} 
        }
        {\readjudge{\delegcontext}{p}{\says{\ell}{\tau}}}

\Rule{R-Fun}
           {
            \readjudge{\delegcontext}{p}{\tau_1} \\
            \readjudge{\delegcontext}{p}{\tau_2}
           }
          {\readjudge{\Pi}{p}{\func{τ_1}{\pc}{\tau_2}}}

\Rule{R-TFun}
           {
            \readjudge{\Pi}{p}{\tau}
           }
           {\readjudge{\Pi}{p}{\tfunc{X}{pc}{\tau}}}
\end{mathpar}
}
\caption{reads judgments.}
\label{fig:readjudgment}
\end{figure}

\begin{figure}
{\small
\begin{flalign*}
& \boxed{\protjudge*{ℓ}{τ}} &
\end{flalign*}
\begin{mathpar}
	\Rule{P-Unit}{}
       {\protjudge{\delegcontext}{\ell}{\voidtype}}

	\Rule{P-Pair}
	{\protjudge{\delegcontext}{\ell}{\tau_1} \\
		\protjudge{\delegcontext}{\ell}{\tau_2} 
	}
	{\protjudge{\delegcontext}{\ell}{\prodtype{\tau_1}{\tau_2}}} 

     \Rule{P-Fun}{
            \protjudge{\delegcontext}{ℓ}{\tau_2} \\
            \protjudge{\delegcontext}{ℓ}{\pc'}}
            {\protjudge{\Pi}{\ell}{\func{τ_1}{\pc'}{\tau_2}}}

     \Rule{P-TFun}{
       \protjudge{\delegcontext}{ℓ}{\tau} \\
       \protjudge{\delegcontext}{ℓ}{\pc'}}
       {\protjudge{\Pi}{\ell}{\tfunc{X}{\pc'}{\tau}}}

	\Rule{P-Lbl}
        {\rflowjudge{\delegcontext}{\ell}{\ell'} 
	}
	{\protjudge{\delegcontext}{\ell}{\says{\ell'}{\tau}}}

\end{mathpar}
}
\caption{Type protection levels}
\label{fig:protect}
\end{figure}

%\clearpage
%}

\begin{figure*}
%\Rule{PartandL}
%\begin{figure*}
{\small
\begin{subfigure}{\textwidth}
\begin{flalign*}
& \boxed{\rafjudge{\Pi}{p}{q}} &
\end{flalign*}
    \begin{mathpar}
      %SAF_Bot
    \rafrule[Π]{Bot}{}{p}{⊥}
    \and
      %SAF_Top
    \rafrule[\Pi]{Top}{}{⊤}{p}
    \and
      %SAF_Refl
    \rafrule[\Pi]{Refl}{}{p}{p}
    \and
      %SAF_Trans
    \rafrule[\Pi]{Trans}{
      \rafjudge{Π}{\!p\!}{\!q} \\
      \rafjudge{Π}{\!q\!}{\!r}
    }{\!p\!}{\!r}

    \\\\

      %SAF_Proj_*
    \rafrule[\Pi]{Proj}{\rafjudge{Π}{p}{q}}{p^{π}}{q^{π}}
    \and
      %SAF_ProjR_*
    \rafrule[\Pi]{ProjR}{}{p}{p^{π}}
    \and
      %SAF_ProjIdemp_*
    \rafrule[\Pi]{ProjIdemp}{}{(p^{π})^{π}}{p^{π}}
      %SAF_ProjBasis1/2
    \rafrule[\Pi]{ProjBasis}{π ≠ π'}{⊥}{(p^{π})^{π'}}

    %SAF_ProjDistConj*
    \rafrule[\Pi]{ProjDistConj}{}{p^{π} ∧ q^{π}}{(p ∧ q)^{π}} 
    \and
     %SAF_ProjDistDisj*
    \rafrule[\Pi]{ProjDistDisj}{}{(p ∨ q)^{π}}{p^{π} ∨ q^{π}}  

    \\\\

    % SAF_ConjL*
    \rafrule[\Pi]{ConjL}{
      \rafjudge{Π}{p_k}{p} \\\\
      k ∈ \{1,2\}
    }{p_1 ∧ p_2}{p}
    \and
      %SAF_ConjR
    \rafrule[\Pi]{ConjR}{
      \rafjudge{Π}{p}{p_1} \\\\
      \rafjudge{Π}{p}{p_2}
    }{p}{p_1 ∧ p_2}
    \and
      %SAF_ConjBasis
    \rafrule[\Pi]{ConjBasis}{}{p^{\confid} ∧ p^{\integ} \wedge 
p^{\avail}}{p} 

    \\\\

     \rafrule[\Pi]{DisjL}{
      \rafjudge{Π}{p_1}{p} \\\\
      \rafjudge{Π}{p_2}{p}
    }{p_1 ∨ p_2}{p}
    \and
    \rafrule[\Pi]{DisjR}{
      \rafjudge{Π}{p}{p_k} \\\\
      k ∈ \{1,2\}
    }{p}{p_1 ∨ p_2}
    \and
    \rafrule[\Pi]{DisjBasis}{
%\{\pi,\pi'\}\in \{"c","i","a"\}, 
\pi \neq \pi'}{⊥}{p^{\pi} ∨ q^{\pi'}}

  \\\\ 

      %SAF_ConjDistDisjL
    \rafrule[\Pi]{ConjDistDisjL}{}{(p ∧ q) ∨ (p ∧ r)}{p ∧ (q ∨ r)}
    \and
      %SAF_ConjDistDisjR
    \rafrule[\Pi]{ConjDistDisjR}{}{p ∧ (q ∨ r)}{(p ∧ q) ∨ (p ∧ r)} 
    \and
    \rafrule[\Pi]{DisjDistConjL}{}{(p ∨ q) ∧ (p ∨ r)}{p ∨ (q ∧ r)}
    \and
    \rafrule[\Pi]{DisjDistConjR}{}{p ∨ (q ∧ r)}{(p ∨ q) ∧ (p ∨ r)} 
    \hfill
    \end{mathpar}
   \caption{Static principal lattice rules, adapted from FLAC~\cite{jflac}}
  \label{fig:static}
\end{subfigure}
\hfill
\begin{subfigure}{\textwidth}
  \begin{mathpar}
\Rule{PAndL}
{\rafjudge{\Pi}{p_i}{p} 
% \rafjudge{\Pi}{p_2}{p}
}
{\rafjudge{\Pi}{\comor{p_1}{p_2}}{p}}

\Rule{PAndR}
{\rafjudge{\Pi}{p}{p_1}\\\\
 \rafjudge{\Pi}{p}{p_2}}
{\rafjudge{\Pi}{p}{\comor{p_1}{p_2}}}

\Rule{AndPAnd}
{}
{\rafjudge{\Pi}{p \wedge q}{\comor{p}{q}}}

\Rule{PAndPOr}
{}
{\rafjudge{\Pi}{\comor{p}{q}}{\selor{p}{q}}}

\Rule{ProjPAndL}
{}
{\rafjudge{\Pi}{\comor{p^{\pi}}{q^{\pi}}}{(\comor{p}{q})^{\pi}}}

\Rule{ProjPAndR}
{}
{\rafjudge{\Pi}{(\comor{p}{q})^{\pi}}{\comor{p^{\pi}}{q^{\pi}}}}

\Rule{ProjPOrL}
{}
{\rafjudge{\Pi}{\selor{p^{\pi}}{q^{\pi}}}{(\selor{p}{q})^{\pi}}}

\Rule{ProjPOrR}
{}
{\rafjudge{\Pi}{(\selor{p}{q})^{\pi}}{\selor{p^{\pi}}{q^{\pi}}}}

\Rule{POrOr}
{}
{\rafjudge{\Pi}{\selor{p}{q}}{p \vee q}}

\\\\

\Rule{AndDistPOrR}{}
{\rafjudge{\Pi}{p \wedge (\selor{q}{r})}
{\selor{(p \wedge q)}{(p \wedge r)}}}

\Rule{POrDistAndR}{}
{\rafjudge{\Pi}{\selor{p}{(q \wedge r)}}
{(\selor{p}{q}) \wedge (\selor{p}{r})}}

\Rule{AndDistPOrL}{}
{\rafjudge{\Pi}
{\selor{(p \wedge q)}{(p \wedge r)}}
{p \wedge (\selor{q}{r})}}

\Rule{{POrDistAndL}}{}
{\rafjudge{\Pi}
{(\selor{p}{q}) \wedge (\selor{p}{r})}
{\selor{p}{(q \wedge r)}}}

\Rule{{OrDistPOrR}}{}
{\rafjudge{\Pi}{p \vee (\selor{q}{r})}
{\selor{(p \vee q)}{(p \vee r)}}}

\Rule{{OrDistPOrL}}{}
{\rafjudge{\Pi}
{\selor{(p \vee q)}{(p \vee r)}}{p \vee (\selor{q}{r})}}

\Rule{{POrDistOrR}}{}
{\rafjudge{\Pi}{\selor{p}{(q \vee r)}}
{(\selor{p}{q}) \vee (\selor{p}{r})}}

\Rule{{POrDistOrL}}{}
{\rafjudge{\Pi}
{(\selor{p}{q}) \vee (\selor{p}{r})}{\selor{p}{(q \vee r)}}}

\\\\

\Rule{AndDistPAndR}{}
{\rafjudge{\Pi}{p \wedge (\comand{q}{r})}
{\comand{(p \wedge q)}{(p \wedge r)}}}

\Rule{PAndDistAndR}{}
{\rafjudge{\Pi}{\comand{p}{(q \wedge r)}}
{(\comand{p}{q}) \wedge (\comand{p}{r})}}

\Rule{AndDistPAndL}{}
{\rafjudge{\Pi}
{\comand{(p \wedge q)}{(p \wedge r)}}
{p \wedge (\comand{q}{r})}}

\Rule{{PAndDistAndL}}{}
{\rafjudge{\Pi}
{(\comand{p}{q}) \wedge (\comand{p}{r})}
{\comand{p}{(q \wedge r)}}}

\Rule{{OrDistPAndR}}{}
{\rafjudge{\Pi}{p \vee (\comand{q}{r})}
{\comand{(p \vee q)}{(p \vee r)}}}

\Rule{{OrDistPAndL}}{}
{\rafjudge{\Pi}
{\comand{(p \vee q)}{(p \vee r)}}{p \vee (\comand{q}{r})}}

\Rule{{PAndDistOrR}}{}
{\rafjudge{\Pi}{\comand{p}{(q \vee r)}}
{(\comand{p}{q}) \vee (\comand{p}{r})}}

\Rule{{PAndDistOrL}}{}
{\rafjudge{\Pi}
{(\comand{p}{q}) \vee (\comand{p}{r})}{\comand{p}{(q \vee r)}}}
\hfill
\end{mathpar}
\caption{Partial conjunction and disjunction acts-for rules.}
\label{fig:disjorandactsfor}
\end{subfigure}
}
\caption{Complete FLAQR acts-for rules.}
\label{fig:partialactsforfull}
\end{figure*}

\label{sec:ubrules}
\begin{figure}
\begin{flalign*}
   \UB{\func{\tau_1}{pc}{\tau_2}} &= \UB{\tau_1} \join pc \join \UB{\tau_2} \\
   \UB{\tfunc{X}{pc}{\tau}} &= pc \join \UB{\tau} \\
   \UB{\says{\ell}{\tau}} &= \UB{\tau} \\
   \UB{\sumtype{\tau_1}{\tau_2}} &= \UB{\tau_1} \join \UB{\tau_2} \\
   \UB{\prodtype{\tau_1}{\tau_2}} &= \UB{\tau_1} \join \UB{\tau_2} \\
   \UB{\voidtype} &= \bot
\end{flalign*}
\caption{Clearance function}
\label{fig:UBFunction}
\end{figure}

\begin{figure}[]
  \begin{align*}
  %\txsel{\ell_1}{\ell_2} &= \txsel{\ell_2}{\ell_1} \\
  % \txcmp{\ell_2}{\ell_1} &= \txsel{\ell_1}{\ell_2} \\
   %\txsel{\ell}{\ell} &= \ell \\
  %(\says{\ell_1}{\says{\ell_2}{\tau}})^a = (\ell_1)^a \vee (\ell_2)^a \vee \tau^a \\ 
   (\voidtype)^{\avail} &= \top \\
   (\sumtype{\tau_1}{\tau_2})^{\avail} &= {\tau_1}^{\avail} \join {\tau_2}^{\avail} \\
   (\prodtype{\tau_1}{\tau_2})^{\avail} &= {\tau_1}^{\avail} \join {\tau_2}^{\avail} \\
   (\says{\ell}{\tau})^{\avail} &= \ell^{\avail} \join \tau^{\avail} \\
   (\func{\tau_1}{pc}{\tau_2})^a &= 
        {\tau_1}^{\avail} \join pc^{\avail} \join {\tau_2}^{\avail} \\
   (\tlam{X}{pc}{\tau})^{\avail} &= 
        pc^{\avail} \join {\tau}^{\avail}
  \end{align*}
\label{fig:availOftype}
\caption{Availability of Types.}
\end{figure}

\begin{figure*}
%\begin{subfigure}
%\begin{figure*}
\begin{flushleft}
 % \rulefiguresize
  \boxed{\TValGpcw{e}{τ}} \\
  \begin{mathpar}
    \Rule{Var}{\Gamma(x)=\tau  \\ \rafjudge{\Pi}{c}{pc}}{\TValGpcw{x}{τ}}

    \Rule{Unit}{\rafjudge{\Pi}{c}{pc}}{\TValGpcw{\void}{\voidtype}}
    
    \Rule{Fail}{ \rafjudge{\Pi}{c}{pc}}
     {\TValGpcw{\faila{\tau}}{\tau}}
    
    \Rule{Lam}{%
     \TVal{\Pi;Γ,x\ty τ_1;\pc';u}{e}{τ_2} \\ \rafjudge{\Pi}{c}{pc} \\\\
     %\rafjudge{\Pi}{c}{\UB{\func{τ_1}{\pc'}{τ_2}}} \\  
     u = \UB{\func{τ_1}{\pc'}{τ_2}} \\
     \rafjudge{Π}{c}{u}
    }{\TValGpcw{\lamc{x}{τ_1}{\pc'}{e}}{\func{τ_1}{\pc'}{τ_2}}}

    \Rule{App}{%
      \TValGpcw{e_1}{\func{τ'}{\pc'}{τ}}\\\\
      \TValGpcw{e_2}{τ'} \\ 
      \drflowjudge{\Pi}{\pc}{\pc'}\\ 
      \rafjudge{\Pi}{c}{pc}
    }{\TValGpcw{e_1~e_2}{τ}}

    \Rule{TLam}{%
    \TVal{\Pi;Γ,X;\pc';u}{e}{τ} \\ \rafjudge{\Pi}{c}{pc} \\\\
    u = \UB{\tau} \\
    \rafjudge{\Pi}{c}{u} \\  
    %\rafjudge{Π}{c}{u}
    }{\TValGpcw{\tlam{X}{\pc'}{e}}{\tfunc{X}{\pc'}{τ}}}

    \Rule[$τ'$ is well-formed in $Γ$]{%
      TApp}{\TValGpcw{e}{\tfunc{X}{\pc'}{τ}}\\\\
      \drflowjudge{\Pi}{\pc}{\pc'} \\ \rafjudge{\Pi}{c}{pc}
    }{\TValGpcw{(e~τ')}{\subst{τ}{X}{τ'}}}

    \Rule{Pair}{%
      \TValGpcw{e_1}{τ_1} \\
      \TValGpcw{e_2}{τ_2} \\\\ 
      \rafjudge{\Pi}{c}{pc}
    }{\TValGpcw{\paira{e_1}{e_2}{\prodtype{τ_1}{τ_2}}}{\prodtype{τ_1}{τ_2}}}

    \Rule{UnPair}{%
      \TValGpcw{e}{\prodtype{τ_1}{τ_2}} \\ \rafjudge{\Pi}{c}{pc}
    }{\TValGpcw{\proji{e}}{τ_i}}

    \Rule{Inj}{%
      \TValGpcw{e}{τ_i} \\ \rafjudge{\Pi}{c}{pc}
    }{\TValGpcw{\injia{\sumtype{τ_1}{τ_2}}{e}}{\sumtype{τ_1}{τ_2}}}

    \Rule{Case}{%
      \TValGpcw{e}{\sumtype{τ_1}{τ_2}} \\ 
      \drflowjudge{\Pi}{\pc}{τ} \\\\ 
      \rafjudge{\Pi}{c}{pc} \\
      \rafjudge{\Pi}{{\tau_i}^a}{\tau^a} \\\\
      \TVal{\Pi;Γ,x\ty τ_1;\pc;\worker}{e_1}{τ} \\
      \TVal{\Pi; Γ,x\ty τ_2;\pc;\worker}{e_2}{τ} \\
    }{\TValGpcw{\casexpan{e}{x}{e_1}{e_2}{\sumtype{τ_1}{τ_2}}}{τ}}

    %\Rule{UnitM}{%
    \Rule{UnitM}{
      \TValGpcw{e}{τ} \\ 
      \drflowjudge{\Pi}{\pc}{ℓ} \\\\ %pc influences the signature ℓ 
      \rafjudge{\Pi}{c}{pc}  
%\\ \stflowjudge*{\jmath(\tau)}{ℓ}
    }{\TValGpcw{\return{ℓ}{e}}{\says{ℓ}{τ}}}   

     \Rule{Sealed}{
      \TValGpcw{v}{τ} \\
      \rafjudge{\Pi}{c}{pc}   %p{W}{ℓ}\\\\
      %\protjudge*{w^I}{ℓ^I} \\ %implicit due to p{c}{pc} and \stflowjudge*{\pc}{ℓ}
      %\protjudge*{ℓ^C}{w^C} \\\\
    }{\TValGpcw{\returnv{ℓ}{v}}{\says{ℓ}{τ}}} 

    \Rule{BindM}{%
      \TValGpcw{e'}{\says{ℓ}{τ'}} \\ 
      \drflowjudge{\Pi}{\ell \sqcup pc}{\tau} \\\\
      %\stflowjudge*{\jmath(\tau)} \\ 
      \TVal{\Pi;Γ,x\ty τ';ℓ\sqcup pc;\worker}{e}{τ} \\ 
      \rafjudge{Π}{c}{pc} 
    }{\TValGpcw{\bind{x}{e'}{e}}{τ}}

\Rule{Run}{
\TVal{\Pi;\Gamma;\pc';c'}{e}{τ'} \\
\drflowjudge{\Pi}{pc}{pc'} \\\\
\rafjudge{\Pi}{c}{pc} \\
\rafjudge{\Pi}{c}{\UB{\tau'}}\\\\
\tau = \says{pc'^{\integ \avail}}{\tau'} \\
%\rafjudge{\Pi}{\pc'^a}{\tau'^a}
}
{\TValGpcw{\runa{\tau}{e}{c'}}{τ}}

\Rule{Ret}{
\TValGpcw{e}{τ} \\
%\runnable{\Pi}{\tau}{c'} \\\\
\rafjudge{\Pi}{c'}{\UB{\tau}} \\\\
\rafjudge{\Pi}{c}{pc} \\
}
{\TValGpcw{\ret{e}{c'}}{\says{pc{^{\integ \avail}}}{\tau}}}

\Rule{Compare}{      
\forall i \in \{1,2\}.\TValGpcw{e_i}{\says{\ell_i}{τ}} \\\\
\readjudge{\Pi}{c}{\says{\ell_i}{\tau}} \\
% j= (j_1 \sqcup j_2 )\\
\rafjudge{\Pi}{c}{pc}
}
{\TValGpcw{\comparea{\says{(\txcmp{\ell_1}{\ell_2})}{\tau}}{e_1}{e_2}}{\says{(\txcmp{\ell_1}{\ell_2})}{\tau}}}
%{\says{(\ell_1 \wedge \ell_2 )}
%{\tau}{j_1 \sqcup j_2}}}

\Rule{Select}{
\forall i \in \{1,2\}.
\TValGpcw{e_i}{\says{\ell_i}{τ}} \\\\
%\rafjudge{\Pi}{c}{(\ell_2 \join \ell_2)^i} \\ 
% the above does not make sense as in select we dont sign we downgrade integrity.
\rafjudge{\Pi}{c}{pc} \\ 
}
{\TValGpcw{\selecta{e_1}{e_2}{\says{(\txsel{\ell_1}{\ell_2})}{\tau}}}{\says{(\txsel{\ell_1}{\ell_2})}{\tau}}}

     \Rule{Expect}{ \rafjudge{\Pi}{c}{pc}}
      {\TValGpcw{\expecta{\tau}}{\tau}}
 
\hfill
\end{mathpar}
\end{flushleft}
\caption{Typing rules for expressions (Full version).}
\label{fig:Annotatedtypes}
\end{figure*}
%proofs.tex

%\clearpage
%}

\begin{figure*}
\small
\begin{mathpar}
\text{\underline{Syntax}} \hfill \\
\begin{array}{rcl}
  v &::=& \dots \sep \bracket{v}{v} \\[0.4em]
  f &::=& \dots \sep \bracket{f}{f} \\[0.4em]
  e  &::=& \dots \sep \bracket{e}{e}
\end{array}\hfill \\

\begin{array}{rcl}
 \outproj{\emptystack}{k} & = & 
 {\emptystack} \\[0.4em]
\outproj{\stackapp{e}{c}{s}}{k} & = &
\stackapp{\outproj{e}{k}}{c}{\outproj{s}{k}} \\[0.4em]
\outproj{\distcon{e}{c}{s}}{k} & = &
\distcon{\outproj{e}{k}}{c}{\outproj{s}{k}} \\[0.4em]
 \outproj{E[\expecta{\tau}]}{k} & = & 
 \outproj{E}{k}[\expecta{\tau}] \\[0.4em]
\outproj{E[\bracket{e_1}{e_2}]}{k} & = & 
 \outproj{E}{k}[e_k] \\[0.4em]
\outproj{\comparea{\says{(\txcmp{\ell_1}{\ell_2})}{\tau}}
{f_1}{f_2}}{k} & = & 
\comparea{\says{(\txcmp{\ell_1}{\ell_2})}{\tau}}
{\outproj{f_1}{k}}{\outproj{f_2}{k}} \\[0.4em]
\outproj{\selecta
{f_1}{f_2}{\says{(\txsel{\ell_1}{\ell_2})}{\tau}}}{k} & = & 
\selecta{\outproj{f_1}{k}}{\outproj{f_2}{k}}
{\says{(\txsel{\ell_1}{\ell_2})}{\tau}}\\[0.4em]
\outproj{\ret{e}{c}}{k} & = & 
\ret{\outproj{e}{k}}{c} \\[0.4em]
\outproj{\runa{\tau}{e}{c}}{k} & = & 
\runa{\tau}{\outproj{e}{k}}{c} \\[0.4em]
\outproj{\proji{e}}{k} & = & 
\proji{\outproj{e}{k}} \\[0.4em]
\outproj{\casexp{e_1}{x}{e_2}{e_3}}{k} & = & 
\casexp{\outproj{e_1}{k}}{x}
{\outproj{e_2}{k}}{\outproj{e_3}{k}}\\[0.4em]
\outproj{\bracket{e}{\bullet}}{2} & = & \bullet 
\end{array}\hfill\\

\hfill
\\\\ 
 \end{mathpar}
%\end{subfigure}
\caption{Projection for bracketed expressions.}
\label{fig:bracketProj}
\end{figure*}

\begin{figure*}
\small
\begin{mathpar}
%\text{\underline{Evaluation rules}} \hfill \\
%\begin{subfigure}{\textwidth}
\berule{B-Step}{e_i \stepsone e'_i \\ e'_j =e_j \\ \{i, j \} = \{1, 2\} }{\bracket{e_1}{e_2}}{\bracket{e'_1}{e'_2}}{}

%\berule{B-UnitM}{}{\return{\ell}{\bracket{w_1}{w_2}}{j}}{\bracket{\return{\ell}{w_1}{j}}{\return{\ell}{w_2}{j}}}{}

    \berule*{B-App}{}{\bracket{v_1}{v_2}~v}
{\bracket{v_1~\outproj{v}{1}}{v_2 ~\outproj{v}{2}}}{}

    \berule*{B-TApp}{}{\bracket{v}{v'}~τ}{\bracket{v~τ}{v'~τ}}{}

    %\berule*{B-UnPair}{}{\proj{i}{\bracket{w}{w'}}}{\bracket{\proj{i}{w}}{\proj{i}{w'}}}{}
    \berule*{B-BindM}{}{\bind{x}{\bracket{v}{v'}}{e}}{\bracket{\bind{x}{v}{\outproj{e}{1}}}{\bind{x}{v'}{\outproj{e}{2}}}}{}

\berule*{B-CompareCommon}
{\outproj{\comparea{\says{\txcmp{\ell_1}{\ell_2}}{\tau}}
{\bracket{f_{11}}{f_{12}}}
{\bracket{f_{21}}{f_{22}}}}{i}
\stepsone f_i \quad \quad \forall i \in \{1,2\}}
{\comparea{\says{\txcmp{\ell_1}{\ell_2}}{\tau}}{\bracket{f_{11}}{f_{12}}}{\bracket{f_{21}}{f_{22}}}}
{\bracket{f_1}{f_2}}
{}

\berule*{B-CompareCommonRight}
{\outproj{\comparea{\says{\txcmp{\ell_1}{\ell_2}}{\tau}}
{\bracket{f_{11}}{f_{12}}}
{f}}{i}
\stepsone f_i \quad \quad \forall i \in \{1,2\}}
{\comparea{\says{\txcmp{\ell_1}{\ell_2}}{\tau}}{\bracket{f_{11}}{f_{12}}}{f}}
{\bracket{f_1}{f_2}}
{}

\berule*{B-CompareCommonLeft}
{\outproj{\comparea{\says{\txcmp{\ell_1}{\ell_2}}{\tau}}
{f}{\bracket{f_{21}}{f_{22}}}}{i}
\stepsone f_i \quad \quad \forall i \in \{1,2\}}
{\comparea{\says{\txcmp{\ell_1}{\ell_2}}{\tau}}{f}{\bracket{f_{21}}{f_{22}}}}
{\bracket{f_1}{f_2}}
{}

\berule*{B-SelectCommon}
{\outproj{\selecta{\bracket{f_{11}}{f_{12}}}{\bracket{f_{21}}{f_{22}}}{\says{\txsel{\ell_1}{\ell_2}}{\tau}}}{i} \stepsone f_i
\quad \quad \forall i \in \{1,2\}}
{\selecta{\bracket{f_{11}}{f_{12}}}{\bracket{f_{21}}{f_{22}}}{\says{\txsel{\ell_1}{\ell_2}}{\tau}}}{\bracket{f_1}{f_2}}
{}

\berule*{B-SelectCommonLeft}
{\outproj{\selecta
{\bracket{f_{11}}{f_{12}}}
{f}{\says{\txsel{\ell_1}{\ell_2}}{\tau}}}{i}
\stepsone f_i \quad \quad \forall i \in \{1,2\}}
{\selecta{\bracket{f_{11}}{f_{12}}}{f}{\says{\txsel{\ell_1}{\ell_2}}{\tau}}}
{\bracket{f_1}{f_2}}
{}

\berule*{B-SelectCommonRight}
{\outproj{\selecta
{f}{\bracket{f_{21}}{f_{22}}}{\says{\txsel{\ell_1}{\ell_2}}{\tau}}}{i}
\stepsone f_i \quad \quad \forall i \in \{1,2\}}
{\selecta{f}{\bracket{f_{21}}{f_{22}}}{\says{\txsel{\ell_1}{\ell_2}}{\tau}}}
{\bracket{f_1}{f_2}}
{}

\berule*{B-Fail1}{}
{\return{\ell}{\bracket{v}{\faila{\tau}}}}
{\bracket{\returnv{\ell}{v}}{\faila{\says{\ell}{\tau}}}}
{}

\berule*{B-Fail2}{}
{\return{\ell}{\bracket{\faila{\tau}}{v}}}
{\bracket{\faila{\says{\ell}{\tau}}}{\returnv{\ell}{v}}}
{}

\berule*{B-Fail}{}
{\return{\ell}{\bracket{\faila{\tau}}{\faila{\tau}}}}
{\faila{\tau}}
{}

\derule{B-RunLeft}{}{\distcon{\bracket{E[\runa{\tau}{e_1}{c'}]}{e_2}}
{c}{s}}{\distcon{\bracket{\ret{e_1}{c}}{\bullet}}{c'}
{\stackapp{\bracket{E[\expecta{\tau}]}
{e_2}}{c}{s}}}

\derule{B-RunRight}{}{\distcon{\bracket{e_1}{E[\runa{\tau}{e_2}{c'}]}}
{c}{s}}
{\distcon{\bracket{\bullet}{\ret{e_2}{c}}}{c'}{\stackapp
{\bracket{e_1}{E[\expecta{\tau}]}}{c}{s}}}

\derule{B-RetLeft}{
 f' = {\begin{cases}
     \returnv{\ell}{v} &\mbox{ if } f = v \\
     \faila{\says{\ell}{\tau}} &\mbox{ if } f = \faila{\tau}
      \end{cases}
     }}
 {\distcon{\bracket{\ret{f}{c}}{\bullet}}{c'}
{\stackapp{\bracket{E[\expecta{\says{\ell}{\tau}}]}
{e_2}}{c}{s}}}{\distcon{\bracket{E[f']}{e_2}}{c}{s}}

\derule{B-RetRight}{
 f' = {\begin{cases}
     \returnv{\ell}{v} &\mbox{ if } f = v \\
     \faila{\says{\ell}{\tau}} &\mbox{ if } f = \faila{\tau}
      \end{cases}
     }}
 {\distcon{\bracket{\bullet}{\ret{f}{c}}}{c'}
{\stackapp{\bracket{e_2}{E[\expecta{\says{\ell}{\tau}}]}}
{c}{s}}}{\distcon{\bracket{e_2}{E[f']}}{c}{s}}

\derule{B-RetV}
{f_i' = {\begin{cases}
          \returnv{\ell}{v} &\mbox{ if } f_i = v \\
          \faila{\says{\ell}{\tau}} &\mbox{ if } f_i = \faila{\tau}
         \end{cases}
}}
{\distcon{\ret{\bracket{f_1}{f_2}}{c}}{c'}
{\stackapp{E[\expecta{\says{\ell}{\tau}}]}{c}{s}}}
{\distcon{E[\bracket{f_1'}{f_2'}]}{c}{s}}
{}

\hfill
\\\\ 
 \end{mathpar}
%\end{subfigure}
\caption{Bracketed Evaluation Rules.}
\label{fig:brackets}
\end{figure*}

\begin{figure*}
\small
\begin{mathpar}
%\end{subfigure}
%\begin{subfigure}{\textwidth}
\text{\underline{Typing rules}} \hfill \\
    \Rule{Bracket}                                                                                 {
           \rflowjudge{\delegcontext}{(H^\pi \sqcup \pc)}{{\pc'}} \\
           e_1 = v_1  \iff e_2 \ne v_2 \\\\                        
           \TValP{Γ;\pc';c}{e_1}{\tau} \\
           \TValP{Γ;\pc';c}{e_2}{\tau} \\\\
           \protjudge{\delegcontext}{H^{\pi}}{\cfun{\tau}}\\
           \rafjudge{\Pi}{c}{pc} %\\ \pi \in \{"i","c"\}
           }
         {\TValGpcw{\bracket{e_1}{e_2}}{\tau}}

  \Rule{Bracket-Values}                                                                            {
           \TValGpcw{v_1}{\tau} \\                               
           \TValGpcw{v_2}{\tau} \\\\  
           \protjudge{\delegcontext}{H^\pi}{\cfun{\tau}} \\
          \rafjudge{\Pi}{c}{pc}
          }
         {\TValGpcw{\bracket{v_1}{v_2}}{\tau}}

  \Rule{BullR}
  {\TValP{\Gamma;\pc;c}{e}{\tau} \\
   %\pi \in \{"i","c"\}
  }
  {\TValP{\Gamma;\pc;c}{\bracket{e}{\bullet}}{\tau}}

  \Rule{BullL}
  {\TValP{\Gamma;\pc;c}{e}{\tau} \\
   %\pi \in \{"i","c"\}
  }
  {\TValP{\Gamma;\pc;c}{\bracket{\bullet}{e}}{\tau}}

  \Rule{Bracket-Fail-L}
  {\TValP{\Gamma;\pc;c}{e}{\tau} \\
   %\pi \in \{"i","c"\}
  }
  {\TValP{\Gamma;\pc;c}{\bracket{e}{\faila{\tau}}}{\tau}}

  \Rule{Bracket-Fail-R}
  {\TValP{\Gamma;\pc;c}{e}{\tau} \\
  % \pi \in \{"i","c"\}
  }
  {\TValP{\Gamma;\pc;c}{\bracket{\faila{\tau}}{e}}{\tau}}
  
  \Rule{Bracket-Fail-A}
  {\TValP{\Gamma;\pc;c}{e_i}{\tau} \\
   e_i \not= \faila{\tau} \\
   \pi = "a"
  }
  {\TValP{\Gamma;\pc;c}{\bracket{e_1}{e_2}}{\tau}}

  \Rule{Bracket-Same}
  {\TValP{\Gamma;\pc;c}{v}{\tau} \\
  }
  {\TValP{\Gamma;\pc;c}{\bracket{v}{v}}{\tau}}

 \end{mathpar}
%\end{subfigure}
\caption{Typing rules for Bracketed Expressions.}
\label{fig:bracketTypes}
\end{figure*}

\begin{figure*}
$\mathscr{C}(\voidtype) =  \voidtype$ \\
$\mathscr{C}(\sumtype{\tau_1}{\tau_2}) =  \sumtype{\mathscr{C}(\tau_1)}{\mathscr{C}(\tau_2)}$\\
$\mathscr{C}(\prodtype{\tau_1}{\tau_2}) =  \prodtype{\mathscr{C}(\tau_1)}{\mathscr{C}(\tau_2)}$\\
$\mathscr{C}(\func{\tau_1}{pc}{\tau_2}) =  \func{\mathscr{C}({\tau_1})}{pc}{\mathscr{C}({\tau_2})}$\\
$\mathscr{C}(\tlam{X}{pc}{\tau}) = \tlam{X}{pc}{\mathscr{C}(\tau)}$\\
$\mathscr{C}(\says{(\selor{\ell_1}{\ell_2})}{\tau}) =  \says{(\ell_1 \vee \ell_2)}{\mathscr{C}{(\tau)}}$\\
$\mathscr{C}(\says{(\comor{\ell_1}{\ell_2})}{\mathscr{\tau}}) =  \says{(\ell_1 \wedge \ell_2)}{\mathscr{C}{(\tau)}}$\\
$\mathscr{C}(\says{\ell}{\tau}) =  \says{\ell}{\mathscr{C}(\tau)}$(when $\ell$ is not of forms $(\selor{\ell_1}{\ell_2})$ or $(\comor{\ell_1}{\ell_2})$)\\
\label{fig:Cfunction}
\caption{The $\mathscr{C}$ function on types.}
\end{figure*}

\begin{figure*}
\small
\begin{mathpar}
%\text{\underline{Syntax}} \hfill \\
%\begin{array}{rcl}
%  w &::=& \dots \sep \bracket{w}{w} \\[0.4em]
%  f &::=& \dots \sep \bracket{f}{f} \\[0.4em]
%  e  &::=& \dots \sep \bracket{e}{e}
%\end{array}\hfill \\

%\text{\underline{Distibuted Evaluation rules}} \hfill \\
  
     \hfill
     \\\\
  
\Rule{Bracket-Stack}
        {
          \TValGpcc{e}{\tau'}\\
          \drflowjudge{\Pi}{pc}{pc'} \\
          \forall i \in \{1,2\}.\TValGpcS{{s_i}}{[\tau']\tau}
        }
        {\TValGpcS{\distcon{e}{c}{\bracket{s_1}{s_2}}}{\tau}}

\Rule{Bracket-Head}
        {
          \TValGpcc{\bracket{e_1}{e_2}}{\tau'}\\
          \drflowjudge{\Pi}{pc}{pc'} \\
          \TValGpcS{s}{[\tau']\tau}
        }
        {\TValGpcS{\distcon{\bracket{e_1}{e_2}}{c}{s}}{\tau}}
%   \Rule{Bracket-Empty}
%        {}
%        {\TValGpcS{\bracket{empty}{empty}[\tau]}{\tau}}
   \\\\
\end{mathpar}
\caption{Distributed Typing rules for Bracketed Expressions.}
\label{fig:distbracketTypes}
\end{figure*}

\section{Subject Reduction Related Proofs.} \label{sec:subjRedRel}

In the following proofs we use the simple and the annotated 
FLAQR syntax interchageably.

\begin{lemma}[UniqueType]\label{lemma:UniqueType}
If~{\TValGpcw{e}{\tau}}~ and~ {\TValGput{e}{\mathring{\tau}}}~then~
$\tau = \mathring{\tau}$
\end{lemma}
\begin{proof}
Straightforward by induction on typing derivation of $e$.
\end{proof}

\begin{lemma}[WaitUniqueT] \label{lemma:WaitUniqueT}
If~\TValGpcw{E[\expecta{\hat{\tau}}]}{\tau}~and~
\TValGput{E[\expecta{\hat{\tau}}]}{\tau'}~then~
$\tau=\tau'$.
\end{lemma}
\begin{proof}
Straightforward using induction over structure of $E$.
\end{proof}

\begin{lemma}[stackUniqueT]
If type of the tail 
~\TValGpcS{t}{[\hat{\tau}]\tau}~and \TValGpcS{t}{[\hat{\tau}]\tau'}~ then~
$\tau=\tau'$. 
\end{lemma}
\begin{proof}
The proof is by induction over typing derivation of s.
\end{proof}

\begin{lemma}[distUniqueT]\label{lemma:distUniqueT}
If~\TValGpcS{\distcon{e}{c}{s}}{\tau}~and~ \TValGpcS{\distcon{e}{c}{s}}{\tau'}
~then~ $\tau =\tau'$.
\end{lemma}
\begin{proof}
Straightforward proof using lemmas UniqueType \ref{lemma:UniqueType} and
\ref{lemma:WaitUniqueT}.
\end{proof}

\begin{lemma}[$\Gamma$-Weakening]
If \TValGpcw{e}{\tau} and for 
all $\tau'$ and $x \notin dom(\Gamma)$,
\TValP{\Gamma,x\ty\tau';pc;c}{e}{\tau}
\end{lemma}
\begin{proof}
By Induction on structure of $e$.
\end{proof}

\begin{lemma}[CTXif] \label{lemma:CTXif}
%PM{subject reduction uses CTX lemma, and CTX lemma needs uniquetypes}
If \TValGpcw{E[v]}{\tau} and 
$x\notin dom(\Gamma)$ then
$\exists \tau'$, such that 
\TValP{\Gamma,x\ty\tau',\pc;c}{E[x]}{\tau} and \TValGpcw{v}{\tau'}.
\end{lemma}
\begin{proof}
By induction on structure of $E$.
\end{proof}

\begin{lemma}[CTXonlyif] \label{lemma:CTXonlyif}
\TValGpcw{v}{\tau'} and
\TValP{\Gamma,x\ty\tau',\pc;c}{E[x]}{\tau} then
\TValGpcw{E[v]}{\tau}, when %$x \not\in FV(v)$ and 
$x \not\in FV(E)$($FV$ returns the free variables).
\end{lemma}
\begin{proof} Proof by induction over structure of $E$. 
\end{proof}

\begin{lemma}[CTXiff] \label{lemma:CTXiff}
\TValGpcw{E[v]}{\tau} iff 
$\exists \tau'$, such that 
\TValP{\Gamma,x\ty\tau',\pc;c}{E[x]}{\tau} and \TValGpcw{v}{\tau'}
when $x\notin dom(\Gamma)$.
\end{lemma}
\begin{proof}
Straightforward from lemma 
\ref{lemma:CTXif} and \ref{lemma:CTXonlyif}.
\end{proof}

\begin{lemma}[Expect]\label{lemma:expecta}
If \TValGpcw{E[v]}{\tau} and \TValGpcw{v}{\tau'} then 
\TValGpcw{E[\expecta{\tau'}]}{\tau}.
\end{lemma}
\begin{proof}
Straightforward using 
lemma CTXonlyif \ref{lemma:CTXonlyif}.
\end{proof}

\begin{lemma}[RExpect] \label{lemma:rexpect}
\TValGpcw{v}{\tau'} and
\TValGpcw{E[\expecta{\tau'}]}{\tau} then
\TValGpcw{E[v]}{\tau}.
\end{lemma}
\begin{proof} Straightforward from 
lemma CTXonlyif \ref{lemma:CTXonlyif}.
\end{proof}

\begin{lemma}[Clearance] \label{lemma:clearance}
If $\TValGpcw{e}{\tau}$ then $\rafjudge{\Pi}{c}{pc}$
\end{lemma}
\begin{proof}
Proof is straightforward by induction on the typing judgments.
\end{proof}

\begin{lemma}[Values PC] \label{lemma:valuespc}
%Values are typable at any host and any \pc
Let $\TValGpcw{v}{\tau}$. If  
$\rafjudge{\Pi}{c'}{pc'}$ and 
$\rafjudge{\Pi}{c'}{\UB{\tau}}$
then $\TValGpcdashpc{v}{\tau}$.
\end{lemma}
\begin{proof} Given that,
\begin{align}
\rafjudge{\Pi}{c'}{pc'}   \label{VTH} \\
\rafjudge{\Pi}{c'}{\UB{\tau}}  \label{VTH.5} 
\end{align}
Using induction over values.\\ \\
%\textbf{case FAIL.} Using \ruleref{Fail} and (\ref{VTH}) we have 
%{\TValGpcdashpc{\fail{\says{\ell}{\tau}}}{\says{\ell}{\tau}}}\\ \\
\textbf{Case \ruleref{UnitM}.} 
Using \ruleref{Unit} typing rule and (\ref{VTH}) we have  
{\TValGpcdashpc{\void}{\voidtype}}\\
\textbf{Case \ruleref{Pair}.} Given 
\begin{equation}\label{pairVTH}
{\TValGpcw{\paira{v_1}{v_2}{\prodtype{τ_1}{τ_2}}}{\prodtype{τ_1}{τ_2}}}
\end{equation}
Inverting (\ref{pairVTH}) we get 
\begin{equation}\label{pairVTH1}
\TValGpcw{v_1}{τ_1}
\end{equation} 
and
\begin{equation}\label{pairVTH2}
\TValGpcw{v_2}{τ_2}
\end{equation}
By applying induction hypothesis on (\ref{pairVTH1}) and (\ref{pairVTH2}), we get
\begin{equation}\label{pairVTH3}
\TValGpcdashpc{v_1}{τ_1}
\end{equation}
and
\begin{equation}\label{pairVTH4}
\TValGpcdashpc{v_2}{τ_2}
\end{equation}
From rule \ruleref{Pair}, (\ref{pairVTH3}), (\ref{pairVTH4}), and (\ref{VTH}) we get
{\TValGpcdashpc{\paira{v_1}{v_2}{\prodtype{τ_1}{τ_2}}}{\prodtype{τ_1}{τ_2}}}\\ \\
\textbf{Case \ruleref{Inj}.} Similar to \textbf{case \ruleref{Pair}}.\\ \\
\textbf{Case \ruleref{Sealed}.} Given 
\begin{equation}\label{sealVTH}
{\TValGpcw{\returnv{ℓ }{v}}{\says{ℓ }{τ}}}
\end{equation}
Inverting (\ref{sealVTH}) we get
\begin{equation}\label{sealVTH1}
 \TValGpcw{v}{\tau} 
\end{equation}
By applying induction hypothesis on (\ref{sealVTH1}) we get
\begin{equation}\label{sealVTH2} 
\TValGpcdashpc{v}{τ}
\end{equation}
Thus from rule [SEAL], (\ref{VTH}) and (\ref{sealVTH2}) we get 
{\TValGpcdashpc{\returnv{ℓ  }{v}}{\says{ℓ  }{τ}}}\\ \\
\textbf{Case \ruleref{Lam}.} We have 
\begin{equation}\label{lamVTH}
{\TValGpcw{\lamc{x}{τ_1}{\pc''}{e}}{\func{τ_1}{\pc''}{τ_2}}}
\end{equation}
Inverting (\ref{lamVTH}) we get
\begin{align}
\TVal{\Pi;Γ,x\ty τ_1;\pc'';u}{e}{τ_2}  \label{lamVTH1} \\
\rafjudge{\Pi}{c}{pc}  \\ 
u = \UB{\func{τ_1}{\pc''}{τ_2}} \\
\rafjudge{\Pi}{c}{\UB{\func{τ_1}{\pc''}{τ_2}}}  \\
\intertext{Applying IH on \ref{lamVTH1}}
\TVal{\Pi;Γ,x\ty τ_1;\pc'';u}{e}{τ_2}  \label{lamVTH1.0} \\
\intertext{given in lemma statement} 
\rafjudge{\Pi}{c'}{pc'}  \label{lamvth1.2} \\                                  
\rafjudge{\Pi}{c'}{\UB{\func{τ_1}{\pc''}{τ_2}}}  \label{lamvth1.3} \\ 
\intertext{from \ref{lamvth1.2}, \ref{lamvth1.3}, 
\ref{lamVTH1.0} and \ruleref{Lam} rule we get}
{\TValGpcdashpc{\lamc{x}{τ_1}{\pc''}{e}}{\func{τ_1}{\pc''}{τ_2}}} \\ \\
\end{align}
\textbf{Case \ruleref{TLam}.} Similar to \textbf{case \ruleref{Lam}}.\\ \\
\textbf{Case \ruleref{Bracket}.}
Given,
\begin{align}
{\TValGpcw{\bracket{v_1}{v_2}}{\tau}}  \label{br1} \\
\intertext{inverting \ref{br1} we get}
\protjudge{\delegcontext}{H^\pi}{\cfun{\tau}} \label{br2} \\
\TValGpcw{v_1}{\tau} \label{br3} \\      
\TValGpcw{v_2}{\tau} \label{br4} \\
\intertext{IH on \ref{br3} and \ref{br4}}
\TValGpcdashpc{v_1}{\tau} \label{br5} \\
\TValGpcdashpc{v_2}{\tau} \label{br6} \\
\intertext{and given in lemma statement,}
\rafjudge{\Pi}{c'}{pc'} \label{br7}\\
\intertext{from \ref{br5},\ref{br6}, \ref{br7} 
and \ref{br2} we get}
{\TValGpcdashpc{\bracket{v_1}{v_2}}{\tau}}
\end{align}
\end{proof}

\begin{lemma}[$\pc$ reduction] \label{lemma:pcreduction}
 Let \TValGpcw{e}{\tau}.
  For all $\pc, \pc'$, such that
  \rflowjudge{\delegcontext}{\pc'}{\pc} and 
  \rafjudge{Π}{c}{pc'}  
  then 
  \TValP{\varcontext;pc';c}{e}{\tau} holds.
\end{lemma}
\begin{proof}
Proof is by induction on the derivation of the typing judgment.
Given that,
\begin{align}
\rflowjudge{\delegcontext}{\pc'}{\pc} \label{pcred1} \\
\rafjudge{Π}{c}{pc'} \label{pcred2} 
\end{align}
\textbf{Case \ruleref{Run}.}
From the premises of \ruleref{Run} typing rule \\
${\TValGpcw{\runa{\tau}{e}{c'}}{τ}}$ \\ 
we get
\begin{align}
\TVal{\Pi;\Gamma;\pc'';c'}{e}{τ'} \label{pcredrun1} \\
\drflowjudge{\Pi}{pc}{pc''} \label{pcredrun2} \\
\rafjudge{\Pi}{c}{pc} \label{pcredrun3} \\
\rafjudge{\Pi}{c}{\UB{\tau'}}  \label{pcredrun4} \\
\intertext{where $\tau' = \says{pc''^{\integ \avail}}{\tau'}$.}
\intertext{From \ref{pcredrun5} and \ref{pcredrun2} we have }
\drflowjudge{\Pi}{pc'}{pc''} \label{pcredrun5} \\
\intertext{From \ref{pcredrun1}, \ref{pcred2}
\ref{pcredrun4}, \ref{pcredrun5} we get}
\end{align}
${\TValP{\Gamma;pc';c}{\runa{\tau}{e}{c'}}{τ}}$\\
\textbf{Case \ruleref{Ret}.}
From the premises of \ruleref{Ret} typing rule \\
${\TValGpcw{\ret{e}{c'}}{\says{pc{^{\integ \avail}}}{\tau}}}$ \\
we get
\begin{align}
\TValGpcw{e}{τ} \label{pcredret1}\\
\rafjudge{\Pi}{c'}{(\UB{\tau})} \label{pcredret2} \\
\rafjudge{\Pi}{c}{pc} \label{pcredret3} \\
\intertext{Applying Induction hypothesis to \ref{pcredret1} we get}
\TValP{\Gamma;pc';c}{e}{\tau} \label{pcredret4} \\
\intertext{From \ref{pcredret4}, \ref{pcredret2} 
and \ref{pcred2} we get}
\end{align}
${\TValP{\Gamma;pc';c}{\ret{e}{c'}}{\says{pc{^{\integ \avail}}}{\tau}}}$\\
\textbf{Case \ruleref{Compare}:} Straightforward using IH and  \ref{pcred2}.\\
\textbf{Case \ruleref{Select}:}  Straightforward using IH and  \ref{pcred2}.\\
\textbf{Case \ruleref{Bracket}:} From the premises Bracket typing rule \\ 
$\TValGpcw{\bracket{e_1}{e_2}}{\tau}$,\\
we get,
\begin{align}
 \rflowjudge{\delegcontext}{(H^\pi \sqcup \pc)}{{\pc''}} \label{pcredbr1}\\
 e_1 = v_1  \iff e_2 \ne v_2  \label{pcredbr2} \\
 \TValP{Γ;\pc'';c}{e_1}{\tau}  \label{pcredbr3} \\
 \TValP{Γ;\pc'';c}{e_2}{\tau}  \label{pcredbr4} \\
 \protjudge{\delegcontext}{H^{\pi}}{\cfun{\tau}}  \label{pcredbr5} \\
 \rafjudge{\Pi}{c}{pc}
\intertext{From and \ref{pcred1} we can write}
\rflowjudge{\delegcontext}{(H^\pi \sqcup \pc')}{(H^\pi \sqcup \pc)} \label{pcredbr6} \\
\intertext{From \ref{pcredbr6} and \ref{pcredbr1} we get}
 \rflowjudge{\delegcontext}{(H^\pi \sqcup \pc')}{{\pc''}} \label{pcredbr7}
\intertext{Thus from \ref{pcredbr7},\ref{pcredbr3},\ref{pcredbr4},\ref{pcredbr5} 
and \ref{pcred2} we can write}
\end{align}
$\TValP{\Gamma;pc';c}{\bracket{e_1}{e_2}}{\tau}$\\
\textbf{Case \ruleref{Bracket-Values}:}  Straightforward using IH and  \ref{pcred2}.\\
\textbf{Case \ruleref{Bracket-*}:}  Straightforward using IH and  \ref{pcred2}.\\
\textbf{Other Cases :}  Straightforward from \ref{pcred2} and $pc~reduction$ 
lemma in \cite{jflac}.
\end{proof}

\begin{lemma}[Variable substitution]
If \TValP{\Gamma,x:\tau';\pc;c}{e}{\tau} and \TValGpcw{v}{τ'}, then \TValGpcw{\subst{e}{x}{v}}{\tau}.
\end{lemma}
\begin{proof}
Proof is by induction on the typing derivation of $e$.
\\
\textbf{Case \ruleref{Lam}}
Given,
\begin{align}
\TValGpcw{v}{τ'} \label{99} \\
{\TValP{\Gamma,x\ty \tau';\pc;c}{\lamc{y}{τ_1}{\pc'}{e}}{\func{τ_1}{\pc'}{τ_2}}} \label{100} \\
\intertext{inverting \ref{100}}
{\TValP{Γ,x\ty \tau',y\ty τ_1;\pc';c}{e}{τ_2}} \label{101} \\
\rafjudge{\Pi}{c}{pc} \label{101.0} \\
\rafjudge{\Pi}{c}{\UB{\func{τ_1}{\pc'}{τ_2}}} \label{101.1} \\
\intertext{Applying lemma values host pc and Weakening lemma in \ref{99}} 
{\TValP{Γ,y\ty τ_1;\pc';c}{v}{τ'}} \label{103} \\
\intertext{IH, \ref{101}, \ref{103}}
{\TValP{Γ,y\ty τ_1;\pc';c}{\subst{e}{x}{v}}{τ_2}} \label{104} \\ 
\intertext{from \ref{104}, \ref{101.0}, \ref{101.1} and \ruleref{Lam} rule we get}
\TValGpcw{\lamc{y}{τ_1}{\pc'}{\subst{e}{x}{v}}}{\tau_2} \\
\end{align}
\textbf{Case \ruleref{Bracket}} 
Given,
$\TValP{Γ,x\ty τ_1;\pc;c}{\bracket{e_1}{e_2}}{\tau}$
  We have to prove t
  \[
  \TValP{Γ,x\ty τ_1;\pc;c}{\bracket{e_1[x\mapsto 
  \outproj{{v}}{1}]}{e_2[x \mapsto \outproj{{v}}{2}]}}{\tau}
  \]

 We first describe the case \ruleref{Bracket}. 
The proof for the case \ruleref{Bracket-Values} is analogous.
  From \ruleref{Bracket}, we have 
  \begin{align}
           \TValP{Γ,x\ty τ_1;\pc';c}{e_1}{\tau} \label{eq:dsubst1} \\
	   \TValP{Γ,x\ty τ_1;\pc';c}{e_2}{\tau} \label{eq:dsubst2} \\
	   \rflowjudge{\delegcontext}{(H^\pi \sqcup \pc)}{{\pc'}} \\
	   \protjudge{\delegcontext}{H^{\pi}}{\cfun{\tau}}
  \end{align}
Applying clearance (Lemma~\ref{lemma:clearance}), we have 
\rafjudge{\delegcontext}{c}{\pc'}.
Depending on whether ${v}$ is a bracket value, we have two cases
\begin{description}

\item{Case ${v} = \bracket{{v}_1}{{v}_2}$:}

  From \ruleref{Bracket-Values}, we have 
  \begin{align}
           \TValP{Γ,x\ty τ_1;\pc;c}{{v}_1}{\tau'} \\
	   \TValP{Γ,x\ty τ_1;\pc;c}{{v}_2}{\tau'} \\
	   \protjudge{\delegcontext}{H^{\pi}}{\cfun{\tau'}}
  \end{align}

Since values can be typed under any 
$\pc$ which acts for the host under which value is typed
(Lemma~\ref{lemma:valuespc}), we have 
\TValP{Γ,x\ty τ_1;\pc';c}{{v}_i}{\tau'} for $i =\{1,2\}$.
Applying induction to the premises 
\eqref{eq:dsubst1} and \eqref{eq:dsubst2}, we get
  \begin{align}
         \TValP{Γ,x\ty τ_1;\pc';c}{e_1[x \mapsto {v}_1]}{\tau} \\
	\TValP{Γ,x\ty τ_1;\pc';c}{e_2 [x \mapsto {v}_2]}{\tau}
  \end{align}
and thus from [Bracket] we have
\[
\TValP{Γ,x\ty τ_1;\pc;c}
{\bracket{e_1[x \mapsto {v}_1]}{e_2[x \mapsto {v}_2]}}{\tau}
\]

\item{Case ${v} \ne \bracket{{v}_1}{{v}_2}$:}
Since values can be typed under any $\pc$ 
which acts for the place under which value 
is typed (Lemma~\ref{lemma:valuespc}), we have
\TValP{Γ,x\ty τ_1;\pc';c}{{v}}{\tau'}.
Applying induction to the premises 
\eqref{eq:dsubst1} and \eqref{eq:dsubst2}, we get
  \begin{align}
         \TValP{Γ,x\ty τ_1;\pc';c}{e_1[x \mapsto {v}]}{\tau} \\
	 \TValP{Γ,x\ty τ_1;\pc';c}{e_2 [x \mapsto {v}]}{\tau}
  \end{align}
and thus from [Bracket] rule 
\[
\TValP{Γ,x\ty τ_1;\pc;c}{\bracket{e_1[x \mapsto {v}]}
{e_2[x \mapsto {v}]}}{\tau}
\]
\end{description}
\textbf{Case \ruleref{Compare}:} Straightforward using IH. \\
\textbf{Case \ruleref{Select}:} Straightforward using IH. \\
\textbf{Case \ruleref{Run}:} Straightforward using IH and 
values pc lemma. (Same as \ruleref{Lam} case.) \\
\textbf{Case \ruleref{Ret}:} Straightforward using IH.
\end{proof}

\begin{lemma}[Type Substitution]\label{lemma:tsubst}
	Let $\tau'$ be well-formed in $\varcontext, X, \varcontext'$.
	If $\TValP{\varcontext, X, \varcontext';\pc;c}{e}{\tau}$ 
then $\TValP{\varcontext, \varcontext'[X \mapsto \tau'];\pc;c}{e[X \mapsto \tau']}{\tau [X \mapsto \tau']}$. 
\end{lemma}
\begin{proof}
	Proof is by the induction on the typing derivation of 
$\TValP{\varcontext, X, \varcontext';\pc;c}{e}{\tau}$.
\end{proof}

\begin{lemma}[Soundness]\label{lemma:sound}
        If  $e \stepsone e'$ then
        $\outproj{e}{k} \stepsto \outproj{e'}{k}$ for $k ∈ \{1,2\}$.
 \end{lemma}
\begin{proof}
By induction on the evaluation of $e$.   \\
\textbf{Case \ruleref{B-CompareCommon}:}
From rule \ruleref{B-CompareCommon} we can say if \\
${\comparea{\says{(\txcmp{\ell_1}{\ell_2})}{\tau}}
{\bracket{f_{11}}{f_{12}}}{\bracket{f_{21}}{f_{22}}}} $ 
$\stepsone f$ \\ 
then for $k \in \{1,2\}$ \\
$\outproj{\comparea{\says{(\txcmp{\ell_1}{\ell_2})}{\tau}}
{\bracket{f_{11}}{f_{12}}}{\bracket{f_{21}}{f_{22}}}}{k} $ \\
$\stepsone \outproj{f}{k}$ \\
\\
\textbf{Case B-Compare*:} Similar to the case above.
\\
\textbf{Case B-Select*:} Similar to the case above.
\\
\textbf{Case \ruleref{B-Step}:}
$\outproj{e}{i} \stepsone \outproj{e'}{i}$  
and $\outproj{e}{j} = \outproj{e'}{j}$ .\\
\textbf{Other Cases}
All other cases in Figure~\ref{fig:brackets}
only expand brackets
So, $\outproj{e}{k} = \outproj{e'}{k}$ for $k \in \{1,2\}$. 
\end{proof}

\begin{lemma}[Stuck expressions]
\label{lemma:localstuck}
	If  $e$ gets stuck then $\outproj{e}{i}$ 
is stuck for some $i \in \{1,2\}$.
\end{lemma}
\begin{proof}
	By induction on the structure of $e$.
\begin{description}
\item [Case: $\comparea{\tau}{e_1}{e_2}$:] 
\ruleref{E-Compare} can not be applied, i.e. 
$e_1$ and/or $e_2$ are/is not of the form 
$\returnv{\ell}{w}$.That means 
\ruleref{E-Compare} can not applied to 
$\outproj{\comparea{\tau}{e_1}{e_2}}{i}$. 
From I.H. either $\outproj{e_1}{i}$ is stuck
or  $\outproj{e_2}{i}$ is stuck for $i \in \{1,2\}$.\\

[\ruleref{E-CompareFail}*] can not be applied, i.e. 
$e_1$ and/or $e_2$ are/is not of the form 
$\returnv{\ell}{w}$ or $\faila{\tau'}$.That means 
[E-COMPARE*] can not applied to 
$\outproj{\comparea{\tau}{e_1}{e_2}}{i}$. 
From I.H. either $\outproj{e_1}{i}$ is stuck
or  $\outproj{e_2}{i}$ is stuck for $i \in \{1,2\}$.\\

[B-COMPARE*] can not be applied, i.e. 
[E-COMPARE*] not 
be applied to 
$\outproj{\comparea{\tau}{e_1}{e_2}}{i}$ 
for $i \in \{1,2\}$. This means 
$\outproj{\comparea{\tau}{e_1}{e_2}}{i}$
is stuck for $i \in \{1,2\}$. 
From I.H. we can say 
either $\outproj{e_1}{i}$ is stuck
or  $\outproj{e_2}{i}$ is stuck for $i \in \{1,2\}$.

\item[Case: $\selecta{e_1}{e_2}{\tau}$:] Same as "compare". 

\item[Case: $\ret{e}{c}$:] From \ruleref{E-RetStep}
rule we can say if $\ret{e}{c}$ is stuck 
then $e$ is stuck. 
From I.H. we can say $e$ stuck only when 
$\outproj{e}{i}$ is stuck for $i \in \{1,2\}$.  

\item[Case: $\return{\ell}{e}$:] 
\ruleref{E-Sealed} step can not be taken. So \ruleref{E-Sealed}
step can not be taken for $\outproj{\return{\ell}{e}}{i}$.
Which means $\outproj{e}{i}$ is stuck for $i \in \{1,2\}$.

\ruleref{BFail2} and \ruleref{B-Fail1} steps can not be taken. 
Which means $e$ is not of the forms 
$\bracket{v}{\faila{\tau}}$
or $\bracket{\faila{\tau}}{v}$. 
Which again, from I.H., implies either $\outproj{e}{1}$ 
is stuck or $\outproj{e}{2}$ is stuck.
 \item[Case $\proj{j}{e}$:]  Similar to the above case.
                \item[Case $\inj{j}{e}$:] Similar to the above case.
                \item[Case $\pair{e}{e}$:] Similar to the above case.
                \item[Case $\casexp{e}{x}{e_1}{e_2}$:]
                    Since \ruleref{B-Case}, and \ruleref{E-Case} are not applicable, it follows that $e$ is not of the form
                    $\bracket{v}{v'}$, or $\inj{j}{v}$. It follows that $\outproj{\casexp{e}{x}{e_1}{e_2}}{i}$ is also stuck.
                \item[Case $\bind{x}{v}{e'}$:] Similar to the above case.
\end{description}
\end{proof}

\begin{lemma}[Completeness]
  \label{lemma:complete}
  If $\bone{e} \stepsto v_1 $ and $\btwo{e} \stepsto v_2$, then there exists some $v$ such that $e \stepsto v$.
\end{lemma}
\begin{proof}
The rules in Figure~\ref{fig:brackets} move brackets out of subterms, and
therefore can only be applied a finite number of times.  Therefore, by
Lemma~\ref{lemma:sound}, 
if $e$ diverges, either $\bone{e}$ or $\btwo{e}$
diverge; this contradicts our assumption.

Furthermore, by Lemma~\ref{lemma:localstuck}, 
if the evaluation of $e$ gets stuck, either
$\bone{e}$ or $\btwo{e}$ gets stuck. 
Therefore, since we assumed $\bgeti{e} \stepsto v_i$, then 
$e$ must terminate. Thus, there exists 
some $v$ such that $e \stepsto v$.
\end{proof}

\begin{lemma}[Soundness for global bracketed semantics] \label{lemma:sounddist}
        If  $~~~~~~~~~~~~~~~~~~$ 
$\distcon{e}{c}{t} \Longrightarrow \distcon{e'}{c'}{t'}$ then
        $\outproj{\distcon{e}{c}{t}}{k} \Longrightarrow^{*} \outproj{\distcon{e'}{c'}{t'}}{k}$ for $k \in \{1,2\}$.
\end{lemma}
\begin{proof}
If $c=c'$ and $t=t'$ then the lemma statement holds
following lemma \ref{lemma:sound}. 
The cases where the stack and the head of the configuration 
changes are the interesting ones. 
\begin{description}
\item[Case \ruleref{B-RetV}:]  Straightforward from [B-RetV] evaluation rule. 
The premise of the rule says $\forall k \in \{1,2\}$,
the $k$th projection of the expression should take a step.
\item[Case \ruleref{B-Ret*}:] Same as the above case.
\item[Case \ruleref{B-RunLeft}:]  The lemma trivially holds for $k=1$.
\item[Case \ruleref{B-RunRight}:] The lemma trivially holds for $k=2$. 
\end{description}
\end{proof}

\begin{lemma}[Dist Stuck expressions] \label{stuckdist}
\label{lemma:stuck}
	If  $\distcon{e}{c}{t}$ gets stuck 
then $\outproj{\distcon{e}{c}{t}}{k}$ 
is stuck for some $k \in \{1,2\}$.
\end{lemma}
\begin{proof}
Induction over structure of $e$.
\begin{description}
\item [Case $\distcon{\ret{e}{c'}}{c}{\stackapp{E[\expecta{\tau}]}{c'}{t}}$:]
This means [E-RetV] and [B-RET*] steps can not be taken. 
So, [E-RetV] and [B-RetV] steps can not be taken for \\
$\outproj{\distcon{\ret{e}{c'}}{c}{\stackapp{E[\expecta{\tau}]}{c'}{t}}}{i}$.
This means $e$ is not of the form $v$ or $\bracket{f_1}{f_2}$.
From I.H. it is clear that $\outproj{v}{i}$ or $f_i$ is stuck 
for $i \in \{1,2\}$.

\item [Case $\distcon{E[\runa{\tau}{e}{c'}]}{c}{t}$] 
This can always take
a \ruleref{E-Run} step and then can get stuck. 
In that case the argument is 
same as case $\ret{e}{c}$ as in lemma \ref{lemma:localstuck}.

\item [Case $\distcon{\bracket{E[\runa{\tau}{e}{c'}]}{e'}}{c}{t}$] 
Can not get stuck as $E[\runa{\tau}{e}{c'}]$ run 
can always take a step.

\item [Case $\distcon{\bracket{e'}{E[\runa{\tau}{e}{c'}]}}{c}{t}$] 
Can not get stuck as $E[\runa{\tau}{e}{c'}]$ run 
can always take a step.

\item [Other Cases:] In all other cases the active configuration
and the stack does not change. So lemma statement 
holds following lemma \ref{lemma:localstuck}. 

\end{description}
\end{proof}

\begin{lemma}[Completeness]
  \label{lemma:completedist}
  If $\bone{\distcon{e}{c}{t}} \stepsto \distcon{v_1}{c}{\emptystack} $ 
and $\btwo{\distcon{e}{c}{t}} \stepsto \distcon{v_2}{c}{\emptystack}$, 
then there exists some $v$ such that $\distcon{e}{c}{t} \stepsto \distcon{v}{c}{\emptystack}$.
\end{lemma}
\begin{proof}
Similar argument as \ref{lemma:complete}
\end{proof}

\begin{lemma}[Label Flowsto SelCmp] \label{lemma:lflowstoselcmp}
If $\drflowjudge{\Pi}{\ell}{\ell_1}$ and 
$\drflowjudge{\Pi}{\ell}{\ell_2}$ then 
$\drflowjudge{\Pi}{\ell}{(\txsel{\ell_1}{\ell_2})}$ and
$\drflowjudge{\Pi}{\ell}{(\txcmp{\ell_1}{\ell_2})}$
%(We consider $\selor{}{}$ to be $\vee$.)
\end{lemma}
\begin{proof}
Given,
\begin{align}
\drflowjudge{\Pi}{\ell}{\ell_1} \\
\drflowjudge{\Pi}{\ell}{\ell_2} \\
\intertext{which implies}
\rafjudge{\Pi}{\ell^{\integ}}{\ell_1^{\integ}} \label{scmp1} \\
\rafjudge{\Pi}{\ell^{\integ}}{\ell_2^{\integ}} \label{scmp2} \\
\rafjudge{\Pi}{\ell^{\avail}}{\ell_1^{\avail}} \label{scmp3} \\
\rafjudge{\Pi}{\ell^{\avail}}{\ell_2^{\avail}} \label{scmp4} \\
\rafjudge{\Pi}{\ell_1^{\confid}}{\ell^{\confid}} \label{scmp5} \\
\rafjudge{\Pi}{\ell_2^{\confid}}{\ell^{\confid}} \label{scmp6} \\
\intertext{\ref{scmp1}, \ref{scmp2} and \ruleref{R-ConjR} implies}
\rafjudge{\Pi}{\ell^{\integ}}
{\ell_1^{\integ} \wedge \ell_2^{\integ}} \label{scmp7} \\
\intertext{from \ref{scmp1}, \ref{scmp2} and \ruleref{PAandR}}
\rafjudge{\Pi}{\ell^{\integ}}{\comor{\ell_1^{\integ}}{\ell_2^{\integ}}}
\intertext{\ref{scmp5} and \ruleref{ConjL} implies}
\rafjudge{\Pi}{\ell_1^{\confid} \wedge \ell_2^{\confid}}
{\ell^{\confid}} \label{scmp8} 
\\
\intertext{\ref{scmp3} and \ruleref{DisjR} implies}
\rafjudge{\Pi}{\ell^{\avail}}{\ell_1^{\avail} \vee \ell_2^{\avail}} \label{scmp9}
\end{align}
\ref{scmp7}, \ref{scmp8} and \ref{scmp9}
together proves 
$\drflowjudge{\Pi}{\ell}{(\txcmp{\ell_1}{\ell_2})}$\\
Similarly we can prove 
$\drflowjudge{\Pi}{\ell}{(\txsel{\ell_1}{\ell_2})}$\\
\end{proof}

\begin{lemma}[Projection Preserves Types]\label{lemma:dproj}
	If \TValGpc{e}{\tau}, then \TValGpc{\outproj{e}{i}}{\tau} for $i = \{1, 2\}$.
\end{lemma}
\begin{proof}
	Proof is by induction on the typing derivation of \TValGpc{e}{\tau}. The interesting case is $e = \bracket{e₁}{e₂}$.  By \ruleref{Bracket}, we have  
\TValP{\Gamma;pc'}{eᵢ}{\tau} for some $\pc'$ such that \rflowjudge{\delegcontext}{(H^\pi \sqcup \pc^\pi)}{{\pc'^\pi}}.
      Therefore, by Lemma~\ref{lemma:pcreduction}($pc$ reduction), 
we have \TValP{\Gamma;pc}{eᵢ}{\tau}. 
\end{proof}

\begin{lemma}[Subject Reduction(within a host)] \label{subjRedhost}
Let  \TValGpcw{e}{τ} and \rafjudge{\Pi}{c}{\UB{\tau}}.  
If $e \stepsone e'$
%\rafjudge{\delegcontext}{c}{\pc}  
then \TValGpcw{e'}{τ}.
\end{lemma}
\begin{proof} 
\textbf{Case \ruleref{E-App}}
 Given $e = (\lamc{x}{\tau_1}{\pc'}{e})~v$ and 
$e' = e[x \mapsto v]$. Also $\fexp{\delegcontext}
{\varcontext}{\pc;c}{\lamc{x}{\tau_1}{\pc'}{e}~v}{\tau_2}$.
			From the premises of \ruleref{App}, we have:
			\begin{align}
				\fexp{\delegcontext}{\varcontext}{\pc;c}{\lamc{x}{\tau_1}{\pc'}{e}}{\func{τ_1}{\pc'}{τ_2}} \label{eq:subapp1}\\
	\fexp{\delegcontext}{\varcontext}{\pc;c}{v}{τ_1} \label{eq:subapp2} \\
	\rflowjudge{\delegcontext}{\pc}{\pc'} \label{eq:subapp3}
			\end{align}
			From \eqref{eq:subapp1}, we further have that $\fexp{\delegcontext}{\varcontext, x:\tau_1}{\pc';c}{e}{\tau_2}$.
			Since \eqref{eq:subapp3} holds, 
we can now apply PC reduction to get 
$\fexp{\delegcontext}{\varcontext, x:\tau_1}{\pc}{e}{\tau_2}$.
Applying substitution preservation 
using \eqref{eq:subapp2},  we have 
$\fexp{\delegcontext}{\varcontext}{\pc;c}{e[x \mapsto v]}{\tau_2}$.\\

\textbf{Case \ruleref{E-BindM}.} Given $e={\bind{x}{\returnv{ℓ}{v}}{e_1}}$ and 
$e'= \subst{e_1}{x}{v}$ and also
\begin{equation}\label{bind1}
{\TValGpcw{\bind{x}{\returnv{ℓ}{v}}{e_1}}{τ}}
\end{equation}
From the premises of (\ref{bind1}) we have
\begin{align}
  \TValGpcw{\returnv{ℓ}{v}}{\says{\ell}{\tau'}} \label{eq:sbindm0} \\
  \TValGpcw{v}{\tau'} \label{eq:sbindm1} \\
  \TVal{\Pi;\Gamma,x:\tau';\pc \sqcup \ell ;c}
    {e_1}{\tau} \label{sbindm2} \\
  \protjudge{\delegcontext}{\pc \sqcup \ell}{\tau} \label{eq:sbindm4} \\
  \rafjudge{\delegcontext}{c}{\pc}
\end{align}
(Since $\rafjudge{\delegcontext}{c}{\pc}$) applying $\pc$ reduction lemma in 
(\ref{sbindm2}) we get $\TValP{\Gamma, x:\tau'; \pc}{e_1}{\tau} \label{eq:sbindm2}$.
Invoking variable substitution lemma, we thus have 
${\TValGpcw{\subst{e_1}{x}{v}}{\tau}}$.
\\ \\
\textbf{Case \ruleref{E-RetStep}}
\begin{equation}\label{retcon}
{\ret{e_1}{c'}} \stepsone {\ret{e_1'}{c'}}
\end{equation}
Given, $e=\ret{e_1}{c'}$ and $e'=\ret{e_1'}{c'}$ and also
\begin{equation}\label{ret1}
{\TValGpcw{\ret{e_1}{c'}}{\tau}}
\end{equation} 
From the premises of (\ref{ret1}) we get
\begin{align}
{\TValGpcw{e_1}{\tau'}}  \label{ret1pr1} \\
\intertext{where $\tau = \says{pc^{{\integ}{\avail}}}{\tau'}$}
\rafjudge{\uppi}{c}{pc}  \label{ret1pr2} \\
\rafjudge{\uppi}{c'}{\UB{\tau'}} \label{ret1pr3}
\end{align}
and applying induction hypothesis on the premise of (\ref{retcon})
we get 
\begin{equation}\label{ret1IH}
{\TValGpcw{e_1'}{\tau'}}
\end{equation}
From (\ref{ret1IH}), (\ref{ret1pr2}) and (\ref{ret1pr3}) 
we have ${\TValGpcw{\ret{e_1'}{c'}}{\tau}}$ \\ \\
\textbf{\ruleref{E-Compare}} Given,
$e=(\comparea{\says{\txcmp{\ell_1}{\ell_2}}{\tau}}
{\returnv{\ell_1}{v}}{\returnv{\ell_2}{v}})$ and  
$e'=\returnv{\txcmp{\ell_1}{\ell_2}}{v}$
and also, 
\begin{equation}\label{ce1}
\TValGpcw{\comparea{\says{(\txcmp{\ell_1}{\ell_2})}{\tau}}
{\returnv{\ell_1}{v}}{\returnv{\ell_2}{v}}}
{\says{(\txcmp{\ell_1}{\ell_2})}{\tau}} 
\end{equation}
Inverting (\ref{ce1})
\begin{align}
\TValGpcw{\returnv{\ell_1}{v}}{\says{\ell_1}{τ}} \label{ce2}\\
\TValGpcw{\returnv{\ell_2}{v}}{\says{\ell_2}{τ}} \label{ce3}\\
c^{\confid} \rhd \says{\ell_1}{\tau} \label{ce4}\\
c^{\confid} \rhd \says{\ell_2}{\tau}  \label{ce5}\\
\rafjudge{\Pi}{c}{pc} \label{ce6}\\
%\stflowjudge*{pc} \label{ce7}\\
%\stflowjudge*{pc}  \label{ce7}\\
\end{align}
Inverting (\ref{ce2}) or (\ref{ce3}) further, we get
\begin{equation}\label{ce8}
\TValGpcw{v}{\tau}
\end{equation}
From rule \ruleref{Sealed}, (\ref{ce8}) and (\ref{ce6}) we can say
$\TValGpcw{\returnv{\txcmp{\ell_1}{\ell_2}}{v}}
{\says{(\txcmp{\ell_1}{\ell_2})}{\tau}}$ \\ \\
\textbf{Case \ruleref{E-CompareFail}*.} Trivial, as $\faila{\tau}$ typechecks
with any protected type, and based on our type-system $\tau$ is always 
a protected type.\\ \\
\textbf{Case \ruleref{E-Select}.}
Given,
$e = \selecta{\returnv{\ell_1}{v_1}}{\returnv{\ell_2}{v_2}}
{\says{\ell}{\tau}}$ 
and 
$e'= \returnv{\ell}{v_1}$ 
where $\ell = \txsel{\ell_1}{\ell_2}$.
The following is also given. 
\begin{equation}\label{sel1}
\TValGpcw{\selecta{\returnv{\ell_1}{v_1}}{\returnv{\ell_2}{v_2}}
{\says{\ell}{\tau}}}{\says{\ell}{τ}}
\end{equation}
Inverting (\ref{sel1}) we get the following,
\begin{align}
\TValGpcw{\returnv{\ell_1}{v_1}}{\says{\ell_1}{τ}} \label{se2}\\
\TValGpcw{\returnv{\ell_2}{v_2}}{\says{\ell_2}{τ}} \label{se3}\\
\rafjudge{Π}{c}{pc} \label{se4}\\
%\stflowjudge*{\pc} \label{se5} \\
%\stflowjudge*{\pc} \label{se6}
\end{align}
Further inverting (\ref{se2}) we get,
\begin{equation}\label{se7}
\TValGpcw{v_1}{\tau}  
\end{equation} 
Thus from rule \ruleref{Sealed}, (\ref{se7}) and (\ref{se4}) we can argue,
$\TValGpcw{\returnv{\ell}{v_1}}{\says{\ell}{τ}}$ \\ \\

\textbf{Case \ruleref{E-SelectL},\ruleref{E-SelectR}.} 
Similar to above. The only diffrence is we need to
invert both (\ref{se2}) and (\ref{se3}) 
and argue both $\TValGpcw{\returnv{\ell}{v_1}}{\tau}$
and $\TValGpcw{\returnv{\ell}{v_2}}{\tau}$ holds.\\ \\                    

\textbf{Case \ruleref{E-SelectFail}.} Trivial, as $\faila{\tau}$ typechecks for any
protected type, and based on our type-system $\tau$ is always
a protected type. \\

\textbf{Case \ruleref{B-SelectCommon}.}
This case has a number of sub-cases based on whether $f_{ij}$ 
is a value or a $\fail{}$ term. We will show few cases.
Proof of the rest of the combinations will be
similar. \\
Given,
$e = \select{\bracket{\returnv{\ell_1}{v_1}}
{\returnv{\ell_1}{v_1}}}
{\bracket{\returnv{\ell_2}{v_1'}}{\returnv{\ell_2}{v_2'}}} $\\
and 
$e' = \bracket{\returnv{(\txsel{\ell_1}{\ell_2})}{v_1}}
{\returnv{(\txsel{\ell_1}{\ell_2})}{v_2}}$ \\

$\TValGpcw{e}{\says{(\txsel{\ell_1}{\ell_2})}{\tau}}$ \label{bselc1} \\
Inverting the above we get,
\begin{align}
\TValGpcw{\bracket{\returnv{\ell_1}{v_1}}{\returnv{\ell_1}{v_2}}}
{\says{\ell_1}{τ}} \label{bselc2}\\
\TValGpcw{\bracket{\returnv{\ell_2}{v_1'}}{\returnv{\ell_2}{v_2'}}}{\says{\ell_2}{τ}} \label{bselc3}\\
\rafjudge{Π}{c}{pc} \label{se4}\\
\intertext{inverting \ref{bselc2} and \ref{bselc3} we get}
\flowjudge{\Pi}{H^\pi}{\cfun{\says{\ell_1}{τ}}} \label{bsel4} \\
\flowjudge{\Pi}{H^\pi}{\cfun{\says{\ell_2}{τ}}} \label{bsel5} \\
\drflowjudge{\Pi}{(H^{\pi} \join pc)}{pc'} \label{bsel6} \\
\TValGpcc{\returnv{\ell_1}{v_1}}{\says{\ell_1}{τ}} \label{bsel6.5} \\
\TValGpcc{\returnv{\ell_1}{v_2}}{\says{\ell_1}{τ}} \label{bsel6.6} \\
\TValGpcc{\returnv{\ell_2}{v_1'}}{\says{\ell_2}{τ}} \label{bsel6.7} \\
\TValGpcc{\returnv{\ell_2}{v_2'}}{\says{\ell_2}{τ}} \label{bsel6.8} \\
\intertext{From \ref{bsel4} and \ref{bsel5} and lemma 
\ref{lemma:lflowstoselcmp} we get}
\flowjudge{\Pi}{H^\pi}{\cfun{\says{(\txsel{\ell_1}{\ell_2})}{τ}}} \label{bsel7} \\ 
\intertext{applying \ruleref{UnitM} in \ref{bsel6.6}
\ref{bsel6.6}, \ref{bsel6.7} , 
\ref{bsel6.8} we get:}
\TValGpcc{v_1}{\tau} \label{bsel8} \\
\TValGpcc{v_2}{\tau} \label{bsel9} \\
\intertext{applying pc-reduction lemma , \ruleref{Sealed} in \ref{bsel8} and 
\ref{bsel9} we get}
\TValGpcc{\returnv{(\txsel{\ell_1}{\ell_2})}{v_1}}
{\says{(\txsel{\ell_1}{\ell_2})}{\tau}} \label{bsel10} \\
\TValGpcc{\returnv{(\txsel{\ell_1}{\ell_2})}{v_2}}
{\says{(\txsel{\ell_1}{\ell_2})}{\tau}} \label{bsel11} 
\end{align} 
{from \ref{bsel10}, \ref{bsel11} ,
\ref{bsel7} and \ruleref{Bracket-Values} we get}
$\TValGpcw{\bracket{\returnv{\txsel{\ell_1}{\ell_2}}{v_1}}
{\returnv{(\txsel{\ell_1}{\ell_2})}{v_2}}}
{\says{(\txsel{\ell_1}{\ell_2})}{\tau}}$
Let us do another case, where\\
$e= \select{\bracket{\faila{\says{\ell_1}{\tau}}}
{\returnv{\ell_1}{v'}}}
{\bracket{\returnv{\ell_2}{v}}
{\faila{\says{\ell_2}{\tau}}}}$\\
and \\
$e' = \bracket{\returnv{\txsel{\ell_1}{\ell_2}}{v}}
{\returnv{\txsel{\ell_1}{\ell_2}}{v'}}$\\
Similar proof as above by inverting using \ruleref{Bracket} rule
on \\ 
$\bracket{\faila{\says{\ell_1}{\tau}}}
{\returnv{\ell_1}{v'}}$ and 
$\bracket{\returnv{\ell_2}{v}}
{\faila{\says{\ell_2}{\tau}}}$\\

\textbf{Case \ruleref{B-SelectCommonLeft}.}\label{case:bselectcomright}\\
$e= \select{\bracket{\faila{\says{\ell_1}{\tau}}}
{\returnv{\ell_1}{v}}}
{\faila{\says{\ell_2}{\tau}}}$\\
and \\
$e' = \bracket{\faila{\says{\txsel{\ell_1}{\ell_2}}{\tau}}}
{\returnv{\txsel{\ell_1}{\ell_2}}{v}}$\\
Proof is straightforward using \ruleref{Bracket} to invert 
$\bracket{\faila{\says{\ell_1}{\tau}}}
{\returnv{\ell_1}{v}}$ and then using 
\ruleref{Bracket-Fail-L} rule to prove the conclusion. \\
Let us prove another case where \\
$e= \select{\bracket{\returnv{\ell_1}{v_1}}
{\returnv{\ell_1}{v_2}}}
{\returnv{\ell_2}{v_3}}$\\
and \\
$e' = \bracket{\returnv{\txsel{\ell_1}{\ell_2}}{v_1}}
{\returnv{\txsel{\ell_1}{\ell_2}}{v_2}}$\\
and $\pi = "a"$ \\
Proof is straightforward using \ruleref{Bracket-Values} to invert \\
$\bracket{\returnv{\ell_1}{v_1}}
{\returnv{\ell_1}{v_2}}$ and then using 
\ruleref{Bracket-Fail-A} rule to prove the conclusion. \\

\textbf{Case \ruleref{B-CompareCommon}:}\\
Similar to \ruleref{B-SelectCommon}. \\

\textbf{Case \ruleref{B-CompareCommonLeft}}\\
Similar to \ruleref{B-SelectCommonLeft}. \\

\textbf{Case \ruleref{B-Fail1} and \ruleref{B-Fail2}}\\
Straightforward using \ruleref{Bracket-Fail-L} 
\ruleref{Bracket-Fail-R} rules.

\textbf{Case \ruleref{B-Step}}
Straightforward using Induction Hypothesis.\\

\textbf{Case \ruleref{B-App}:} 
         Given $e = \bracket{v_1}{v_2}~v'$ and $e'= \bracket{v_1~\outproj{v'}{1}}{v_2~\outproj{v}{2}}$. Also given that $\TValGpc{ \bracket{v_1}{v_2}~v'}{\tau_2}$ is well-typed, from \ruleref{App}, we have the following:
         \begin{align}
                 & \TValGpc{ \bracket{v_1}{v_2}}{\func{\tau_1}{\pc''}{\tau_2}}  \label{eq:bappsubred1} \\
                & \TValGpc{v'}{\tau_1}  \label{eq:bappsubred2}\\
                & \rflowjudge{\delegcontext}{\pc}{\pc''}  \label{eq:bappsubred3}
         \end{align}
         Thus from \ruleref{Bracket-Values}, we have $\rflowjudge{\delegcontext}{H^\pi}{\cfun{(\func{\tau_1}{\pc''}{\tau_2})^\pi}}$. 
That is, from the definition of type protection 
(Figure~\ref{fig:protect}), 
we have $\rflowjudge{\delegcontext}{H^\pi}{\cfun{\func{\tau_1}{{\pc''}^\pi}{\tau_2^\pi}}}$.  From \ruleref{P-Fun}, we thus have
         \begin{align}
                & \rflowjudge{\delegcontext}{H^\pi}{\tau_2^\pi} \label{eq:bappsubred4} \\ 
                & \rflowjudge{\delegcontext}{H^\pi}{{\pc''}^\pi}  \label{eq:bappsubred5}
         \end{align}
         We need to prove
         \[
          \TValGpc{\bracket{v_1~\outproj{v'}{1}}{v_2~\outproj{v'}{2}}}{\tau_2}
         \]
         That is we need the following premises of \ruleref{Bracket}.
         \begin{align}
                \TValP{\Gamma;\pc'}{v_1~\outproj{v'}{1}}{\tau_2}  \label{eq:bappsubred6} \\
                \TValP{\Gamma;\pc'}{v_2~\outproj{v'}{2}}{\tau_2}   \label{eq:bappsubred7}\\
                \rflowjudge{\delegcontext}{H^\pi \sqcup \pc^\pi}{{\pc'}^\pi}  \label{eq:bappsubred8} \\
                \rflowjudge{\delegcontext}{H^\pi }{\cfun{\tau_2^\pi}}  \label{eq:bappsubred9}
         \end{align}

         Let $\pc' = \pc''$.
         We have  \eqref{eq:bappsubred8} from   \eqref{eq:bappsubred3} and \eqref{eq:bappsubred5}.
         We already have  \eqref{eq:bappsubred9} from  \eqref{eq:bappsubred4}.
         To prove \eqref{eq:bappsubred6}, we need the following premises:
            \begin{align}
                \TValP{\Gamma;\pc'}{v_1}{\func{\tau_1}{\pc''}{\tau_2}}  \label{eq:bappsubred10} \\
                \TValP{\Gamma;\pc'}{\outproj{v'}{1}}{\tau_2}  \label{eq:bappsubred11} \\
                \rflowjudge{\delegcontext}{\pc'}{\pc''} \label{eq:bappsubred12}
            \end{align}
            The last premise \ \eqref{eq:bappsubred12} holds trivially (from reflexivity).
            Applying Lemma~\ref{lemma:valuespc} (values can be typed under any \pc) to  \eqref{eq:bappsubred1} we have \eqref{eq:bappsubred10}.
            Applying Lemma~\ref{lemma:valuespc} (values can be typed under any \pc) and  Lemma~\eqref{lemma:dproj} (projection preserves typing) to \eqref{eq:bappsubred2} we have \eqref{eq:bappsubred11}.
            Thus from \ruleref{App}, we have  \eqref{eq:bappsubred6}. Similarly,  \eqref{eq:bappsubred7} holds.
            Hence proved.

\textbf{Case \ruleref{B-TApp}} Similar to \ruleref{B-App}.\\
\if 0
\textbf{Case {B-Unpair}}
		Given $e =\proj{i}{\bracket{\pair{v_{11}}{v_{12}}}{\pair{v_{21}}{v_{22}}}}$ and $e'  = \bracket{v_{1i}}{v_{2i}}$.
		Also 
			\TValGpc{\proj{i}{\bracket{\pair{v_{11}}{v_{12}}}{\pair{v_{21}}{v_{22}}}}}{\tau_i} 
		We have to prove
		\[
			\TValGpc{\bracket{v_{1i}}{v_{2i}}}{\tau_i} 
		\]
		From \ruleref{UnPair}, we have:
		\begin{align}
			\TValGpc{\bracket{\pair{v_{11}}{v_{12}}}{\pair{v_{21}}{v_{22}}}}{\prodtype{\tau_1}{\tau_2}} \label{eq:dbunpair1}
		\end{align}
                Since they are already values, they can be inverted using \ruleref{Bracket-Values}. This approach is more conservative.
		\begin{align}
			\TValGpc{\pair{v_{11}}{v_{12}}}{\prodtype{\tau_1}{\tau_2}} \label{eq:dbunpair2}\\
			\TValGpc{\pair{v_{21}}{v_{22}}}{\prodtype{\tau_1}{\tau_2}} \label{eq:dbunpair3}\\
			\protjudge{\delegcontext}{H^{\pi}}{(\prodtype{\tau_1}{\tau_2})^\pi} \label{eq:dbunpair5}
                        \end{align}	
		From \eqref{eq:dbunpair2}, \eqref{eq:dbunpair3} and \ruleref{UnPair}, we have
			$\TValGpc{v_{1i}}{\tau_i}$ and $\TValGpc{v_{2i}}{\tau_i}$  for $i = \{1, 2\}$. 
		From \eqref{eq:dbunpair5},  type projection (Figure~\ref{fig:protect}) and \ruleref{P-Pair}, we have
		$\protjudge{\delegcontext}{H^\pi}{\tau_i^\pi}$.
		Combining with other premises, 
                $\TValGpc{\bracket{v_{1i}}{v_{2i}}}{\tau_i}$ follows from \ruleref{Bracket-Values}.
\fi
\textbf{Case \ruleref{B-BindM}}
		Given $e = {\bind{x}{\bracket{\returng{\ell}{v_1}}{\returng{\ell}{v_2}}}{e}}$.
		We have that:
		$$e' = \bracket{\bind{x}{\returng{\ell}{v_1}}{\outproj{e}{1}}}{\bind{x}{\returng{\ell}{v_2}}{\outproj{e}{2}}}$$
		Also $\TValGpc{\bind{x}{\bracket{\returng{\ell}{v_1}}{\returng{\ell}{v_2}}}{e}}{\tau}$. 
		From \ruleref{BindM}, we have
		\begin{align}
			&\TValGpc{\bracket{\returng{\ell}{v_1}}{\returng{\ell}{v_2}}}{\says{\ell}{\tau'}} \label{eq:db1bindm1} \\
			&\TValP{\Gamma, x:\tau'; \pc \sqcup \ell}{e}{\tau} \label{eq:db1bindm2} \\
                      	& \protjudge{\delegcontext}{\pc \sqcup \ell}{\tau} \label{eq:db1bindm4}
		\end{align}
		From \eqref{eq:db1bindm1} and \ruleref{Bracket-Values}, we have
		\begin{align}
			\TValGpc{\returng{\ell}{v_1}}{\says{\ell}{\tau'}} \label{eq:db1bracket1}\\
                        \TValGpc{\returng{\ell}{v_2}}{\says{\ell}{\tau'}} \label{eq:db1bracket2}\\
	                \rflowjudge{\delegcontext}{H^\pi}{\cfun{\says{\ell}{\tau'}}} \label{eq:db1bindm3}
		\end{align}
		We have to prove that
		\[
                \TValGpc{\bracket{\bind{x}{\returng{\ell}{v_1}}{\outproj{e}{1}}}{\bind{x}{\returng{\ell}{v_2}}{\outproj{e}{2}}}}{\tau}
	        \]
		For some $\widehat{\pc}$ we need the following premises to satisfy \ruleref{Bracket}:
		\begin{align}
			\TValP{\Gamma;\widehat{\pc}}{\bind{x}{\returng{\ell}{v_1}}{\outproj{e}{1}}}{\tau} \label{:eq:db1bracket8}\\
			\TValP{\Gamma; \widehat{\pc}}{\bind{x}{\returng{\ell}{v_2}}{\outproj{e}{2}}}{\tau} \label{:eq:db1bracket9}\\
			\rflowjudge{\delegcontext}{(H^\pi \sqcup \pc^\pi)}{{\widehat{\pc}^\pi}} \label{:eq:db1bracket10}\\
			\protjudge{\delegcontext}{H^\pi}{\cfun{\tau^\pi}} \label{:eq:db1bracket11}
		\end{align}
		A natual choice for $\widehat{pc}$ is $\pc \sqcup \ell$. 
                From Lemma~\ref{lemma:valuespc} (values  can be typed under any $\pc$), we have
                \[
        	\TValP{\Gamma;\widehat{\pc}}{\returng{\ell}{v_i}}{\tau'}
	        \]
               Applying Lemma~\ref{lemma:dproj} (bracket projection preserves typing) to \eqref{eq:db1bindm2}, we have 
                \[
        	\TValP{\Gamma, x:\tau';\widehat{\pc}}{\outproj{e}{i}}{\tau}
	        \]
                From \ruleref{BindM}, we therefore have \eqref{:eq:db1bracket8} and \eqref{:eq:db1bracket9}.
		Applying \ruleref{Trans} to \eqref{eq:db1bindm3} and \eqref{eq:db1bindm4}, we have \eqref{:eq:db1bracket11}.
                Thus we have all required premises.
%\\
%\textbf{Case \ruleref{B-Case}.} Straightforward. 
\end{proof}

\begin{lemma}[Subject Reduction(inter-host)] \label{subjRedinter}
If {\TValGpcS{\distcon{e}{c}{t}}{\tau}} and 
$\distcon{e}{c}{t} \Longrightarrow \distcon{e'}{c'}{t'}$ hold, then
{\TValGpcS{\distcon{e'}{c'}{t'}}{\tau}}.
%if $\rafjudge{\Pi}{c'}{pc}$. 
\end{lemma}
\begin{proof}Induction over typing derivation of $\distcon{e}{c}{t}$.\\
\textbf{Case \ruleref{E-Run}}.
Given,
${\distcon{E[\runa{\hat{\tau}}{e}{c'}]}{c}{t}}$ $\Longrightarrow$ \\  
${\distcon{\ret{e}{c}}{c'}{\stackapp{E[\expecta{\hat{\tau}}]}{c}{t}}}$ 
\begin{align}
\TValGpcS{\distcon{E[\runa{\hat{\tau}}{e}{c'}]}{c}{t}}{\tau} \label{sri1} \\
\intertext{need to prove}
\TValGpcS{\distcon{\ret{e}{c}}{c'}{\stackapp{E[\expecta{\hat{\tau}}]}{c}{t}}}{\tau} \label{sri2} \\
\intertext{Inverting \ref{sri1} we get}
\TValGpcc{{E[\runa{\hat{\tau}}{e}{c'}]}}{\tau'}\label{sri3} \\
\TValGpcS{t}{[\tau']\tau} \label{sri4} \\
\rafjudge{\Pi}{c}{pc} \label{sri5} \\
\stflowjudge*{pc}{pc'} \label{sri6} \\
\intertext{applying lemma CTX to \ref{sri3} we get} 
\TVal{{{\Pi}};\Gamma;pc';c}{\runa{\hat{\tau}}{e}{c'}}{\hat{\tau}}\label{sri7} \\
%\TValP{\Gamma,x\ty \hat{\tau};pc';c}{E[x]}{\tau'[x\notin dom(\Gamma)]} \label{sr8} \\
\intertext{inverting \ref{sri7} we get}
\TVal{{{\Pi}};\Gamma,\hat{\pc};c'}{e}{\hat{\tau'}} \label{sri9} \\
\intertext{where}
\hat{\tau'} = \says{\hat{pc}^{{\integ}{\avail}}}{\hat{\tau}}\\
\rafjudge{{\Pi}}{c}{pc'} \label{sri9.5} \\
\stflowjudge*{pc'}{\hat{pc}} \label{sri10} \\
\rafjudge{\Pi}{c}{\UB{\hat{\tau}}} \label{sri10.5} \\
\intertext{applying pc-reduction on \ref{sri9} with $\pc$ we get }
\TVal{{\Pi};\Gamma,pc;c'}{e}{\hat{\tau'}} \label{sri9.6} \\
\intertext{applying clearance (Lemma \ref{lemma:clearance}) in \ref{sri9.6}}
\rafjudge{{\Pi}}{c'}{pc}  \label{sri9.7 we get} \\
\intertext{applying clearance lemma in \ref{sri9}}
\rafjudge{{\Pi}}{c'}{\hat{pc}} \label{sri11} \\ 
\intertext{from \ref{sri6} and \ref{sri10} we get}
\stflowjudge*{pc}{\hat{pc}} \label{sri12} \\
\intertext{from \ref{sri11} and \ref{sri9} \ref{sri10.5} 
and \ref{sri11} and RET we get}
\TVal{{\Pi};\Gamma,\hat{\pc};c'}{\ret{e}{c}}{\hat{\tau}} \label{sri13} \\
\intertext{from \ref{sri3} and \ref{sri7} and Expect \ref{lemma:expecta} lemma we get}
\TValGpcc{E[\expecta{\hat{\tau}}]}{\tau'}\label{sri13.5} \\
\intertext{from \ref{sri4}, \ref{sri5}, \ref{sri6},\ref{sri13.5} and \ruleref{Tail} rule we get}
\TValGpcS{\stackapp{E[\expecta{\hat{\tau}}]}{c}{t}}{[\hat{\tau}]\tau} \label{sri14} \\
\intertext{from \ref{sri12},\ref{sri13},\ref{sri14} \ref{sri5} and \ruleref{Head} rule we get}
\TValGpcS{\distcon{\ret{e}{c}}{c'}{\stackapp{E[\expecta{\hat{\tau}}]}{c}{t}}}{\tau} \label{sri2}
\end{align}
\textbf{Case \ruleref{E-RetV}}. Given, \\
${\distcon{\ret{v}{c}}{c'}{\stackapp{E[\expecta{\tau'}]}{c}{t}}}$ 
$\Longrightarrow$ \\ 
${\distcon{E[\returnv{pc'^{{\integ}{\avail}}}{w}]}{c}{t}}$ \\
\begin{align}
\intertext{(where $\tau'=\says{pc'^{{\integ}{\avail}}}{\tau''}$)}
\TValGpcS{\distcon{\ret{v}{c}}{c'}
{\stackapp{E[\expecta{\tau'}]}{c}{t}}}{\tau} \label{retv1} \\
\rafjudge{\Pi}{c}{pc}  \label{retv1.1} \\
\intertext{need to prove}
\TValGpcS{\distcon{E[\returnv{pc'^{{\integ}{\avail}}}{w}]}{c}{t}}{\tau} \label{retv2} \\
\intertext{inverting \ref{retv1} we get}
\TValP{\Gamma;pc';c'}{\ret{v}{c}}{\tau'} \label{retv3} \\
\TValGpcS{\stackapp{E[\expecta{\tau'}]}{c}{t}}{[\tau']\tau} \label{retv4} \\
\stflowjudge*{pc}{pc'} \label{retv5} \\
\rafjudge{\Pi}{c'}{pc} \label{retv6} \\
\intertext{inverting \ref{retv3} we get} 
\TValP{\Gamma;pc';c'}{v}{\tau''} \label{retv7} \\ 
\rafjudge{\Pi}{c}{\UB{\tau''}} \label{retv7.5} \\
\rafjudge{\Pi}{c'}{pc'} \label{retv8} \\
\intertext{inverting \ref{retv4}}
\TValP{\Gamma;\hat{pc};c}{E[\expecta{\tau'}]}{\hat{\tau}} \label{retv9} \\
\stflowjudge*{pc}{\hat{pc}} \label{retv10} \\ 
\rafjudge{\Pi}{c}{pc} \label{retv11} \\
\TValGpcS{t}{[\hat{\tau}]\tau} \label{retv12} \\
\intertext{applying clearance lemma on \ref{retv9}}
\rafjudge{\Pi}{c}{\hat{pc}} \label{retv13} \\
\intertext{from \ref{retv13}, \ref{retv7}, \ref{retv7.5} and ValuesHost lemma}
\TValP{\Gamma;\hat{pc};c}{v}{\tau''} \label{retv14} \\
\intertext{pc reduction Lemma \ref{lemma:pcreduction} in \ref{retv14}}
\TValP{\Gamma;{pc};c}{v}{\tau''} \label{retv14.1} \\
\intertext{applying UNITM rule in \ref{retv14.1}}
\TValP{\Gamma;{pc};c}{\returnv{pc'^{{\integ}{\avail}}}{v}}{\tau'} \label{retv14.15} \\
\intertext{clearance lemma on \ref{retv14.15}}
\rafjudge{{\Pi}}{c}{pc} \label{retv14.2}\\
\intertext{from \ref{retv14.15}, \ref{retv9} and RExpect \ref{lemma:rexpect} lemma }
\TValP{\Gamma;\hat{pc};c}{E[\returnv{pc'^{{\integ}{\avail}}}{v}]}{\hat{\tau}} \label{retv15} \\  
\intertext{from \ref{retv15}, \ref{retv10}, \ref{retv12}, \ref{retv14.2} 
and HEAD rule we get}
\TValGpcS{\distcon{E[\returnv{pc'^{{\integ}{\avail}}}{v}]}{c}{t}}{\tau} 
\end{align}

\textbf{Case \ruleref{E-DStep}.} Straightforward using Induction hypothesis\\
\textbf{Case \ruleref{B-RunLeft}.} Given,
${\distcon{\bracket{E[\runa{\hat{\tau}}{e}{c'}]}{e_2}}{c}{t}}$ \\
$\Longrightarrow {\distcon{\bracket{\ret{e}{c}}{\bullet}}
{c'}{\stackapp{\bracket{E[\expecta{\hat{\tau}}]}{e_2}}{c}{t}}}$ \\
and
\begin{align}
\TValGpcS{\distcon{\bracket{E[\runa{\hat{\tau}}{e}{c'}]}{e_2}}
{c}{t}}{\tau} \label{sr1} 
\end{align}
need to prove \\
\TValGpcS{\distcon{\bracket{\ret{e}{c}}{\bullet}}
{c'}{\stackapp{\bracket{E[\expecta{\hat{\tau}}]}{e_2}}
{c}{t}}}{\tau} 
\begin{align}
\intertext{Inverting \ref{sr1} we get}
\TValGpcc{\bracket{E[\runa{\hat{\tau}}{e}{c'}]}{e_2}}
{\tau'}\label{sr3} \\
\TValGpcS{t}{[\tau']\tau} \label{sr4} \\
\rafjudge{\Pi}{c}{pc} \label{sr5} \\
\stflowjudge*{pc}{pc'} \label{sr6} \\
\intertext{inverting \ref{sr3} we get}
\TValP{\Gamma;pc'',c}{E[\runa{\hat{\tau}}{e}{c'}]}{\tau'} \label{sr3.1} \\
\TValP{\Gamma;pc'',c}{e_2}{\tau'} \label{sr3.2} \\
\drflowjudge{\Pi}{H^{\pi}}{\cfun{\tau'}} \label{sr3.3} \\
\drflowjudge{\Pi}{H^{\pi} \join pc'}{pc''} \label{sr3.4} \\
\intertext{applying lemma CTX to \ref{sr3.1} we get} 
\TVal{{{\Pi}};\Gamma;pc'';c}
{\runa{\hat{\tau}}{e}{c'}}{\hat{\tau}}\label{sr7} \\
%\TValP{\Gamma,x\ty \hat{\tau};pc';c}{E[x]}{\tau'[x\notin dom(\Gamma)]} \label{sr8} \\
\intertext{inverting \ref{sr7} we get}
\TVal{{{\Pi}};\Gamma,\hat{\pc};c'}{e}{\hat{\tau'}} \label{sr9} \\
\intertext{where}
\hat{\tau'} = \says{\hat{pc}^{{\integ}{\avail}}}{\hat{\tau}}\\
\rafjudge{{\Pi}}{c}{pc''} \label{sr9.5} \\
\stflowjudge*{pc''}{\hat{pc}} \label{sr10} \\
\rafjudge{\Pi}{c}{\UB{\hat{\tau}}} \label{sr10.5} \\
\intertext{applying pc-reduction on \ref{sr9} with $\pc$ we get }
\TVal{{\Pi};\Gamma,pc';c'}{e}{\hat{\tau'}} \label{sr9.6} \\
\intertext{applying clearance in \ref{sr9.6}}
\rafjudge{{\Pi}}{c'}{pc'}  \label{sr9.7 we get} \\
\intertext{applying clearance lemma in \ref{sr9}}
\rafjudge{{\Pi}}{c'}{\hat{pc}} \label{sr11} \\ 
\intertext{from \ref{sr6}, \ref{sr3.1} and \ref{sr10} we get}
\stflowjudge*{pc}{\hat{pc}} \label{sr12} \\
\intertext{from \ref{sr11} and \ref{sr9} \ref{sr10.5} 
and \ref{sr11} and RET we get}
\TVal{{\Pi};\Gamma,\hat{\pc};c'}{\ret{e}{c}}{\hat{\tau}} \label{sr13} \\
\intertext{from \ref{sr3.1} and \ref{sr7}   
and Expect \ref{lemma:expecta} lemma we get}
\TValP{\Gamma;pc'';c}{E[\expecta{\hat{\tau}}]}{\tau'}\label{sr13.5} \\
\intertext{from \ref{sr3.2}, \ref{sr3.3}, \ref{sr3.4}, 
\ref{sr13.5} we get}
\TValP{\Gamma;pc'';c}{\bracket{E[\expecta{\hat{\tau}}]}
{e_2}}{\tau'}\label{sr13.5} \\
\intertext{from \ref{sr13} and \ruleref{BullL} rule }
\TVal{{\Pi};\Gamma,\hat{\pc};c'}{\bracket{\ret{e}{c}}
{\bullet}}{\hat{\tau}} \label{sr13} \\
\intertext{from \ref{sr4}, \ref{sr5}, \ref{sr6},\ref{sr13.5} and \ruleref{Tail} rule we get}
\TValGpcS{\stackapp{\bracket{E[\expecta{\hat{\tau}}]}
{e_2}}{c}{t}}{[\hat{\tau}]\tau} \label{sr14}
\end{align}
from \ref{sr12},\ref{sr13},\ref{sr14} \ref{sr5} and 
\ruleref{Head} rule we get \\
$\TValGpcS{\distcon{\bracket{\ret{e}{c}}{\bullet}}
{c'}{\stackapp{\bracket{E[\expecta{\hat{\tau}}]}{e_2}}{c}{t}}}{\tau}$

\begin{figure*}
  {\small
\[
\begin{array}{l l  l}
\observefc{\faila{\tau}}{\Pi}{p} & = &  \circ \\
\observefc{x}{\delegcontext}{p} & = &   x \\
\observefc{()}{\delegcontext}{p} & = &   \circ \\
\observefc{\return{\ell}{e}}{\delegcontext}{p} & = & 
           \return{\ell}{\observefc{e}{\delegcontext}{p}} \\
\observefc{\returng{\ell}{v}}{\delegcontext}{p} & = & 
    \begin{cases}
   \returng{\ell}{\observefc{v}{\delegcontext}{p}} &  \drflowjudge{\delegcontext}{\ell^{π}}{p^{\pi}} \\
     \circ &   \text{otherwise}
    \end{cases}\\[1.25em]
\observefc{\lamc{x}{\tau}{\pc}{e}}{\delegcontext}{p} & = &
    \begin{cases}
  \lamc{x}{\tau}{\pc}{\observefc{e}{\delegcontext}{p}} &  \drflowjudge{\delegcontext}{\pc^{π}}{p^{\pi}} \\
    \circ &   \text{otherwise}
   \end{cases}\\[1.25em]
\observefc{\tlam{X}{\pc}{e}}{\delegcontext}{p} & = & 
   \begin{cases}
   \tlam{X}{\pc}{\observefc{e}{\delegcontext}{p}} &  \drflowjudge{\delegcontext}{\pc^{π}}{p^{\pi}} \\
   \circ  &  \text{otherwise}
 \end{cases}\\[1.25em]
\observefc{e~e'}{\delegcontext}{p} & = & \observefc{e}{\delegcontext}{p}~ \observefc{e'}{\delegcontext}{p} \\
\observefc{\pair{e_1}{e_2}}{\delegcontext}{p} & = & 
  \begin{cases}
   \circ & \observefc{e_i}{\delegcontext}{p} = \circ \\
  \pair{\observefc{e_1}{\delegcontext}{p}}{\observefc{e_2}{\delegcontext}{p}} & \text{otherwise}
    \end{cases} \\
   \observefc{\proji{e}}{\delegcontext}{p} & = & \proji{\observefc{e}{\delegcontext}{p}} \\
\observefc{\inji{e}}{\delegcontext}{p} & = &  
   \begin{cases}
    \circ & \observefc{e}{\delegcontext}{p} = \circ \\
   \inji{\observefc{e}{\delegcontext}{p}} & \text{otherwise}
   \end{cases} \\
   \observefc{\casexp{e}{x}{e_1}{e_2}}{\delegcontext}{p} & = & \casexpa{\observefc{e}{\delegcontext}{p}} \\ 
    && \quad \phantom{\mid}~\casexpb{x}{\observefc{e_1}{\delegcontext}{p}} \\ 
    && \quad \casexpc{x}{\observefc{e_2}{\delegcontext}{p}} \\ 
    \observefc{\bind{x}{e}{e'}}{\delegcontext}{p} & = & \bind{x}{\observefc{e}{\delegcontext}{p}}{\observefc{e'}{\delegcontext}{p}} \\
\observefc{\select{e_1}{e_2}}{\delegcontext}{p} & = &
      \select{\observefc{e_1}{\Pi}{p}}{\observefc{e_2}{\Pi}{p}} \\
\observefc{\compare{e_1}{e_2}}{\delegcontext}{p} & = &
\compare{\observefc{e_1}{\Pi}{p}}{\observefc{e_2}{\Pi}{p}} \\ 
\observefc{\distcon{e}{c}{s}}{\Pi}{\ell} & = &
      \begin{cases}
        \observefc{e}{\Pi}{\ell} & s=\emptystack \\
        \observefc{e}{\Pi}{\ell}\& \observefc{s}{\Pi}{\ell} \\
      \end{cases} \\[1.25em] 
\observefc{\stackapp{e}{c}{s}}{\Pi}{\ell} & = &
     \begin{cases}
       \observefc{e}{\Pi}{\ell}  & s=\emptystack \\
       \observefc{e}{\Pi}{\ell} :: \observefc{s}{\Pi}{\ell} \\
     \end{cases}\\ [1.25em]
\observefc{E[\runa{\tau}{e}{c}]}{\Pi}{\ell} & = & \observefc{E[e]}{\Pi}{\ell}\\
\observefc{\ret{e}{c}}{\Pi}{\ell} & = & \observefc{e}{\Pi}{\ell} \\ 
%\observefc{\bracket{e_1}{e_2}}{\Pi}{\ell} & = & 
%\observefc{\bracket{\observefc{e_1}{\Pi}{\ell}}
%{\observefc{e_2}{\Pi}{\ell}}}{\Pi}{\ell} \\
%\observefc{\bracket}{}{} & = & \observefc{}{}{}
\end{array}
\]
}
\caption{Observation function for intermediate FLAQR terms.}
\label{fig:observe}
\end{figure*}

\textbf{Case \ruleref{B-RunRight}.} \\
Same as above.\\

\textbf{Case \ruleref{B-RetLeft}} \\
Same as above.\\

\textbf{Case \ruleref{B-RetRight}} \\
Same as above.\\

\textbf{Case \ruleref{B-RetV}} \\
Same as above.\\

\end{proof}

\ciNonInterference*
\begin{proof}
From Subject reduction of bracketed 
FLAQR constructs we can write
$$\TValP{\Gamma;pc}{\distcon{v}{c}{\emptystack}}{\says{\ell}{\tau}}$$
We will write $\outproj{v}{i}$ as $v_i$.
We need to show 
$\observet{{v}_{1}}{\Pi}{\ell}{\pi} =
\observet{{v}_{2}}{\Pi}{\ell}{\pi}$.
Since $v_i$ has protected type, we know that
it is of form $\returnv{\ell}{v_i'}$
Since $\drflowjudge{\Pi}{\ell^{\pi}}{\ell^{\pi}}$
for $\pi \in \{"c", "i"\}$, 
we just have to show\\
$\returnv{\ell}{\observet{v_1'}{\Pi}{\ell}{\pi}} = 
\returnv{\ell}{\observet{v_2'}{\Pi}{\ell}{\pi}}$\\
Which is true if we can show \\
 ${\observet{v_1'}{\Pi}{\ell}{\pi}}
= {\observet{v_2'}{\Pi}{\ell}{\pi}}$.\\
Which can be easily shown by induction
over structure of $v_i$s. 
\end{proof}

\availNonInterference*
\begin{proof}
From subject reduction ( \ref{subjRedhost} and \ref{subjRedinter} ) we know, 
$\TValGpcw{\outproj{f}{i}}{\says{\ell_{\quo}}{\tau}}$.
Because ${\protsA{\quo}{(\says{\ell_{\quo}}{\tau})}}$ and
$\reachn{H} \in \A_{\trans{\quo}}~$ we can write
$\notrecrafjudge{\Pi}{\reachn{H}}{\says{\ell_{\quo}}{\tau}}$
from rule \ruleref{Q-Guard}.
This ensures if
$\outproj{f}{1} \neq \faila{\says{\ell_{\quo}}{\tau}}$,
then $\outproj{f}{2} \neq \faila{\says{\ell_{\quo}}{\tau}}$, 
and vice-versa.
\end{proof}

\section{Correctness of Blame Semantics.}
\label{sec:blameproof}
\begin{figure}
{\footnotesize
\begin{flalign*}
& \boxed{\entails{\mathcal{C}_1}{\mathcal{C}_2}} &
\end{flalign*}
\begin{mathpar}
        \Rule{C-IN}{\rafjudge{\Pi}{\ell'}{\ell}}
        {\entails{\inF{\ell'}{\FN}}{\inF{\ell}{\FN}}}

	\Rule{C-OR}{\entails{\blame_1}{\inF{\ell}{\FN}} \quad \quad \entails{\blame_2}{\inF{\ell}{\FN}}}
	{\entails{\blame_1 \OR \blame_2}{\inF{\ell}{\FN}}}

        \Rule{C-ANDL}{\exists i \in \{1,2\}. ~\entails{\blame_i}{\inF{\ell}{\FN}}}
        {\entails{\blame_1 \AND \blame_2}{\inF{\ell}{\FN}}} 
\end{mathpar}
}
\caption{Blame Membership.}
\label{fig:blame member}
\end{figure}

\begin{figure*}
\derule{C-CompareFail}{v_1 \neq v_2 ~~ \blame':= \last(v_1,
v_2, \blame,\ell_1,\ell_2)}
{\concon{\comparea{\says{\txcmp{ℓ₁}{ℓ₂}}{τ}}
{\returnv{ℓ₁}{v_1}}{\returnv{ℓ₂}{v_2}}}{c}{s}{\blame}}
{\concon{\faila{\says{\txcmp{ℓ₁}{ℓ₂}}{τ}}}{c}{s}{\blame'}}
\begin{flalign*}
\intertext{Change this algo to over approximate the attacker}
\last(x,y,\blame,\ell_1,\ell_2) =~& \MATCH~ (x,y)~ \WITH \\
|& (\returnv{\ell}{v_1},\returnv{\ell}{v_2}) =  
                           \IF~ \entails{\blame}{\inF{\ell_1}{\FN}}~ \THEN~ \blame \\
&~~~~~~~~~~~~~~~~~~~~~~~~~~~~~~~~~\ELSE~\IF~ \entails{\blame}{\inF{\ell_2}{\FN}}~\THEN~ \blame\\
&~~~~~~~~~~~~~~~~~~~~~~~~~~~~~~~~~\ELSE~\IF~ \entails{\blame}{\inF{\ell}{\FN}}~\THEN~ \blame\\
&~~~~~~~~~~~~~~~~~~~~~~~~~~~~~~~~~\ELSE ~\last(v_1,v_2,\blame,\ell,\ell) \\
|& (\return{\ell}{e_1},\return{\ell}{e_2}) = \last(e_1,e_2,\blame,\ell_1,\ell_2) \\ 
|& (\injia{\tau}{e_1},\injia{\tau}{e_2}) = \last(e_1,e_2,\blame,\ell_1,\ell_2) \\
|& (\paira{e_{11}}{e_{12}}{\tau},\paira{e_{21}}{e_{22}}{\tau}) = 
                                     \last(e_{11},e_{21},(\last(e_{12},e_{22},\blame,\ell_1,\ell_2)),\ell_1,\ell_2) ~ \\
% & \quad \quad \quad \quad \quad \quad \quad \quad \quad \quad \quad \quad \quad \quad \quad \quad \quad \last(w_{12},w_{22},\blame,\ell_1,\ell_2) \\
|& (\runa{\tau}{e_1}{p},\runa{\tau}{e_2}{p}) = \last(e_1,e_2,\blame,\ell_1,\ell_2) \\
%|& (\ret{e_1}{p},\ret{e_2}{p}) = \last(e_1,e_2,\blame,\ell_1,\ell_2) \\
|& (\selecta{e_1}{e_2}{\tau},\selecta{e_1'}{e_2'}{\tau}) = \last(e_1,e_1',
(\last(e_2,e_2',\blame,\ell_1,\ell_2)),\ell_1,\ell_2) \\ 
|& (\comparea{\tau}{e_1}{e_2},\comparea{\tau}{e_1'}{e_2'}) = \last(e_1,e_1', 
(\last(e_2,e_2',\blame,\ell_1,\ell_2)),\ell_1,\ell_2) \\
|& (\lamc{x}{\tau}{pc}{e_1},\lamc{x}{\tau}{pc}{e_2}) = \last(e_1,e_2,\blame,\ell_1,\ell_2) \\
|& (\tlam{X}{pc}{e_1},\tlam{X}{pc}{e_2}) = \last(e_1,e_2,\blame,\ell_1,\ell_2) \\
|& (\proji{e_1},\proji{e_2}) = \last(e_1,e_2,\blame,\ell_1,\ell_2) \\ 
|& (\bind{x_1}{e_1}{e_1'},\bind{x_2}{e_2}{e_2'}) = \last(e_1,e_2,\last(e_1',e_2',\blame,\ell_1,\ell_2)
,\blame,\ell_1,\ell_2) \\
|& (\casexpan{e_1}{z}{e_2}{e_3}{\tau},\casexpan{e_1'}{z}{e_2'}{e_3'}{\tau}) = \\
&~~~~~~~~~~~~~~~~~~~~~~~~~~~~~~~~~\last(e_1,e_1',\last(e_2,e_2',\last(e_3,e_3',\blame,\ell_1,\ell_2),\ell_1,\ell_2),\ell_1,\ell_2)\\
|& (f_1,f_2) =  \IF~ f_1 = f_2 \THEN ~ \blame \\
 &~~~~~~~~~~~~~~~~~~~  \ELSE ~\IF~ \entails{\blame}{\inF{\ell_1}{\FN}}~ \THEN~ \blame \\
 &~~~~~~~~~~~~~~~~~~~  \ELSE ~\IF~ \entails{\blame}{\inF{\ell_2}{\FN}} ~\THEN~ \blame~ \\
 &~~~~~~~~~~~~~~~~~~~  \ELSE ~ ((\inF{\ell_1}{\FN}) \AND \blame) \OR 
                 ((\inF{\ell_2}{\FN}) \AND \blame)
\end{flalign*}
\caption{Function to construct blame constraint $\blame$.}
\label{fig:Blameconst}
\end{figure*}

\begin{figure*}
\begin{flalign*}
\inF{\ell}{\FN} \AND \blame~ =>~ & \MATCH~ \blame~ \WITH \\
              |& \FN=\emptyset => \inF{\ell}{\FN} \\
              |& \inF{\ell'}{\FN} => \inF{\ell}{\FN} \AND \inF{\ell'}{\FN} \\
              |& \blame_1 \OR \blame_2 => (\inF{\ell}{\FN} \AND \blame_1) \OR (\inF{\ell}{\FN} \AND \blame_2) \\
              |& \blame_1 \AND \blame_2 => \blame_1 \AND \blame_2  \AND \inF{\ell}{\FN} \\
\end{flalign*}
\caption{AND Function.}
\label{fig:ANDfunction}
\end{figure*}

\textbf{In the following proofs, a possible faulty set 
$\FN$ is referred as a faulty set that 
satisfies the blame constraint, or $\FN$ implied by $\blame$. 
And $\recrafjudge{\Pi}{b}{\tau}$ 
is equivalent to saying $\recrafjudge{\Pi}{b^{{\integ}{\avail}}}{\tau}$.
}

\begin{lemma}[Reach ConjL]\label{lemma:reachLemma1}
If $\rafjudge{\Pi}{\reachn{b}}{t}$ then 
$\rafjudge{\Pi}{\reachn{b \wedge p}}{t}$
\end{lemma}
%\textbf{Proof.}
\begin{proof}
Given, $\rafjudge{\Pi}{\reachn{b}}{t}$.
So from \ruleref{ConjL} we can write:
$\rafjudge{\Pi}{\reachn{b}\wedge \reachn{p}}{t}$,
which is same as writing
$\rafjudge{\Pi}{\reachn{b \wedge p}}{t}$.
\end{proof}

\begin{lemma}[Reach ConjR Type]\label{lemma:reachLemma2}
If $\recrafjudge{\Pi}{{b}}{\says{\ell_1}{\tau}}$  and 
$\recrafjudge{\Pi}{{b}}{\says{\ell_2}{\tau}}$ then 
$\recrafjudge{\Pi}{{b}}
{\says{(\txsel{\ell_1}{\ell_2})}{\tau}}$
\end{lemma}
%\textbf{Proof.}
\begin{proof}
If $\recrafjudge{\Pi}{{b}}{\tau}$ then
$\recrafjudge{\Pi}{{b}}
{\says{(\txsel{\ell_1}{\ell_2})}{\tau}}$ is true as well 
from \ruleref{A-Type}.
If $\notrecrafjudge{\Pi}{{b}}{\tau}$
then it is obvious that
$\rafjudge{\Pi}{b^{\avail}}{\ell_1^{\avail}}$ and
$\rafjudge{\Pi}{b^{\avail}}{\ell_2^{\avail}}$.
Which along with \ruleref{R-ConjR} implies,
$\rafjudge{\Pi}{b^{\avail}}{\ell_1^{\avail} 
\wedge \ell_2^{\avail}}$.
Thus from \ruleref{A-Avail} we can write
$\recrafjudge{\Pi}{{b}}
{\says{(\txsel{\ell_1}{\ell_2})}{\tau}}$
\end{proof}

\begin{lemma}[FAIL RESULT ONESTEP TO FAIL]\label{FN1}
If $\concon{e}{c}{s}{\blame}$ $\stepsone 
\concon{\faila{\tau}}{c'}{s'}{\blame'}$ and $e$ is a source level 
term, then $e$ must be of the 
form $\comparea{\tau}{v_1}{v_2}$
and the only evaluation step that has been taken
to transition from $e$ to $\faila{\tau}$ is
\ruleref{C-CompareFail}.
\end{lemma}
%\textbf{Proof.} 
\begin{proof} 
Trivial(by inspection on evaluation rules).
\end{proof}

\begin{lemma}[FAIL SUBEXP ONESTEP]\label{FN2}
If $\concon{e}{c}{s}{\blame} \stepsone
\concon{e'}{c'}{s'}{\blame'}$ and $e$ is a 
source level term but $e'$ has a $\faila{}$ term in it
then the only evaluation step that has been taken is
\ruleref{C-CompareFail}.                                                  \end{lemma}
%\textbf{Proof.} 
\begin{proof}
Trivial(by inspection on  evaluation rules). 
\end{proof}

\begin{lemma}[INTRO GTR] \label{gtrdotAND}
If $\recrafjudge{\Pi}{\ell_1}{\tau}$ then for any $\ell_2$
$\recrafjudge{\Pi}{\ell_1 \wedge \ell_2}{\tau}$
\end{lemma}
%\textbf{Proof.}
\begin{proof}
Proof is straightforward from $fail$ judgements $\gtrdot$ 
(Figure \ref{fig:avail-actsfor}) rules and \ruleref{ConjL}
rule.
\end{proof}

\begin{lemma}[FAIL ONESTEP NON-EMPTY BLAME SET] \label{FN3}
If  $\concon{e}{c}{s}{\blame}$ $\stepsone$ 
$\concon{e'}{c'}{s'}{\blame'}$, $e$ is a source level term, 
and $e'$ contains $\faila{\tau}$, then 
for any $\FN$ that satisfies $\blame'$ the
following condition holds:\\
$~~~~~~~~~~~~~\recrafjudge{\Pi}{{b'}}{\tau}$, where
$b'=\bigwedge_{\forall f \in \FN} f$
%\\ infact \\
%$~~~~~~~~~~~~~\rafjudge{\Pi}{{b'}^a}{\tau}$
\end{lemma}
\begin{proof}
Since transition from $e$ to $e'$ introduces a
$\faila{\tau}$ term, from lemma \ref{FN2}, we 
know that the only evaluation step that was taken is \ruleref{C-CompareFail}.
Without loss of generality we can state that 
$e= E[\comparea{\says{\txcmp{\ell_1}{\ell_2}}{\tau'}}
{\returnv{\ell_1}{v_1}}{\returnv{\ell_2}{v_2}}]$
and $\blame' = \last(v_1,v_2,\blame,\ell_1,\ell_2)$, for some
$\ell_1,\ell_2$ and $\tau'$. In the following 
proof for $\FN$ and 
$\FN'$ that satisfies $\blame$ and $\blame'$ respectively,
$b$ and $b'$ will mean $\bigwedge_{\forall f \in \FN} f$ and 
$\bigwedge_{\forall f \in \FN'} f$ respectively.
There are three possibilities for $\blame'$.
\begin{enumerate}
\item $\boldsymbol{\blame' = ((\inF{\ell_1}{\FN}) \AND \blame) \OR 
((\inF{\ell_2}{\FN}) \AND \blame):}$
It is straightforward to say that,
for any $\FN$ that satisfies 
$\blame$, $\FN'=\FN \cup \{\ell_i\}$ 
satisfies $\blame'$, for $i \in \{1,2\}$.
If $b$ is conjunction of labels in $\FN$ then 
$b \wedge \ell_i$, $i \in \{1,2\}$ are 
conjunction of labels in $\FN'$.
We know, $\recrafjudge{\Pi}{\ell_i}
{\says{(\txcmp{\ell_1}{\ell_2})}{\tau'}}$.
% (since $(\says{\txcmp{\ell_1}{\ell_2}}{\tau'})$ 
%= ($\says{(\ell_1 \vee \ell_2)}{\tau'})$).
From lemma ITRO GTR \ref{gtrdotAND} 
we can write for label $b$,
$\recrafjudge{\Pi}{(\ell_i \wedge b)}
{\says{(\txcmp{\ell_1}{\ell_2})}{\tau'}}$.
Beacuse conjunctions of labels in $\FN'$ 
are of form $b'=(\ell_i \wedge b)$,
we can say that for any $\FN'$ that 
satisfies $\blame'$ the following condition holds\\
$~~~~~~~~~~\recrafjudge{\Pi}{{b'}}{\says{(\txcmp{\ell_1}{\ell_2})}{\tau'}}$, where
$b'=\bigwedge_{\forall f \in \FN'} f$.\\
%Thus, the lemma condition is true as $F'$ is an instance of $\FN$.

\item $\boldsymbol{\blame' = (\inF{\ell}{\FN}) \AND \blame:}$ 
Here $\ell$ is some inner layer in $\tau'$,
i.e. $\tau = \says{(\txcmp{\ell_1}{\ell_2})}{\tau'} = 
\says{(\txcmp{\ell_1}{\ell_2})}{(...(\says{\ell}{\tau''})...)}$.\\
It is straightforward to say that,
for any $\FN$ that satisfies $\blame$, $\FN'= \{\ell\} \cup \FN$ 
satisfies $\blame'$.
If $b$ is conjunction of labels in $\FN$ then 
$\ell \wedge b$ is conjunction of labels in $\FN'$.
We know $\recrafjudge{\Pi}{(\ell)}{\tau}$.
From lemma ITRO GTR \ref{gtrdotAND} we can write for any label $b$,
$\recrafjudge{\Pi}{(\ell \wedge b)}{\tau}$.
Beacuse conjunctions of labels in 
$\FN'$ are of the form $b'=(\ell \wedge b)$,
we say that for any $\FN'$ that satisfies 
$\blame'$ the following condition holds\\
$~~~~~~~~~~\recrafjudge{\Pi}{{b'}}{\tau}$, where
$b'=\bigwedge_{\forall f \in \FN'} f$.\\
%Thus, the lemma condition is true as $F'$ is an instance of $\FN$.

\item $\boldsymbol{\blame' = \blame:}$ 
This case occurs when \ruleref{C-CompareFail} rule does not 
update $\blame$ beacuse the label $\ell$ that is responsible for 
generating $\faila{}$ is already included in all possible
$\FN$ in $\blame$. That means  
$\entails{\blame}{\inF{\ell}{\FN}}$ and
$\recrafjudge{\Pi}{\ell}{\tau}$ 
(since the end result is fail and $\ell$ is the responsible label). 
Thus it is straightforward from lemma \ref{gtrdotAND} that
for label $b$, where $b=b_1 \wedge ... \ell ... \wedge b_k$ 
$\recrafjudge{\Pi}{{b}}{\tau}$.
Thus we have,
$~~~~~~~~~~\recrafjudge{\Pi}{{b}}{\tau}$, where
$b=\bigwedge_{\forall f \in \FN} f$.\\

\end{enumerate}
\end{proof}

\begin{lemma}[GTRDOT OR] \label{gtrdottrans}
If $\recrafjudge{\Pi}{\ell}{\says{\ell_1}{\tau}}$ then 
$\recrafjudge{\Pi}{\ell}{\says{(\ell_1 \vee \ell_2)}{\tau}}$.
\end{lemma}
\begin{proof}
Proof is straightforward inspecting 
$\gtrdot$ rule (Figure \ref{fig:avail-actsfor}) 
and using \ruleref{DisjR}.
\end{proof}

\begin{lemma} [FAIL RESULT ONESTEP] \label{FN4}
Given,
\begin{enumerate}
\item $\TValGpcw{\concon{e}{c}{s}{\blame}}{\tau'}$
\item $\concon{e}{c}{s}{\blame} \stepsone 
\concon{\faila{\tau'}}{c}{s}{\blame'}$ 
\end{enumerate} then for every ~$\FN$~ that satisfies ~$\blame'$~ 
it must be the case that 
$\recrafjudge{\Pi}{b'}{\tau'}$ , where  $b'=\bigwedge_{f \in \FN} f$
\end{lemma}
\begin{proof}
%\begin{proof}
Let us proof this by induction over typing derivation $\TValGpcw{e}{\tau}$.
In the proof  $b$ is always $\bigwedge_{f \in \FN} f$.
\begin{itemize}
\item \textbf{Case \ruleref{E-AppFailL}:}
$e=(\lamc{x}{\tau_1}{\pc}{\faila{\func{\tau_1}{pc}{\tau}}})e_1$ and
$e' = \faila{\tau}$.
From IH we get, 
%$\recrafjudge{\Pi}{b}{\func{\tau_1}{pc}{\tau}}$
%From the above and rule \ruleref{A-Fun} we get 
$\recrafjudge{\Pi}{b}{\tau}$,
which is what we wanted.
\item \textbf{Case \ruleref{E-SealedFail}:} 
$e=\return{\ell}{\faila{\tau}}$ and 
$e'=\faila{\says{\ell}{\tau}}$\\
$\concon{\return{\ell}{\faila{\tau}}}{c}{s}{\blame}
\stepsone \concon{\faila{\says{\ell}{\tau}}}{c}{s}{\blame} $\\
Given, $\TValGpcw{\return{\ell}{\faila{\tau}}}{\says{\ell}{\tau}}$.\\
Therefore from induction hypothesis we have,
$\recrafjudge{\Pi}{b}{\tau}$,
$\blame$ does not get updated 
during the evaluation step.
Thus from the rule \ruleref{A-Type} we have,
$\recrafjudge{\Pi}{b}{\says{\ell}{\tau}}$, 

\item \textbf{Case \ruleref{E-InjFail}:} 
$e = \injia{\sumtype{\tau_1}{\tau_2}}{\faila{\tau_i}}$ and 
$e'= \faila{\sumtype{\tau_1}{\tau_2}}$ \\
$\injia{\sumtype{\tau_1}{\tau_2}}{\faila{\tau_i}}
\stepsone \faila{\sumtype{\tau_1}{\tau_2}}$ \\
Given, $\TValGpcw{\injia{\sumtype{\tau_1}{\tau_2}}{\faila{\tau_i}}}
{\sumtype{\tau_1}{\tau_2}}$\\
From IH we can say for every $\FN$ that satisfies $\blame$,
$\recrafjudge{\Pi}{b}{\tau_i}$ \\
$\blame$ does not get updated 
because of the evaluation step.
Using rule \ruleref{A-Sum},
we can write the following,
$\recrafjudge{\Pi}{b}{\sumtype{\tau_1}{\tau_2}}$

\item \textbf{Case \ruleref{E-CaseFail}:} 
From \ruleref{Case} typing rule we know 
$\rafjudge{\Pi}{{\tau_i}^{\avail}}{{\tau}^{\avail}}$ for
$i \in \{1,2\}$.
IH gives us $\recrafjudge{\Pi}{b}{\tau_i}$.
Thus we can write
$\recrafjudge{\Pi}{b}{\tau}$.

\item \textbf{Case \ruleref{E-ProjFail}:} 
Same as case \ruleref{E-AppFail} as the type annotation 
of $\faila{}$ does not change.

\item \textbf{Case \ruleref{E-SelectFail}:} \\
$e=\selecta{\faila{\says{\ell_3}{\tau}}}
{\faila{\says{\ell_4}{\tau}}}{\says{(\ell_3 \oplus \ell_4)}{\tau}}$ \\
$e'=\faila{\says{(\txsel{\ell_3}{\ell_4})}{\tau}}$\\
$\concon{e}{c}{s}{\blame}$
$\stepsone \concon{\faila{\says{(\txsel{\ell_3}{\ell_4})}
{\tau}}}{c}{s}{\blame}$\\
Given,
\begin{align}
\TValGpcw{\selecta{f_1}{f_2}{\says{(\txsel{\ell_3}{\ell_4})}{\tau}}}
{\says{(\txsel{\ell_3}{\ell_4})}{\tau}}
\end{align}
From IH,
\begin{align}
\recrafjudge{\Pi}{b}{\says{\ell_3}{\tau}} \label{fos1} \\
\recrafjudge{\Pi}{b}{\says{\ell_4}{\tau}} \label{fos2} \\
\intertext{Applying Lemma \ref{lemma:reachLemma2} with \ref{fos1} and 
\ref{fos2}}
\recrafjudge{\Pi}{b}{\says{(\txsel{\ell_3}{\ell_4})}{\tau}}\\
\end{align}
\item \textbf{Case \ruleref{E-CompareFailL}:}\\
$e=\comparea{\says{\txcmp{\ell_3}{\ell_4}}{\tau}}
{(\faila{\says{\ell_3}{\tau}})}{\returnv{\ell_4}{v_2}}$\\
$\concon{e}{c}{s}{\blame}$
$\stepsone
\concon{\faila{\says{\txcmp{\ell_3}{\ell_4}}{\tau}}}{c}{s}{\blame}$ \\
From IH we get,
\begin{align}
\recrafjudge{\Pi}{b}{\says{\ell_3}{\tau}} \\
\recrafjudge{\Pi}{b}{\says{\ell_4}{\tau}} \\
\intertext{using \ruleref{A-Type}, or \ruleref{A-Avail} or \ruleref{A-IntegCom} we can write the following}
\recrafjudge{\Pi}{b}{\says{(\txcmp{\ell_3}{\ell_4})}{\tau}} \label{compfail2}
\end{align}
\item \textbf{Case \ruleref{E-CompareFailR}:} same as \ruleref{E-CompareFailL}
\item \textbf{Case \ruleref{C-CompareFail}:}  Trivially true based on definition of $\last$.
and lemma \ref{FN3}.

\item \textbf{Case \ruleref{E-TAppFail}:} \\ $e={\faila{\tfunc{X}{pc}{\tau}}}$ and 
$\concon{\faila{\tfunc{X}{pc}{\tau_1}}~\tau_2}{c}{s}{\blame}$ $\stepsone$
$\concon{\faila{\subst{\tau_1}{X}{\tau_2}}}{c}{s}{\blame}$ 
Given,  
$\recrafjudge{\Pi}{b}{\tfunc{X}{pc}{\tau}}$,
which is same as saying
$\recrafjudge{\Pi}{b}{\subst{\tau}{X}{\tau'}}$
where $b = \bigwedge_{f \in \FN} f$\\

\item \textbf{Case \ruleref{E-RetFail}:} IH gives us $\recrafjudge{\Pi}{b}{\tau}$,
Then from \ruleref{A-Type} $\recrafjudge{\Pi}{b}{\says{pc^{{\integ}{\avail}}}{\tau}}$.
\item \textbf{Other cases:} neither fail propagates nor any change in $\blame$.
\end{itemize}
\end{proof}

\failresult*
\begin{proof}
$e$ does not have any $\faila{}$ terms in it and 
it steps to a $\faila{}$ term. From lemma \ref{FN2} we know that 
there has to be at least one \ruleref{C-CompareFail} step taken during the evaluation.
Either $e$ takes single step or multiple steps to produce the $\faila{\tau}$ 
term as the end result.
\begin{enumerate}
\item $\boldsymbol{\faila{\tau}}$ \textbf{is produced via single step:}
If $e$ takes a single step ,
\begin{itemize}
\item From lemma \ref{FN2} we know $e$ is of the form \\
$\comparea{(\txcmp{\ell_3}{\ell_4})}{v_1}{v_2}$, 
and from lemma \ref{FN3} we know that 
$\recrafjudge{\Pi}{b_i}{\tau}$ \\
\item Another possibility is that,
the last evaluation step produces the $\faila{\tau}$ result.\\
$\concon{e}{c}{\emptystack}{\emptyset}$ $\stepsone^{*}$ 
$\concon{e_{n-1}}{c_{n-1}}{s_{n-1}}{\blame_{n-1}}$ $\stepsone$ 
$\concon{\faila{\tau}}{c}{\emptystack}{\blame}$
From lemma \ref{FN1} $e_{n-1} = \comparea{\tau}{v_1}{v_2}$ and 
from lemma \ref{FN3} we know \\
$\recrafjudge{\Pi}{b_i}{\tau}$
\end{itemize}
\item $\boldsymbol{\faila{\tau}}$ 
\textbf{is produced by propagation of a $\faila{}$ term:}
This means the evaluation takes more than one step.  
Let us prove it by induction over structure of $e$.\\
$\concon{e}{c}{\emptystack}{\emptyset} \stepsone^{*} 
\concon{e_i}{c_i}{s_i}{\blame_i} \stepsone^{*} 
\concon{\faila{\tau}}{c_n}{s_n}{\blame_n}$\\
Without loss of generality we can say that \\ 
the $\faila{}$ term that propagates till 
the end is introduced first in expression $e_i$. That means there 
exists an expression $e_{i-1}$ such that it takes 
\ruleref{C-CompareFail} evaluation rule to step to $e_i$.\\
$\concon{e}{c}{\emptystack}{\emptyset} \stepsone^{*} 
\concon{e_{i-1}}{c_{i-1}}{s_{i-1}}{\blame_{i-1}} \stepsone 
\concon{e_i}{c_i}{s_i}{\blame_i} $ \\
$\stepsone^{*} \concon{\faila{\tau}}{c}{\emptystack}{\blame}$ 
\\
%$\concon{}{}{e_{i-1}}{C_{i-1}} \stepsone \concon{}{}{e_i}{C_i}$ 
($e_{i-1}$ can be $e$ itself). From lemma \ref{FN3}, we know  
that because the step that is taken is \ruleref{C-CompareFail} we have\\
for faulty set $\FN$ satisfying $\blame_i$
$\recrafjudge{\Pi}{b_i}{\tau_i}$ where 
$b_i = \bigwedge_{f \in \FN}f$\\
From $e_i$ onwards any step taken is either 
going to propagate $\faila{}$ or is going to 
not touch the $\faila{}$ term at all.
For every evaluation step the invariant 
$\recrafjudge{\Pi}{b_i}{\tau_i}$ holds following
lemma \ref{FN4}. 
\end{enumerate}
\end{proof}